\documentclass[11pt]{article}
\usepackage{fancyhdr}
\usepackage{amsmath, amsfonts, amssymb,amsthm}
\usepackage{color,verbatim,graphicx, fullpage}
\usepackage{appendix}
\usepackage{url}
\usepackage{algorithm}
\usepackage{mathrsfs}
\usepackage{tikz}
\usetikzlibrary{decorations.pathreplacing}
\usepackage{tikz-3dplot}
\usepackage{amscd}
\usepackage[latin2]{inputenc}
\usepackage{t1enc}
\usepackage[mathscr]{eucal}
\usepackage{pict2e}
\usepackage[margin=2.9cm]{geometry}
\usepackage{epstopdf} 
\usetikzlibrary{decorations.markings}
\tikzstyle{vertex}=[circle, draw, inner sep=0pt, minimum size=6pt]
\newcommand{\vertex}{\node[vertex]}
 
\makeatletter 
 \newcommand\figcaption{\def\@captype{figure}\caption} 
  \newcommand\tabcaption{\def\@captype{table}\caption} 
\makeatother

\theoremstyle{plain}
\newtheorem{Th}{Theorem}[section]
\newtheorem{Lemma}[Th]{Lemma}

\newtheorem{Prop}[Th]{Proposition}

\theoremstyle{definition}
\newtheorem{Def}[Th]{Definition}

\newtheorem{Rem}[Th]{Remark}

\begin{document}

\title{Lattice Identification and Separation: Theory and Algorithm
\thanks{This research was supported by Simons Foundation grant 282311 and 584960.}}
\author{Yuchen He 
 \thanks{royarthur@gatech.edu,
School of Mathematics, Georgia Institute of Technology
686 Cherry Street, Atlanta, GA 30332-0160 USA.}
and 
Sung Ha Kang
\thanks{kang@math.gatech.edu,
http://people.math.gatech.edu/$\sim$kang/,
School of Mathematics, Georgia Institute of Technology
686 Cherry Street, Atlanta, GA 30332-0160 USA.}
}
\date{\vspace{-5ex}}
\maketitle

\begin{abstract}

Motivated by lattice mixture identification and grain boundary detection, we present a framework for lattice pattern representation and comparison, and propose an efficient algorithm for lattice separation.  
We define new scale and shape descriptors, which helps to reduce the size of equivalence classes of lattice bases considerably.  These finitely many equivalence relations are fully characterized by modular group theory.  We construct the lattice space $\mathscr{L}$ based on the equivalent descriptors and define a metric $d_{\mathscr{L}}$  to accurately quantify the visual similarities and differences between lattices.  
Furthermore, we introduce the Lattice Identification and Separation Algorithm (LISA), which identifies each lattice patterns from superposed lattices. LISA finds lattice candidates from the high responses in the image spectrum, then sequentially extracts different layers of lattice patterns one by one.   
Analyzing the frequency components, we reveal the intricate dependency of LISA's performances on particle radius, lattice density, and relative translations.  Various numerical experiments are designed to show LISA's robustness against a large number of lattice layers, moir\'{e} patterns and missing particles. 
\end{abstract}

\section{Introduction}

From material science to wallpaper pattern studies, there is a wide range of fruitful pattern research both in theory and applications.  Earlier studies~\cite{bieber,wallpaper,spacegroup} categorize patterns by symmetries, such as invariance under reflections or rotations.  Frieze and wallpaper groups are applied to identify periodic patterns in computer vision~\cite{regdom}. These pattern recognition typically involves two tasks: representation of regularities and automated classification~\cite{Bishop}, which are closely related.   

Motivated by some of the current developments in material sciences~\cite{jana2016two,voiry2015phase} and crystalline material image analysis~\cite{berkels2007identification,berkels2008extracting,boerdgen2010convex,hirvonen2018grain,lu2018phase,zosso2017two} , we focus on two dimensional lattice, which plays major roles in crystallography \cite{Xray,electdiff}, sampling theory \cite{samplingtheory}, ecology \cite{ecomodel} and many others.   For example, crystal structures of halite (NaCl) and gold (Au) have distinct scales (NaCl constant: $5.640$\AA~\cite{haynes2014crc}; Au constant: $4.065$\AA~\cite{davey1925precision}), which explains their proprietary differences.  
There are considerable research on detecting (non-superposed) lattice patterns from images, e.g., using the peaks of the Fourier power spectrum to identify the lattice structure \cite{FourierLattice}, and propagating an automatically suggested lattice pattern to the whole image by a tracking algorithm~\cite{MSBP}.  In \cite{geogroup}, the authors associated the wallpaper groups with local affine transformations to cluster repeated elements, and Hays et al.~\cite{highorder} propose the higher-order affinities among potential texels to discover visually consistent lattices.  

Our objective is to separate superposed lattices, which is a mixture of multiple two-dimensional lattices laid over another. This structure is referred to as a \textbf{superlattice}~\cite{superlattice}. Subjects characterized by superlattices are explored in solid physics~\cite{thermoelec,elecmob,hardness}, surface waves~\cite{surfacewave,Faraday} and nonlinear optics~\cite{nonlinearopt}. 
One of the most significant discoveries in low-dimensional material sciences is the family of transition metal dichalcogenides (TMDs)~\cite{choi2017recent,novoselov2004electric}, such as $\text{MoS}_{\text{2}}$~\cite{rao2014extraordinary} and $\text{WTe}_{\text{2}}$~\cite{eftekhari2017tungsten}. A single sheet of TMD shows a superlattice structure: the top and the bottom are chalcogen atoms layers, and the middle is a transition metal atoms layer. 

\begin{figure}
\begin{center}
\begin{tabular}{ccc}
(a) & (b) & (c) \\
\includegraphics[height = 1.5in]{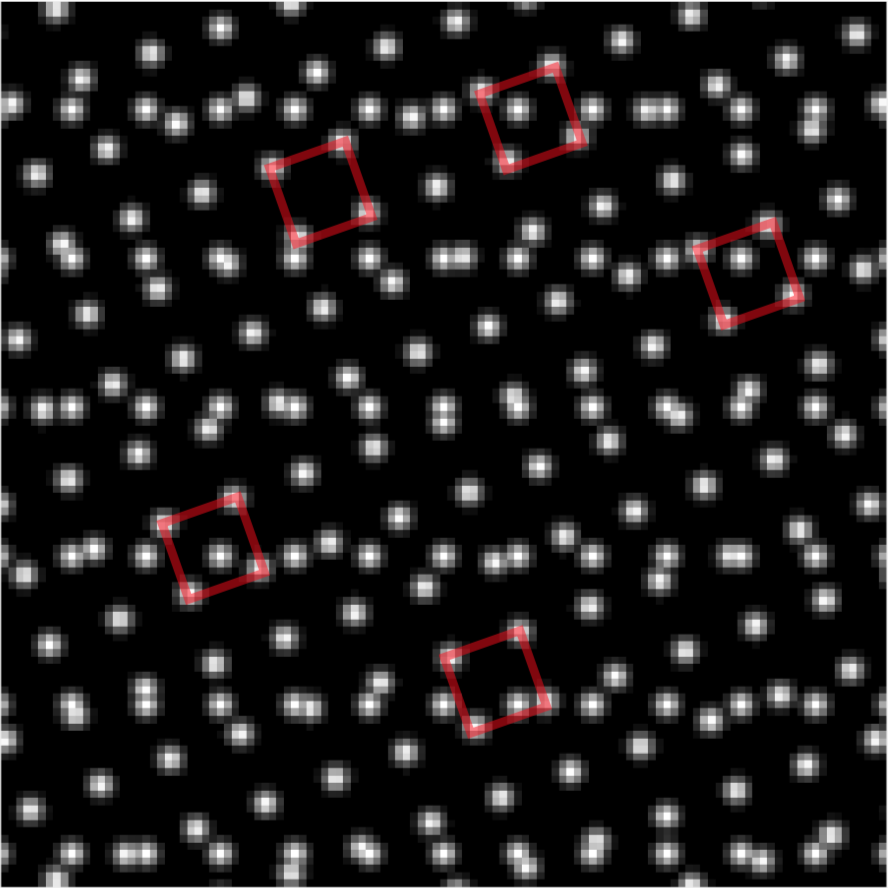} &
\includegraphics[height = 1.5in]{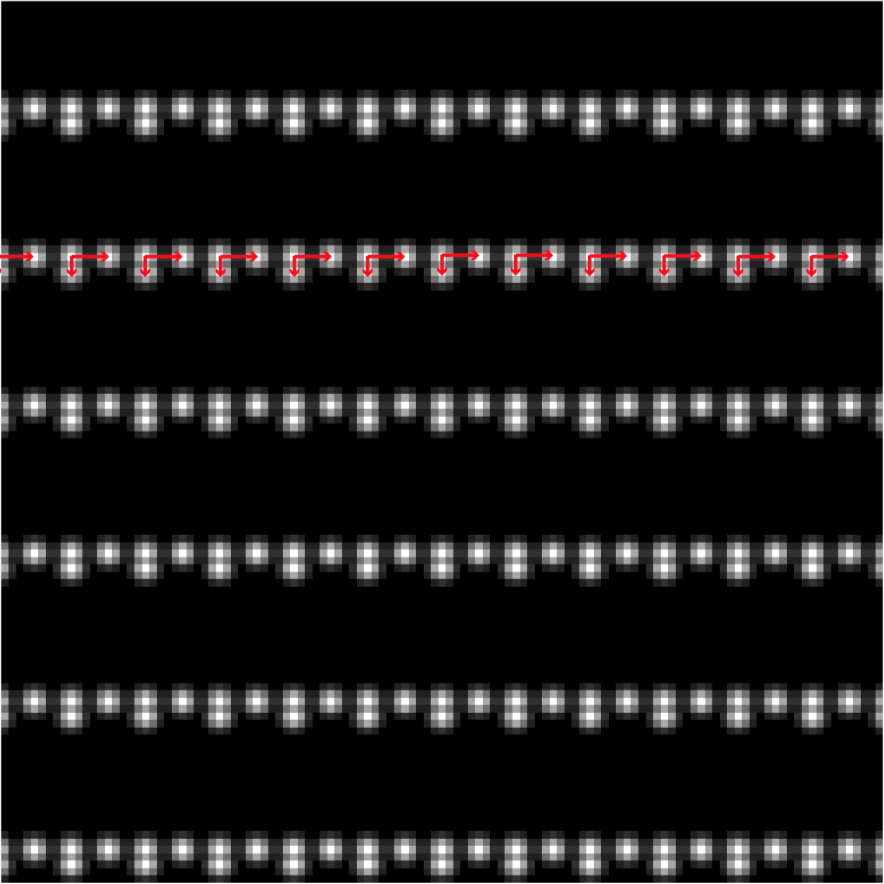} &
\includegraphics[height = 1.5in]{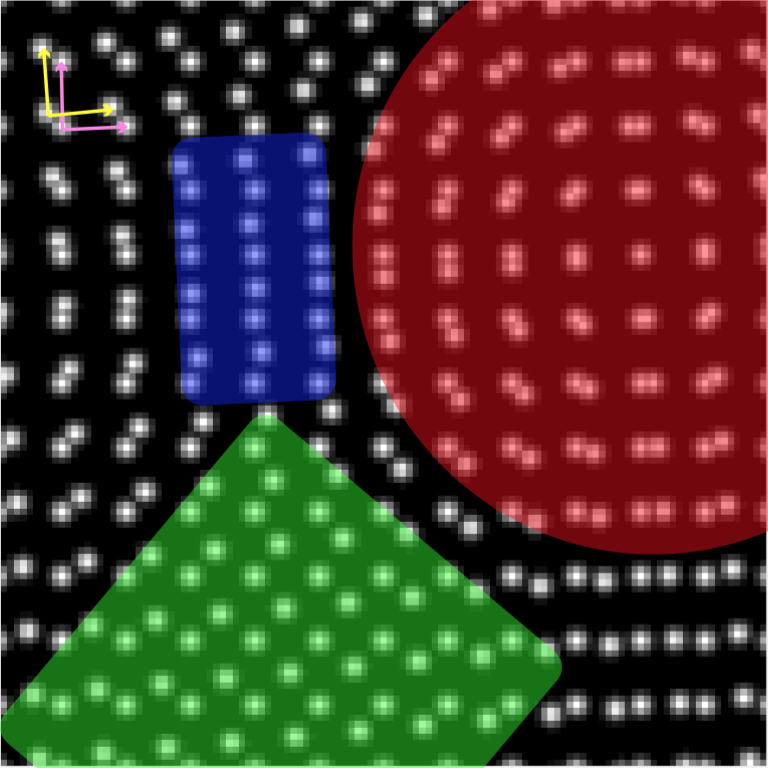} 
\end{tabular}
\end{center}
\caption{[Challenges of pattern separation] Each images have two lattices superposed.  
(a) Each red squares are units of one lattice, yet they have different interior patterns which can confuse the texton approach. 
(b) Using non-superposed lattice identification method, such as~\cite{MSBP}, wrong local feature L-shapes (the red arrows) can be identified.  These red arrows do not correspond to true underlying lattices. (c)  The pink and the yellow L-shapes on the top-left corner denote the true lattice components. There are three different types of Moir\'{e} patterns present (red, blue and green regions).  }\label{F:supchallenge}
\end{figure}

A clear definition of equivalent lattices is a cornerstone to classification.   In many context~\cite{latticespace2,latticespace3, latticespace1}, lattice is considered as an object who shares the structural characteristics with  $\mathbb{Z}^{n}$, $n\in\mathbb{N}_{\geq 1}$, thus all the lattices are equivalent. Focusing on symmetries, the theory of wallpaper groups~\cite{wallpaper} distinguishes five types of lattices: square, rectangular, hexagonal, rhombic and parallelogrammic. Many works in grain boundary detection~\cite{boerdgen2010convex} use the lattice orientation to indicate distinct patterns.  In this paper,  we define equivalent lattices to be identical lattices up to translation.  The scale, as well as rotational differences, are concerned.

Separating individual lattice pattern from a superlattice is challenging.  First, it is difficult to determine the smallest unit, e.g., the texton~\cite{texton}.  Effective methods for non-superposed lattice, such as~\cite{LeughMalik,MSBP} may fail, due to the interaction from different lattice layers.  Figure~\ref{F:supchallenge} (a) shows when the textons of one lattice have inhomogeneous interiors, and (b) shows when local L-shapes~\cite{MSBP} do not represent the correct underlying patterns. 
Secondly, superposed periodic patterns may produce new periodic structures, i.e., moir\'{e} patterns~\cite{moire}, which can confuse the identification process.  (This phenomenon is exploited in some applications~\cite{dislocation,imageresolution}.)  Figure~\ref{F:supchallenge} (c) shows three different moir\'{e} patterns generated by two lattices, whose bases are represented by pink and red L-shapes in the left top corner. 
Thirdly, human supervision~\cite{supervised} can be unreliable.  Psychological evidence ~\cite{density,texton,orientation,Wolfe89} prove that visual search can be interfered by similarities, e.g., less than $15^{\circ}$ of rotational differences between targets and background increases errors~\cite{orientation}, and small differences in densities can interfere target identification~\cite{texton}.

In this paper,  we first establish a framework to model and compare equivalent classes of lattices by constructing the \textbf{lattice space} $\mathscr{L}$ equipped with a new metric $d_{\mathscr{L}}$.  From the positive minimal bases~\cite{lattice}, we derive a new lattice representation using scale and shape descriptors on complex manifolds. Our lattice space consists of equivalent classes of descriptors which represents distinct lattice patterns up to translation. Building upon the Poincar\'{e} metric~\cite{PoinMetric}, a metric structure is then assigned to the lattice space.  

We propose a new Lattice Identification and Separation Algorithm (\textbf{LISA}). It sequentially extracts lattice patterns from a superlattice image without any prior knowledge of the number of layers.  The main idea behind LISA is to measure the periodicities globally by Fourier transform. For higher accuracy of estimating lattice bases, we exploit the Fourier Slice Theorem~\cite{FourierLattice}.   By evaluating pairs of peaks on the power spectrum, the optimal lattice structure is found.  We use a correcting step to obtain a stable estimation. The proposed method is designed to handle the moir\'{e} effect, excessive density, and inhomogeneous texton interior. 
We analytically study the properties of LISA. In particular, we reveal the effects of particle radius, lattice density and relative translations on LISA's performances, and show that LISA is robust against Gaussian perturbation with bounded variation.

Main contributions of this paper are:
\begin{enumerate}
\item{Introduction of the lattice descriptors and a lattice metric space, which gives a unifying representation for lattices and a tool to measure the visual differences and similarities between lattices.}
\item{Proposal for a new efficient lattice identification and separating algorithm. Our method does not require supervision, any prior knowledge of the lattices nor the number of layers.  We provide analytical studies on the construction, efficiency, and robustness of the algorithm. }
\end{enumerate}

This paper is organized as follows.  In Section \ref{sec:pre}, we present a typical notion of the lattice, state assumptions of the image, and review definitions as well as basic concepts.  In Section \ref{sec:lattice}, we introduce and study the descriptors by exploiting minimal bases and modular groups.    The lattice space and its natural metric structure are defined in Section \ref{sec:space}. 
We propose our algorithm LISA in Section \ref{sec:LISA} with analysis on properties of LISA starting in Section \ref{depLISA}.  Various numerical experiments are presented in Section \ref{sec:num}.   We conclude the paper with remarks in section \ref{sec:conc}, followed by Appendix including more discussion about sub-lattices and parent-lattices, and a pseudo-code computing lattice metric. 

\section{Preliminaries and Notations} \label{sec:pre}

A typical definition of lattice starts from two linearly independent vectors, $b_1$ and $b_2$ as basis.  A lattice is a set of linear combination of these basis with integer coefficients.  In two dimensional space, we utilize complex notation, $b_j= x_j+iy_j \in \mathbb{C}$, $x_j,y_j\in\mathbb{R}$, $j=1,2$, for simplicity. 

\begin{Def}[2D Lattice, Basis]\label{latdef} Given a pair of complex numbers $(b_{1},b_{2})\in\mathbb{C}^{2}$ satisfying $b_{1}\neq 0$ and $\text{Im}(b_{2}/b_{1})\neq0$, a 2D lattice determined by $(b_{1},b_{2})$ is defined as the set:
\begin{align*} \Lambda(b_{1},b_{2}) = \{k_{1}b_{1}+k_{2}b_{2}\mid k_{1},k_{2}\in \mathbb{Z}\},
\end{align*}	
and the pair $(b_{1},b_{2})$ is called a basis for $\Lambda(b_{1},b_{2})$.\end{Def}

The condition $\text{Im}(b_{2}/b_{1})\neq0$  represents two vectors $b_1$ and $b_2$ being linearly independent.  For any lattice $\Lambda(b_{1},b_{2})$,  by reordering or multiplying $-1$ if necessary,  we assume that $|b_{1}|\leq |b_{2}|$, and $\text{Im}(b_{2}/b_{1})>0$, i.e. the basis $(b_{1},b_{2})$ is positive. The key to distinguishing lattices depends on the equivalent bases.  
\begin{Def}[Equivalent Bases]\label{equivbases}
Let $(b_{1},b_{2})$ and $(b'_{1},b'_{2})\in\mathbb{C}^{2}$. If $\Lambda(b_{1},b_{2})=\Lambda(b_{1}',b_{2}')$, then $(b_{1},b_{2})$ and $(b'_{1},b'_{2})$ are called a pair of equivalent bases for $\Lambda(b_{1},b_{2})$.
\end{Def}
Given two bases $(b_{1},b_{2})$ and $(b_{1}',b_{2}')$,  they are equivalent if and only if the matrix
\begin{align*}
A=\begin{bmatrix}
	\text{Re}(b'_{1}) & \text{Im}(b'_{1})\\
	\text{Re}(b'_{2}) & \text{Im}(b'_{2})
\end{bmatrix}\begin{bmatrix}
	\text{Re}(b_{1}) & \text{Im}(b_{1})\\
	\text{Re}(b_{2}) & \text{Im}(b_{2})
\end{bmatrix}^{-1}
\end{align*}
has integer entries and the determinant is $\pm 1$.   The definition of lattice using linearly independent vectors is natural and  intuitive, yet, lacks a clear way to define equivalence classes nor has a simple measure for lattice comparison. 

Another important notion is \textbf{minimal basis}~\cite{latgeo}.  A basis $(b_{1},b_{2})$ is minimal if $\max(|b_{1}|,|b_{2}|)\leq|b_{1}\pm b_{2}|$.  Any pair of positive basis can be efficiently transformed to an equivalent minimal basis using the Positive Gauss reduction algorithm~\cite{lattice}. It takes a positive basis $(b_{1},b_{2})$ as the input. While $|b_{2}|<|b_{1}|$, repeat the following until stablization: $(b_{1},b_{2})=(b_{2},-b_{1})$, $q=\lfloor\text{Re}(b_{1}/b_{2})\rceil$, and $b_{2}= b_{2}-qb_{1}$. The output is a minimal basis.
Figure~\ref{F:basisdemo} demonstrates three different bases generating an identical lattice, and one can check that (a)~$(3,4i)$ has the shortest components among all the equivalent bases.

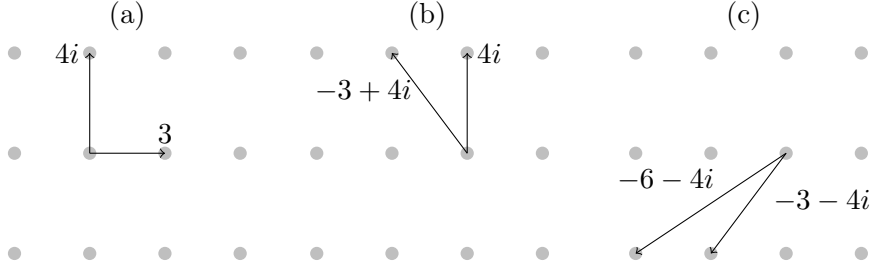
\begin{figure}
\begin{center}
\begin{tabular}{ccccc}
(a) & \hspace{0.5cm} & (b) &\hspace{0.5cm} & (c) \\
  \begin{tikzpicture}
  \node[circle, fill=lightgray, inner sep=0pt,minimum size=5pt] (b) at (0,0) {}; 
  \node[circle, fill=lightgray, inner sep=0pt,minimum size=5pt] (b) at (3/3,0) {};
  \node[circle, fill=lightgray, inner sep=0pt,minimum size=5pt] (b) at (0,4/3) {};
  \node[circle, fill=lightgray, inner sep=0pt,minimum size=5pt] (b) at (3/3,4/3) {};
  \node[circle, fill=lightgray, inner sep=0pt,minimum size=5pt] (b) at (6/3,4/3) {};
  \node[circle, fill=lightgray, inner sep=0pt,minimum size=5pt] (b) at (6/3,0) {};
  \node[circle, fill=lightgray, inner sep=0pt,minimum size=5pt] (b) at (6/3,-4/3) {};
  \node[circle, fill=lightgray, inner sep=0pt,minimum size=5pt] (b) at (-3/3,0) {};
  \node[circle, fill=lightgray, inner sep=0pt,minimum size=5pt] (b) at (-3/3,4/3) {};
  \node[circle, fill=lightgray, inner sep=0pt,minimum size=5pt] (b) at (0,-4/3) {};
  \node[circle, fill=lightgray, inner sep=0pt,minimum size=5pt] (b) at (3/3,-4/3) {};
  \node[circle, fill=lightgray, inner sep=0pt,minimum size=5pt] (b) at (-3/3,-4/3) {};
  \draw[->] (0,0)--(3/3,0) node[above, black] {$3$};
  \draw[->] (0,0)--(0,4/3) node[left, black] {$4i$};
    \end{tikzpicture}
    & &
  \begin{tikzpicture}
  \node[circle, fill=lightgray, inner sep=0pt,minimum size=5pt] (b) at (0,0) {}; 
  \node[circle, fill=lightgray, inner sep=0pt,minimum size=5pt] (b) at (3/3,0) {};
  \node[circle, fill=lightgray, inner sep=0pt,minimum size=5pt] (b) at (0,4/3) {};
  \node[circle, fill=lightgray, inner sep=0pt,minimum size=5pt] (b) at (3/3,4/3) {};
  \node[circle, fill=lightgray, inner sep=0pt,minimum size=5pt] (b) at (-6/3,4/3) {};
  \node[circle, fill=lightgray, inner sep=0pt,minimum size=5pt] (b) at (-6/3,0) {};
  \node[circle, fill=lightgray, inner sep=0pt,minimum size=5pt] (b) at (-6/3,-4/3) {};
  \node[circle, fill=lightgray, inner sep=0pt,minimum size=5pt] (b) at (-3/3,0) {};
  \node[circle, fill=lightgray, inner sep=0pt,minimum size=5pt] (b) at (-3/3,4/3) {};
  \node[circle, fill=lightgray, inner sep=0pt,minimum size=5pt] (b) at (0,-4/3) {};
  \node[circle, fill=lightgray, inner sep=0pt,minimum size=5pt] (b) at (3/3,-4/3) {};
  \node[circle, fill=lightgray, inner sep=0pt,minimum size=5pt] (b) at (-3/3,-4/3) {};
  \draw[->] (0,0)--(0,4/3) node[right, black] {$4i$};
  \draw[->] (0,0)--(-3/3,4/3) node[left,yshift=-0.5cm,xshift=0.4cm,black] {$-3+4i$};
    \end{tikzpicture}
& &
  \begin{tikzpicture}
  \node[circle, fill=lightgray, inner sep=0pt,minimum size=5pt] (b) at (0,0) {}; 
  \node[circle, fill=lightgray, inner sep=0pt,minimum size=5pt] (b) at (3/3,0) {};
  \node[circle, fill=lightgray, inner sep=0pt,minimum size=5pt] (b) at (0,4/3) {};
  \node[circle, fill=lightgray, inner sep=0pt,minimum size=5pt] (b) at (3/3,4/3) {};
  \node[circle, fill=lightgray, inner sep=0pt,minimum size=5pt] (b) at (-6/3,4/3) {};
  \node[circle, fill=lightgray, inner sep=0pt,minimum size=5pt] (b) at (-6/3,0) {};
  \node[circle, fill=lightgray, inner sep=0pt,minimum size=5pt] (b) at (-6/3,-4/3) {};
  \node[circle, fill=lightgray, inner sep=0pt,minimum size=5pt] (b) at (-3/3,0) {};
  \node[circle, fill=lightgray, inner sep=0pt,minimum size=5pt] (b) at (-3/3,4/3) {};
  \node[circle, fill=lightgray, inner sep=0pt,minimum size=5pt] (b) at (0,-4/3) {};
  \node[circle, fill=lightgray, inner sep=0pt,minimum size=5pt] (b) at (3/3,-4/3) {};
  \node[circle, fill=lightgray, inner sep=0pt,minimum size=5pt] (b) at (-3/3,-4/3) {};
  \draw[->] (0,0)--(-3/3,-4/3) node[right,xshift=0.7cm,yshift=0.7cm,black] {$-3-4i$};
  \draw[->] (0,0)--(-6/3,-4/3) node[above,xshift=0.4cm,yshift=0.7cm,black] {$-6-4i$};
      \end{tikzpicture}
    \end{tabular}
    \end{center}
\caption{[Equivalent lattice and minimal bases] (a) $\Lambda(3,4i)$, (b) $\Lambda(4i,-3+4i)$, and (c) $\Lambda(-3-4i,-6-4i)$ are all equivalent.  (a) is a minimal basis: $|\text{Re}(\frac{4i}{3})|=0<\frac{1}{2}$. (b) is not minimal: $|\text{Re}(\frac{-3+4i}{4i})|=1>\frac{1}{2}$, and (c) is not positive: $\text{Im}(\frac{-6-4i}{-3-4i})=-\frac{12}{25}<0$.}\label{F:basisdemo}
\end{figure}

We assume that the given image $U:\mathbb{R}^{2}\to[0,1]$ contains $N$ lattices $\{\mathcal{T}_{\mu_{j}}\Lambda_{j}:=\mathcal{T}_{\mu_{j}}\Lambda_{j}(b_{j,1},b_{j,2})\}_{j=1}^{N}$,
 \begin{align}
U(x,y) = \max_{j=1,\cdots, N}\mathcal{T}_{\mu_{j}}\Lambda_{j}(b_{j,1},b_{j,2})+R(x,y),\, (x,y)\in\mathbb{R}^{2}.\label{imagemodel}
 \end{align}
Here $\mathcal{T}_{\mu}\Lambda(b_{1},b_{2})$ denotes a lattice translated from $0$ by $\mu\in\mathbb{C}$,  and  $R$ is the residual term.  For visualization of lattice points, we put a point spread function (PSF) of Gaussian $G_{\sigma}$ with standard deviation $\sigma$ to each lattice point location \cite{PSF}, i.e.
\begin{align*}
\mathcal{T}_{\mu}\Lambda(b_{1},b_{2}) = \sum_{k_{1},k_{2}\in\mathbb{Z}}G_{\sigma}*\delta(k_{1}b_{1}+k_{2}b_{2}+\mu-x-iy), \;\;\; (x,y)\in\mathbb{R}^{2},
\end{align*}
where $\delta$ is the Dirac function on $\mathbb{C}$ defined by $\delta(x+iy)=1$ if $x+iy=0$, and $\delta(x+iy)=0$ otherwise. All the visible particles are assumed to be homogeneous, i.e., even if multiple lattice points overlapping at the same location, the height is bounded by 1.  This condition is ensured by the normalization in section~\ref{sec:LISA}.

To capture the periodicities of the lattice pattern, we utilize the Fourier and Radon transforms. In complex representation, arguments of the bases are important features and polar coordinate is more efficient in locating the peaks. For example, in Cartesian coordination, for a peak $(\xi,\nu)$ in the frequency domain, the argument estimation error $\Delta \theta$ at $\theta$ and the spatial discretization $\Delta \xi,\Delta \nu$ are related by $|\Delta\theta|\approx|\frac{\xi\Delta\nu-\nu\Delta\xi}{\xi^2+\nu^2}|$. To control $\Delta\theta$, the grid size must vary according to peak locations. We exploit the Fourier Slice Theorem~\cite{FourierSlice} to switch the coordinate system and use polar coordinates in this paper. 
\begin{Th}[Fourier Slice Theorem]\label{FourierSlice}
Consider a function $f:\mathbb{R}^{2}\to\mathbb{R}$, and denote~$\hat{}$~as the Fourier transform, then:
\begin{align*}
	\hat{f}(\gamma\cos\alpha,\gamma\sin\alpha)=\widehat{\mathcal{R}_{\alpha}[f]}(\gamma), \forall\gamma\in\mathbb{R}, \alpha\in[0,\pi),
\end{align*}
where $\mathcal{R}_{\alpha}[f](\gamma):=\mathcal{R}[f](\gamma,\alpha)$, and $\mathcal{R}[f]$ is the radon transform of $f$ defined by:
\begin{align*}
\mathcal{R}[f](\gamma,\alpha)	:=\int_{-\infty}^{+\infty}f(\gamma\cos\alpha-t\sin\alpha,\gamma\sin\alpha+t\cos\alpha)\,dt, \gamma\in \mathbb{R}, \alpha\in [0,\pi).
\end{align*}
\end{Th}

To construct a metric space, we review the following concepts~\cite{metricgeo} to be used in Section \ref{sec:space}.

\begin{Def}[Quotient pseudometric]
	Suppose $(X,D)$ is a metric space, and $\sim$ is an equivalence relation defined on $X$. Then the quotient pseudometric $\overline{D}$ for $X/\sim$ is defined as follows:
	\begin{align*}
		\overline{D}([x],[y])=\inf\{D(p_{1},q_{1})+\cdots +D(p_{n},q_{n})\},
	\end{align*}
	where $\inf$ is taken over all finite sequences $p_{1},\cdots, p_{n}$ and $q_{1},\cdots, q_{n}$ in $X$ such that $[p_{1}]=[x]$, $[q_{n}]=[y]$ and $[p_{i+1}]=[q_{i}]$, $i=1,2,...,n-1$. 
\end{Def}

The spaces in our work are Kolmogorov spaces, i.e., for every pair of distinct points, each has a neighborhood not containing the other. Hence all the quotient pseudometrics in this paper are in fact metrics.

\begin{Def}[Product Metric]
	Suppose $(X_{1},d_{1}), (X_{2},d_{2}),\cdots, (X_{n},d_{n})$ are metric spaces, and $D$ is an Euclidean norm on $\mathbb{R}^{n}$, then the product metric $D_{d_{1},\cdots,d_{n}}$ associated with $d_{1},\cdots, d_{n}$ for the space $X_{1}\times \cdots X_{n}$ is defined as:
	\begin{align*}
		D_{d_{1},\cdots,d_{n}}((x_{1},\cdots,x_{n}),(y_{1},\cdots,y_{n}))=D((d_{1}(x_{1},y_{1}),\cdots,d_{n}(x_{n},y_{n}))).
	\end{align*}
\end{Def}

\begin{Rem}
The formal definition of minimal basis that involves successive minima can be found in~\cite{latgeo}. The minimal basis is a special case of the reduced basis for a lattice in general dimension. Variations of this notion include well-known Minkowski-reduced~basis~\cite{Minred,Minred2}, generalized Gauss-reduced~basis~\cite{Hemred}, Hermite-Korkine-Zolotarev-reduced~basis~\cite{Cass,KZred}, and Lenstra-Lenstra-Lov\'{a}sz-reduced~basis~\cite{LLL}. They consider different relaxations, since finding the shortest vector using $L_{2}$-norm is NP-hard for randomized reductions~\cite{Aji}.
\end{Rem}

\begin{Rem}
Vall\'{e}e and Vera \cite{lattice} also include discussions about acute bases, which is the situation where $\text{Re}(b_{2}/b_{1})\geq0$. If $(b_{1},b_{2})$ is a positive basis, then the orientation is guaranteed, but it is not necessary that $b_{1}$ and $b_{2}$ have acute angle.  If the basis is acute, then it loses the orientability.  In this paper, we prioritize the orientability, thus focus on positive bases.	
\end{Rem}

\section{New Lattice Representation: descriptors $\beta$ and $\rho$}\label{sec:lattice}

We explore a new representation for a lattice using a pair of complex numbers $(\beta,\rho)\in\mathbb{C}^{2}$, which we call descriptors.  These are  derived from the positive minimal bases~\cite{lattice}, and  the key observation is that a lattice is realized as a transformed unit lattice.  
Transformations such as zoom-in, zoom-out and rotation are encoded in the scale descriptor $\beta$, and sheering, skew elongation or shrinking are controlled by the shape descriptor $\rho$. 
One of the advantages of descriptors is that, compared to the Definition~\ref{latdef}, the number of equivalent representations is dramatically reduced from infinite to only a few. These equivalence relations can be fully characterized by exploiting the modular group theory~\cite{modform}.  Descriptors modulo these relations are used as elements for the lattice space in Section~\ref{sec:space}.

\begin{Def}[Scale and Shape Descriptor] \label{D:betarho}
 	Given a lattice $\Lambda(b_{1},b_{2})$ where   $(b_1, b_2)$ is a minimal basis, we define:
 	\begin{align*}
 	&\text{Scale descriptor:} ~\beta=b_{1};\\
 	&\text{Shape descriptor:} ~\rho=b_{2}/b_{1}.	
 	\end{align*}
We denote $\Lambda\langle \beta,\rho\rangle$ to be a  lattice, which is spanned by $\beta$ and $\beta\rho$, i.e. $\Lambda\langle \beta,\rho\rangle = \Lambda(\beta,\beta\rho)$. 
 \end{Def}
 
 \begin{figure}
\begin{center}
\begin{tabular}{ccccc}
(a) &\hspace{0.5cm} & (b)  &\hspace{0.5cm} & (c) \\
  \begin{tikzpicture}
  \node[circle, fill=gray, inner sep=0pt,minimum size=3pt] (b) at (0,-0.5) {}; 
  \node[circle, fill=gray, inner sep=0pt,minimum size=3pt] (b) at (-0.5,0) {};
  \node[circle, fill=gray, inner sep=0pt,minimum size=3pt] (b) at (-0.5,0.5) {};  
  \node[circle, fill=gray, inner sep=0pt,minimum size=3pt] (b) at (-0.5,1) {};
  \node[circle, fill=gray, inner sep=0pt,minimum size=3pt] (b) at (0.5,-0.5) {};
  \node[circle, fill=gray, inner sep=0pt,minimum size=3pt] (b) at (-0.5,-0.5) {};
  \node[circle, fill=gray, inner sep=0pt,minimum size=3pt] (b) at (1,-0.5) {};
  \node[circle, fill=gray, inner sep=0pt,minimum size=3pt] (b) at (0,0) {}; 
  \node[circle, fill=gray, inner sep=0pt,minimum size=3pt] (b) at (0.5,0) {};
  \node[circle, fill=gray, inner sep=0pt,minimum size=3pt] (b) at (0,0.5) {};
  \node[circle, fill=gray, inner sep=0pt,minimum size=3pt] (b) at (0.5,0.5) {};
  \node[circle, fill=gray, inner sep=0pt,minimum size=3pt] (b) at (1,0) {};
  \node[circle, fill=gray, inner sep=0pt,minimum size=3pt] (b) at (0,1) {};
  \node[circle, fill=gray, inner sep=0pt,minimum size=3pt] (b) at (1,0.5) {};
  \node[circle, fill=gray, inner sep=0pt,minimum size=3pt] (b) at (1,1) {};
  \node[circle, fill=gray, inner sep=0pt,minimum size=3pt] (b) at (0.5,1) {};
  \node[circle, fill=gray, inner sep=0pt,minimum size=3pt] (b) at (1.5,1) {};
  \node[circle, fill=gray, inner sep=0pt,minimum size=3pt] (b) at (1.5,0.5) {};
  \node[circle, fill=gray, inner sep=0pt,minimum size=3pt] (b) at (1.5,0) {};
  \node[circle, fill=gray, inner sep=0pt,minimum size=3pt] (b) at (1.5,-0.5) {};
  \node[circle, fill=gray, inner sep=0pt,minimum size=3pt] (b) at (0,1.5) {};
  \node[circle, fill=gray, inner sep=0pt,minimum size=3pt] (b) at (-0.5,1.5) {};
  \node[circle, fill=gray, inner sep=0pt,minimum size=3pt] (b) at (0.5,1.5) {};
  \node[circle, fill=gray, inner sep=0pt,minimum size=3pt] (b) at (1,1.5) {};
  \node[circle, fill=gray, inner sep=0pt,minimum size=3pt] (b) at (1.5,1.5) {};
    \end{tikzpicture}
& &
  \begin{tikzpicture}
  \node[circle, fill=gray, inner sep=0pt,minimum size=3pt] (b) at (-0.5,0.5) {};  
  \node[circle, fill=gray, inner sep=0pt,minimum size=3pt] (b) at (0.5,-0.5) {};
  \node[circle, fill=gray, inner sep=0pt,minimum size=3pt] (b) at (-0.5,-0.5) {};
  \node[circle, fill=gray, inner sep=0pt,minimum size=3pt] (b) at (0.5,0.5) {};
  \node[circle, fill=gray, inner sep=0pt,minimum size=3pt] (b) at (1.5,0.5) {};
  \node[circle, fill=gray, inner sep=0pt,minimum size=3pt] (b) at (1.5,-0.5) {};
  \node[circle, fill=gray, inner sep=0pt,minimum size=3pt] (b) at (-0.5,1.5) {};
  \node[circle, fill=gray, inner sep=0pt,minimum size=3pt] (b) at (0.5,1.5) {};
  \node[circle, fill=gray, inner sep=0pt,minimum size=3pt] (b) at (1.5,1.5) {};
    \end{tikzpicture}
& &
  \begin{tikzpicture}
  \foreach \x in {-2,-1,0,1,2} {
  	\foreach \y in {-2,-1,0,1,2}{
            \node[circle, fill=gray, inner sep=0pt,minimum size=3pt] (b) at (1.7/4*\x-1/4*\y,1/4*\x+1.7/4*\y) {};}}
    \end{tikzpicture}
\\ & & &\\
(d)  &\hspace{0.5cm} & (e)  &\hspace{0.5cm} & (f) \\
  \begin{tikzpicture}
  \foreach \x in {-2,-1,0,1,2} {
  	\foreach \y in {-1,0,1}{
            \node[circle, fill=gray, inner sep=0pt,minimum size=3pt] (b) at (0.5*\x,\y) {};
            }
        }  
    \end{tikzpicture}
& &
  \begin{tikzpicture}
  \foreach \x in {-2,-1,0,1,2} {
  	\foreach \y in {-2,-1,0,1,2}{
            \node[circle, fill=gray, inner sep=0pt,minimum size=3pt] (b) at (0.5*\x-1/4*\y,1.7/4*\y) {};
            }
        }  
    \end{tikzpicture}
& &
  \begin{tikzpicture}
  \foreach \x in {-2,0,2} {
  	\foreach \y in {-2,0,2}{
            \node[circle, fill=gray, inner sep=0pt,minimum size=3pt] (b) at (0.5*\x-1/4*\y,1.7/4*\y) {};
            }
        }  
    \end{tikzpicture}
\end{tabular}
\end{center}
\caption{[Descriptors $\beta$ and $\rho$]  (a) $\Lambda\langle 1,i\rangle $, (b) $\Lambda\langle 2,i\rangle$, (c) $\Lambda\langle e^{i\pi/6},i\rangle$, (d) $\Lambda\langle 1,2i\rangle$, (e) $\Lambda\langle 1,e^{2\pi i/3}\rangle$, and (f) $\Lambda\langle 2,e^{2\pi i/3}\rangle$.  From (a) to (b), only $\beta$ changed from 1 to 2.  From (a) to (c), $\beta$ is rotated.  From (a) to (d), $\rho$ changed from $i$ to $2i$.  From (a) to (e), $\rho$ is rotated,. From (a) to (f), both $\beta$ changed and $\rho$ rotated.}\label{F:descriptors}
\end{figure}
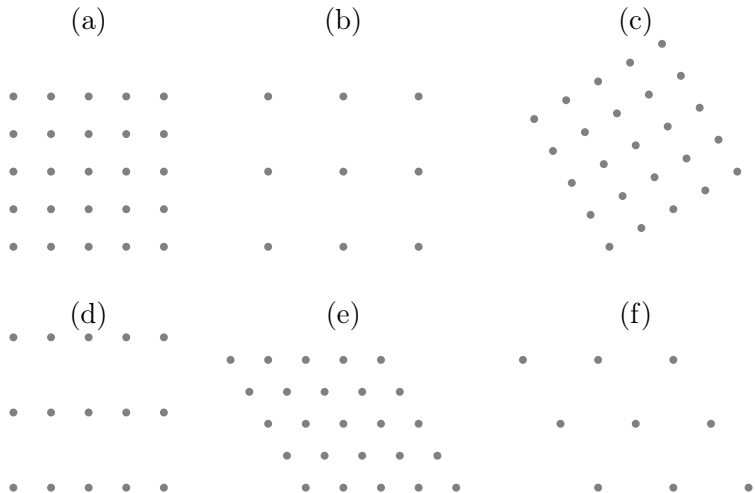

Figure~\ref{F:descriptors} illustrates various effects of changing $\beta$ and $\rho$. From (a) to (b), only  $\beta$ changed from 1 to 2, and from (a) to (c), $\beta$ is rotated.   From (a) to (d)  $\rho$ changed from $i$ to $2i$, and  from (a) to (e), $\rho$ is rotated.  From (a) to (f) both $\beta$ changed and  $\rho$ rotated. Varying the scale descriptor $\beta$ corresponds to zooming and rotations, while the shape descriptor $\rho$ corresponds to  sheering, skew elongation and skew shrinking. More importantly, equivalent bases determine a simple algebraic relation between their descriptors.

\begin{Prop}[Necessary condition]\label{necprop}
If two lattices $\Lambda\langle \beta,\rho\rangle$ and $\Lambda\langle \beta',\rho'\rangle$ are equivalent, then there exists $k_{i}\in\mathbb{Z}$, $i=1,2,3,4$ with $k_{1}k_{4}-k_{2}k_{3}= 1$, such that the following hold:
\begin{align}
&\beta'=e^{i\text{Arg}(k_{1}+k_{2}\rho)}\beta\label{scalecon},~\text{and}\\
&\rho'=(k_{3}+k_{4}\rho)/(k_{1}+k_{2}\rho)\label{shapecon}.
\end{align}
\end{Prop}
\begin{proof}
Note that $\Lambda\langle\beta,\rho\rangle=\Lambda\langle\beta',\rho'\rangle$ if and only if there is a unimodular matrix $U=\begin{bmatrix}
k_{1}&k_{2}\\
k_{3}&k_{4}	
\end{bmatrix}
$, $k_{i}\in \mathbb{Z}$, $i=1,2,3,4$, such that $U\begin{bmatrix}
b_{1}\\b_{2}	
\end{bmatrix}=\begin{bmatrix}
b'_{1}\\b'_{2}	
\end{bmatrix}
$, where $b_{1}=\beta$, $b_{2}=\beta\rho$, $b_{1}'=\beta'$ and $b_{2}'=\beta'\rho'$ are the associated bases respectively.
From the matrix multiplication, (\ref{shapecon}) follows immediately.  Because the bases are minimal, $|b_{1}|=|b_{1}'|$ implies $b_{1}'=e^{i\theta}b_1$ for some $\theta\in[0,2\pi]$. Combining this with $b_{1}'=k_1b_1+k_2b_2$ gives $k_1+k_2\rho = e^{i\theta}$, thus $\theta =\text{Arg}(k_1+k_2\rho)$ and (\ref{scalecon}) follows. In addition, since $(b_{1},b_{2})$ and $(b_{1}',b_{2}')$ are positive, $\text{det}\,U=k_1k_4-k_2k_3=1$.
\end{proof}

In the following, we apply the modular group theory to prove the converse of Proposition \ref{necprop}, hence whether two descriptors generate an identical lattice can be easily determined. As a preparation, we state a lemma.
\begin{Lemma}\label{lemmaeqv}
The converse of Proposition \ref{necprop} holds if 
\begin{align*}
|k_1+k_2\rho|=1.
\end{align*}
\end{Lemma}
\begin{proof}
Denote $c = \frac{1}{|k_1+k_2\rho|}=\frac{e^{i\text{Arg}(k_1+k_2\rho)}}{k_1+k_2\rho}$, then from (\ref{scalecon}), we have $b_1'=\beta'=c(k_1+k_2\rho)\beta=c(k_1b_1+k_2b_2)$. Notice that (\ref{shapecon}) reads $b_{2}'=b_1'\frac{k_3+k_4\rho}{k_1+k_2\rho}=c(k_1+k_2\rho)\beta\frac{k_3+k_4\rho}{k_1+k_2\rho}=c(k_3b_1+k_4b_2)$. The lemma is thus proved.
\end{proof}

Since the equivalence condition~(\ref{scalecon}) shows the dependency of $\beta$ on $\rho$, we start with the details of shape descriptor $\rho$ in (\ref{shapecon}). 

\subsection{Equivalence class of shape descriptor $\rho$}\label{shapeDes}

The definition \ref{D:betarho} of the new descriptors starts from the minimal basis notation (see Section \ref{sec:pre}).  Using basic geometry, it is straightforward to show that the definition of the minimal basis
 is equivalent to the ratio $\rho=b_{2}/b_{1}$ belonging to the following region: 
\begin{align}
\mathcal{P}:=\{z\in\mathbb{C}\mid |z|\geq 1, |\text{Re}(z)|\leq\frac{1}{2}, \text{Im}(z)>0\}\subset\mathbb{C}.\label{Pregion}
\end{align}  

\begin{figure}
\begin{center}
\begin{tabular}{ccc}
(a)  &\hspace{0.5cm} & (b) \\
  \begin{tikzpicture}
   \draw[lightgray] (0,0) circle (2cm);
   \draw[lightgray] (-1,1.7321)--(-1,4);  
   \draw[lightgray] (1,1.7321)--(1,4); 
   \draw[draw=none,samples=35,fill=lightgray]
     plot [domain=0:1] (\x,{sqrt(4-\x^2)})
  -- plot [domain=1.7321:4] (1,\x)
  -- plot [domain=1:-1] (\x,4)
  -- plot [domain=4:1.7321] (-1,\x)
  -- plot [domain=1:0] (-\x,{sqrt(4-\x^2)})
  -- cycle; 
  \draw[gray,->] (-3,0)--(3,0) node[right] {$\text{Re}$};
   \draw[gray,->] (0,-2.5)--(0,4) node[left,yshift=-8]{$\text{Im}$};  
   \draw (2,2.5) coordinate node[right] {$\mathbb{C}$};
   \draw [decorate,decoration={brace},yshift=1.5pt]
(0,0) -- (2,0) node [black,midway,above] 
{\footnotesize $1$};
\draw [decorate,decoration={brace,mirror},yshift=-1.5pt]
(0,1.73) -- (1,1.73) node [black,midway,below] 
{\footnotesize $1/2$};
   \draw[->,thick] (0,0)--(0,3.5) node[midway,yshift=0.8cm,right] {$\rho$};
   \draw[->,dashed] (0,0)--(-2,3/2) node[midway,left,yshift=-2] {$\rho'$};
   \draw[->,dashed] (0,0)--(68/25,-24/25) node[midway,right,xshift=1] {$\rho''$};
   \draw[thick] (-1,1.7321)--(-1,4);
   \draw[thick] (1,1.7321)--(1,4);
    \draw [thick,domain=60:120] plot ({2*cos(\x)}, {2*sin(\x)});
    \end{tikzpicture}
& &
  \begin{tikzpicture}
   \draw[lightgray] (0,0) circle (2cm);
   \draw[lightgray] (-1,1.7321)--(-1,4);  
   \draw[lightgray] (1,1.7321)--(1,4); 
   \draw[draw=none,samples=35,fill=lightgray]
     plot [domain=0:1] (\x,{sqrt(4-\x^2)})
  -- plot [domain=1.7321:4] (1,\x)
  -- plot [domain=1:-1] (\x,4)
  -- plot [domain=4:1.7321] (-1,\x)
  -- plot [domain=1:0] (-\x,{sqrt(4-\x^2)})
  -- cycle; 
  \draw[gray,->] (-3,0)--(3,0) node[right] {$\text{Re}$};
   \draw[gray,->] (0,-2.5)--(0,4) node[left,yshift=-8]{$\text{Im}$};  
   \draw (2,2.5) coordinate node[right] {$\mathbb{C}$};
   \draw[->] (0,0)--(1,2.5) node[midway,right] {$\rho+1$};
   \draw[->,thick] (0,0)--(-1,2.5) node[midway,right] {$\rho$};
   \draw[thick] (-1,1.7321)--(-1,4);
    \draw [thick,domain=90:120] plot ({2*cos(\x)}, {2*sin(\x)});
    \end{tikzpicture} 
\end{tabular}
\end{center} 
\caption{The shaded region in (a) with the boundary is $\mathcal{P}$. Vectors $\rho$, $\rho'$ and $\rho''$ are the shape descriptors for the bases $\Lambda(3,4i)$, $\Lambda(4i,-3+4i)$, and $\Lambda(-3-4i,-6-4i)$  in Figure~\ref{F:basisdemo} (a)-(c) in order.  All represents the same lattice, while  $\rho$ for  $\Lambda(3,4i)$ is a minimal basis.  
(b) A fundamental set of the modular group $\Gamma$ acting on the upper half plane $\mathcal{H}$. If $\text{Re}(\rho)=-1/2$, $\rho$ and $\rho+1$ are in the same orbit.}\label{F:Pdemo}
\end{figure}
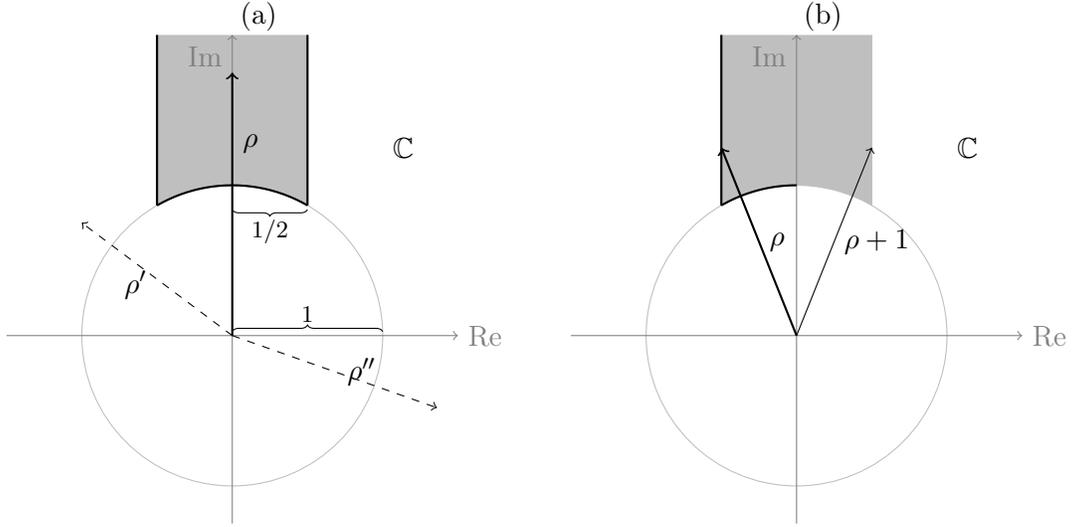

Figure \ref{F:Pdemo} shows $\mathcal{P}$ as the shaded region in (a) with the boundary. As an example, we consider vectors $\rho$, $\rho'$ and $\rho''$ which are the shape descriptors for the bases $\Lambda(3,4i)$, $\Lambda(4i,-3+4i)$, and $\Lambda(-3-4i,-6-4i)$  in Figure~\ref{F:basisdemo} (a)-(c) in order.  These represent the same lattice, while $\rho$ from  $\Lambda(3,4i)$ is  a minimal basis.

The equivalence condition~(\ref{shapecon})  can be viewed as a transformation defined on the upper-half plane $\mathcal{H}=\{z\mid \text{Im}(z)>0\}$ restricted to $\mathcal{P}$. It is expressed as:
\begin{align*}
	z\mapsto\frac{k_{3}+k_{4}z}{k_{1}+k_{2}z},~\{k_{i}\}_{i=1}^{4}\subset\mathbb{Z}, \text{such that~}k_{1}k_{4}-k_{2}k_{3}=1,\forall z\in\mathcal{H},
\end{align*}
Each function of this form is a special case of M\"{o}bius transforms, and the set of these transformations under function composition gives the well-known modular group~\cite{modform}, denoted by $\Gamma$. The elements in $\Gamma$ act on the upper half plane $\mathcal{H}$ naturally.  Within the context of group actions, the equivalence condition (\ref{shapecon}) indicates that equivalent $\rho$ live in the same orbit of $\Gamma$-actions.

The invocation of the modular group $\Gamma$ also reveals the significance of the region $\mathcal{P}$ defined in~(\ref{Pregion}). This $\mathcal{P}$ minus half of its boundary:
\begin{align*}
	 \mathcal{P}\setminus\big(\{z\in\mathcal{H}\mid \text{Re}(z)=\frac{1}{2}\}\cup \{z\in\mathcal{H}\mid 0<\text{Re}(z)<\frac{1}{2},|z|=1\}\big),
\end{align*}
is a fundamental set for $\Gamma$-actions~\cite{modform}.   See Figure \ref{F:Pdemo} (b).  Every element in the fundamental set is a representative of one and only one orbit, and every orbit corresponds to a unique representative. No two shape descriptors in the difference set are equivalent. This result provides a key insight that equivalent shape descriptors only occur on the boundary of $\mathcal{P}$.

Following the approach of Alperin on the modular group~\cite{mod}, we can enumerate all the classes of equivalent shape descriptors systematically.  Any $\Gamma$-action is a composition of a finite sequence of two basic transformations: translation $T$ and inversion followed by reflection $S$. For example, if we denote:
\begin{align*}
T: z\mapsto z+1,~\text{and}~S: z\mapsto -1/z,~\text{for any~}z\in\mathcal{H}	,
\end{align*}
then any element in $\Gamma$ can be written as $S^{k_{1}}T^{l_{1}}S^{k_{2}}T^{l_{2}}\cdots S^{k_{m}}T^{l_{m}}$ for some $k_{j}\in\{0,1\}$, $l_{j}\in\mathbb{Z}$, and $j=1,2,\cdots,m$, where $m\in\mathbb{N}$. Focusing on the sequences of $S$ and $T$ acting on $\mathcal{P}$, we arrive at a full characterization of equivalence classes of shape descriptors.  
\begin{Prop}\label{useful}Given a shape descriptor $\rho\in\mathcal{P}$, we list all the shape descriptors equivalent to it, based on the location of $\rho$ in $\mathcal{P}$ as follows:
\begin{center}  
\begin{table}[H]
\centering
    \begin{tabular}{|c|c|} \hline 
\textbf{Location of $\rho$} & \textbf{All the equivalent shape descriptors}\\\hline
 $\{z\in\mathcal{P}\mid |z|>1,|\text{Re}(z)|<1/2\}$&$\rho$\\\hline
$\{z\in\mathcal{P}\mid \text{Re}(z)=-1/2,|z|>1\}$&$\rho$, $T\rho$\\\hline
$\{z\in\mathcal{P}\mid \text{Re}(z)=1/2,|z|>1\}$&$\rho$, $T^{-1}\rho$\\\hline
$\{z\in\mathcal{P}\mid |z|=1, 0\leq|\text{Re}(z)|<1/2\}$&$\rho$, $S\rho$\\\hline
$e^{i2\pi/3}$&$\rho$, $S\rho$, $T\rho$, $T^{-1}S\rho$, $ST\rho$, $TST\rho$\\\hline
$e^{i\pi/3}$&$\rho$, $S\rho$, $T^{-1}\rho$, $TS\rho$, $ST^{-1}\rho$, $STS\rho$\\\hline
\end{tabular}   
\end{table}
\end{center}
\end{Prop}
Geometrically, the small sizes of equivalence classes come from the restriction, that both $\rho$ and $\rho'$ belong to $\mathcal{P}$. 
In effect, the fundamental principle lurking behind the reduction is the uniqueness of successive minima of a finite dimensional lattice. This requires that the transformations relating two shape descriptors must preserve norm, and they form a proper subset of the modular group.  

\begin{Rem}
Relating to wallpaper groups \cite{wallpaper}, the notion of shape descriptor $\rho$ is compatible with the 5 classes of lattices.  For a lattice $\Lambda\langle \beta,\rho\rangle$, if $\rho=\pm\frac{1}{2}+i\frac{\sqrt{3}}{2}$, then it is hexagonal; if $\rho = i$, then it is square; if $\text{Re}(\rho)=0$, then it is rectangular; if $|\text{Re}(\rho)|=\frac{1}{2}$ or $|\rho|=1$, then it is rhombic; otherwise, it is parallelogrammic. The shape descriptor $\rho$ recognizes finer differences, and with the scale descriptor $\beta$, they represents all lattice patterns up to translation.
\end{Rem}

\subsection{Equivalence class of scale descriptor $\beta$}\label{scaleDes}
The condition~(\ref{scalecon}) shows the dependency of equivalence relations of scale descriptors $\beta$ on that of shape descriptors $\rho$. The choice of $\Gamma$-action that achieves equivalence relation between $\rho$ and $\rho'$ restricts the angles between $\beta$ and its equivalent elements.  Every $\Gamma$ action is associated with a matrix $\begin{bmatrix}
k_{4}&k_{3}\\
k_{2}&k_{1}	
\end{bmatrix}
$, whose entries in the first row are the coefficients in the numerator in~(\ref{shapecon}), and those in the second row the coefficients in the denominator. Corresponding to the nontrivial actions in Proposition~\ref{useful}, the matrix representations are:
\begin{align*}
&T=\begin{bmatrix}
1&1\\
0&1	
\end{bmatrix}, T^{-1}=\begin{bmatrix}
1&-1\\
0&1	
\end{bmatrix}, S=\begin{bmatrix}
0&-1\\
1&0	
\end{bmatrix},
T^{-1}S=\begin{bmatrix}
-1&-1\\
1&0	
\end{bmatrix}, TS=\begin{bmatrix}
1&-1\\
1&0	
\end{bmatrix},\nonumber\\
&TST=\begin{bmatrix}
1&0\\
1&1	
\end{bmatrix}, ST=\begin{bmatrix}
0&-1\\
1&1	
\end{bmatrix},  ST^{-1}=\begin{bmatrix}
0&-1\\
1&-1	
\end{bmatrix}, STS=\begin{bmatrix}
-1&0\\
1&-1	
\end{bmatrix}.
\end{align*}	
This entire list of possible $\Gamma$-actions that relate equivalent shape descriptors contains critical information. First, observe that for any $\rho\in\mathcal{P}$, the corresponding $\Gamma$-actions in Proposition \ref{useful} always satisfy $|k_1+k_2\rho|=1$, where the action is expressed as a matrix $\begin{bmatrix}
k_{4}&k_{3}\\
k_{2}&k_{1}	
\end{bmatrix}$. Therefore, combining Proposition \ref{useful} and Lemma \ref{lemmaeqv} yields our fundamental result.
\begin{Th}
[Equivalent descriptors]\label{fundtheorem}
Two lattices $\Lambda\langle \beta,\rho\rangle$ and $\Lambda\langle \beta',\rho'\rangle$ are equivalent \textbf{if and only if} there exists $k_{i}\in\mathbb{Z}$, $i=1,2,3,4$ with $k_{1}k_{4}-k_{2}k_{3}= 1$, such that the following hold:
\begin{align*}
&\beta'=e^{i\text{Arg}(k_{1}+k_{2}\rho)}\beta,~\text{and}\\
&\rho'=(k_{3}+k_{4}\rho)/(k_{1}+k_{2}\rho).
\end{align*}
\end{Th}Second, this list allows us to summarize all the variants of~(\ref{scalecon}).  
\begin{Prop}
Given a scale descriptor $\beta\in\mathbb{C}\setminus\{0\}$ and two shape descriptors $\rho,\rho'\in\mathcal{P}$. If $\rho$ and $\rho'$ are equivalent using the $\Gamma$-actions in the left column of the following table, then all the scale descriptors that satisfy the scale condition with $\beta$ are listed in the right column correspondingly. 
\begin{center}
\begin{table}[H]
\centering
    \begin{tabular}{|c|c|} \hline 
\textbf{$\Gamma$-actions} & \textbf{Scale condition satisfied with}\\\hline
$I$, $T$ ,$T^{-1}$&$\pm\beta$\\\hline
$S$, $T^{-1}S$, $TS$&$\exp(\pm i\text{Arg}(\rho))\beta$\\\hline
$TST$ ,$ST$& $\exp(\pm i\text{Arg}(1+\rho))\beta$\\\hline
$ST^{-1}$ ,$STS$& $\exp(\pm i\text{Arg}(1-\rho))\beta$\\\hline
\end{tabular}    
\end{table}	
\end{center}
\end{Prop} 
Using the matrix representation, we can identify the group of M\"{o}bius transforms with the projective general linear group $\text{PGL}_{2}(\mathbb{C})$, and the modular group with the projective special linear group $\text{PSL}_{2}(\mathbb{Z})$. In this sense, Subsection~\ref{shapeDes} and~\ref{scaleDes} establish the connection between equivalent lattices with the subgroup $\text{PSL}_{2}(\mathbb{Z})\leq \text{PGL}_{2}(\mathbb{C})$. In the appendix, we extend similar algebraic correspondence to link sub-lattices with monoids and find that it is intrinsically hard to search for a particular sub-lattice or a parent-lattice of a given lattice systematically.

\section{New Definition of Lattice Space and Metric}\label{sec:space}

Using the descriptors, we present the lattice space $\mathscr{L}$ equipped with a metric $d_\mathscr{L}$. 
The equivalence relations allow every lattice pattern be uniquely represented by a point in this space $\mathscr{L}$.  
\begin{Def}[Lattice Space]\label{latfo}
Let $\mathcal{P}$ be the set of shape descriptors  $\rho$ (\ref{Pregion}), and $\mathcal{K}:=\mathbb{C}\setminus\{0\}$ be the set of scale descriptors $\beta$. The lattice space $\mathscr{L}$ is defined as follows:
\begin{align}
	\mathscr{L}=\big(\mathcal{K}/\sim_{1}\times\mathcal{P}/\sim_{2}\big)/\sim_{3},\label{latticedef}
\end{align}
where the three equivalence relations are:
\begin{enumerate}
\item{$\beta\sim_{1}-\beta,~\forall \beta\in \mathcal{K}$,  i.e.,
	$\Lambda\langle \beta,\rho\rangle = \Lambda\langle -\beta,\rho\rangle$}
\item{$\rho\sim_{2}\rho'$, $\forall \rho,\rho'\in\mathcal{P}$, i.e.,
$\Lambda\langle \beta,\rho\rangle = \Lambda\langle \beta,\rho'\rangle$, 
for  $\text{Im}(\rho)=\text{Im}(\rho')$ and $|\text{Re}(\rho)|=|\text{Re}(\rho')|=1/2$, }
\item{ $\langle [\beta]_{1},[\rho]_{2} \rangle \sim_{3} \langle [\beta \rho]_{1},[-1/\rho]_{2} \rangle$, $\forall \beta\in \mathcal{K}$, $\forall \rho\in\mathcal{P}$ ,i.e.,  $\Lambda\langle \beta,\rho\rangle = \Lambda \langle \beta\rho,-1/\rho\rangle$,  for  $|\rho|=1$,}
\end{enumerate}
and $\mathscr{L}$ has the induced topology.  We denote $[\beta,\rho]$ as an element in $\mathscr{L}$ considering  the equivalence relations.
\end{Def}

The first equivalence relation, $\mathcal{K}/\sim_{1}$,  consists of scale descriptors $\beta$ up to sign, which is equivalent to only considering the upper-half plane $\mathcal{H}$ union the positive real axis.  
The second equivalence relation comes from the region $\mathcal{P}$ in Figure \ref{F:Pdemo}, which is naturally turned into a hyperbolic surface when the Poincar\'{e} metric~\cite{PoinMetric} is applied. 
Gluing together the left and right boundaries of $\mathcal{P}$, $\mathcal{P}/\sim_{2}$ becomes homeomorphic to a truncated cylindrical surface.
The third equivalence relation $\sim_{3}$ represents a particular case when the basis vectors have an identical length, i.e. $|b_1|=|b_2|$.  In such a case, there are a number of different representations for the same lattice pattern.  This introduces many different paths for length computation, and these different paths are carefully considered for metric definition below.

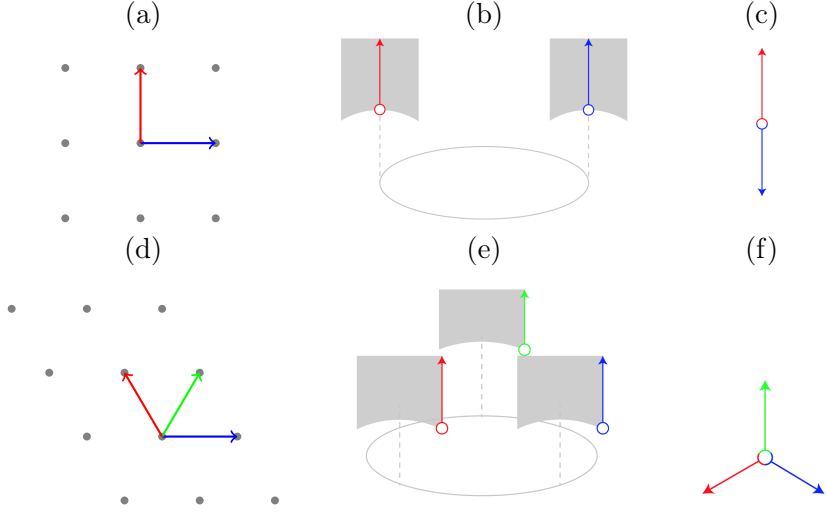
\begin{figure} 
\begin{center}
\begin{tabular}{ccccc}
(a) & \hspace{0.5cm} & (b) & \hspace{0.5cm} & (c) \\
    \begin{tikzpicture}   
     \foreach \x in {-1,0,1} {
  	\foreach \y in {-1,0,1}{
            \node[circle, fill=gray, inner sep=0pt,minimum size=3pt] (b) at (\x,\y) {}; } }
        \draw[->,thick,blue] (0,0)--(1,0) node[] {};
   \draw[->,thick,red] (0,0)--(0,1) node[] {};
    \end{tikzpicture}
& &     \includegraphics[scale=0.24]{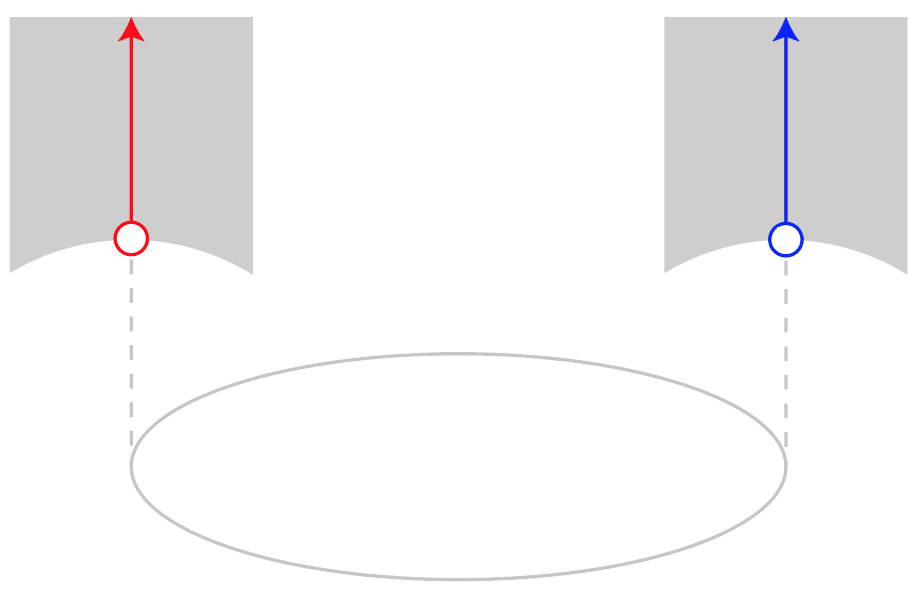} 
& &     \includegraphics[scale=0.24]{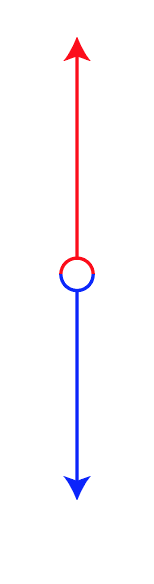} \\
(d) & & (e) &  & (f) \\
   \begin{tikzpicture}  
     \foreach \x in {-1,0,1} {
  	\foreach \y in {-1, 0,1,2}{
            \node[circle, fill=gray, inner sep=0pt,minimum size=3pt] (b) at (1*\x-1/2*\y,1.7/2*\y) {};}}
        \draw[->,thick,blue] (0,0)--(1,0) node[] {};
   \draw[->,thick,green] (0,0)--(0.5,1.7/2) node[] {};
   \draw[->,thick,red] (0,0)--(-0.5,1.7/2) node[] {}; 
    \end{tikzpicture}
 & &
     \includegraphics[scale=0.24]{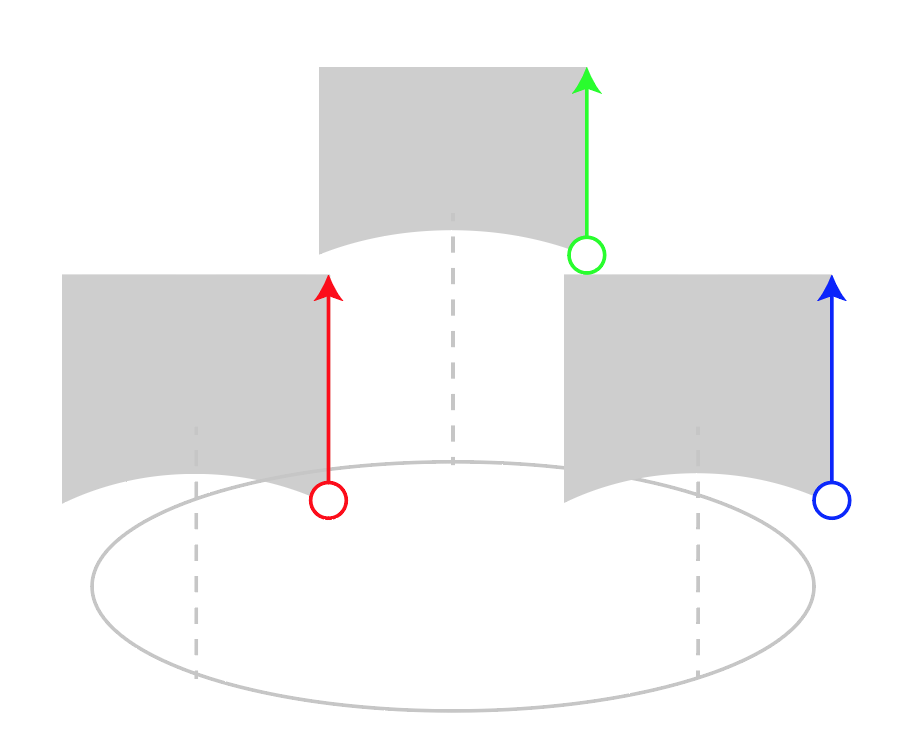} 
& &
    \includegraphics[scale=0.3]{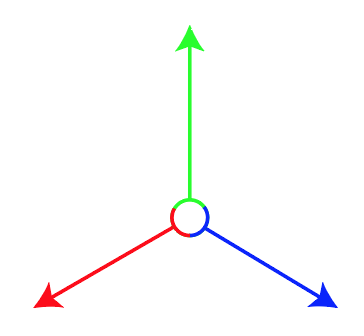} 
\end{tabular}
\end{center}
  \caption{[Examples of subspaces of $\mathscr{L}$] (a) A square lattice  $\Lambda\langle \beta, i \rangle$.  The red and blue arrows indicate two directions. Stretching along them represents two different families of lattices. They form a subspace of $\mathscr{L}$ shown in (b), and it is homeomorphic to $\mathbb{R}$ as in (c).  The second row (d) shows a lattice $\Lambda\langle \beta, e^{i\pi/3} \rangle$. Stretching along the three marked directions generates three distinct families of lattices. (e) is the subspace they form in $\mathscr{L}$, which is homeomorphic to (f). }\label{Hexagonal}
\end{figure}

To give more insights into the topologies of  lattice space $\mathscr{L}$, we present a couple of special types of lattices.  The top row of Figure~\ref{Hexagonal}, (a) illustrates the set of rectangular lattices $\Lambda \langle \beta,i \rangle$ with the red and blue lines.  It is homeomorphic to the real axis $\mathbb{R}$ as shown in (b) and (c).  The midpoint of which represents the square lattice $\Lambda \langle \beta,i \rangle$. A point $r>0$ on the positive side (the red half in (b)) is a lattice of the form $\Lambda \langle \beta,i+r\rangle $, and a point $r<0$ on the negative side (the blue half in (b)) is a lattice of the form $\Lambda \langle \beta e^{i\pi/2},i-r\rangle $. They correspond to stretching a square lattice in two directions in (c), resulting two families of distinct lattices, i.e., a bifurcate structure in $\mathscr{L}$. 
Another example is the union of hexagonal lattices and rhombic lattices whose shape descriptors have magnitude greater than $1$.  Take its subset of lattices having scale descriptors equivalent to $\beta$, as the union of the red, green and blue lines in Figure~\ref{Hexagonal} (d).  This is homeomorphic to (e), which consists of three half lines glued together at their endpoints.  They represent three directions along which stretching a hexagonal lattice gives three distinct families of lattices, rendering a trifurcate structure in $\mathscr{L}$.

On $\mathscr{L}$, we now construct a metric structure starting from defining a metric $D$ on $\mathcal{K}/\sim_{1}\times\mathcal{P}/\sim_{2}$.  Given any two descriptor pairs $(\beta,\rho), (\beta',\rho')\in\mathcal{K}\times\mathcal{P}$, we define
\begin{align}
D((\beta,\rho), (\beta',\rho'))=\sqrt{d_{\mathcal{K}}(\beta,\beta')^{2}+ d_{\mathcal{P}}(\rho,\rho')^{2}},\label{Dist}
\end{align} 
where equivalence relations will be incorporated into the definition of $d_{\mathcal{K}}$ and $d_{\mathcal{P}}$ respectively. Let $D_{\mathcal{K}}$ be a simple metric on $\mathcal{K}$,  which separates the length differences and angle differences as:
\begin{align*}
D_{\mathcal{K}}(\beta,\beta')=	\sqrt{w(|\beta|-|\beta'|)^{2}+(1-w)(\cos^{-1}\frac{\text{Re}(\beta\overline{\beta'})}{|\beta||\beta'|})^{2}}.
\end{align*}
Here $w$ is a parameter which adjusts the sensitivity between angle and length.  We use $w=0.05$ through out this paper.  The quotient metric on $\mathcal{K}$ is then defined as: 
\begin{align*}
d_{\mathcal{K}}(\beta,\beta')=\min\{D_{\mathcal{K}}(\beta,\beta'), D_{\mathcal{K}}(-\beta,\beta')\}.
\end{align*}
Let $D_{\mathcal{P}}$ be the well-known Poincar\'{e} metric~\cite{PoinMetric} restricted to $\mathcal{P}$ computed via:
\begin{align*}
D_{\mathcal{P}}(\rho,\rho')=2\ln\frac{|\rho-\rho'|+|\rho-\overline{\rho'}|}{2\sqrt{\text{Im}(\rho)\text{Im}(\rho')}},	
\end{align*}
and the corresponding quotient metric be
\begin{align*}
d_{\mathcal{P}}(\rho,\rho')= \min\{D_{\mathcal{P}}(\rho,\rho'),D_{\mathcal{P}}(\rho-1,\rho'),D_{\mathcal{P}}(\rho+1,\rho')\}.
\end{align*}

To complete the definition of $d_{\mathscr{L}}$, we consider the third equivalence relation $\sim_{3}$.   This corresponds to a particular class of lattices whose minimal bases satsify: $|b_1| =|b_2|$. They have multiple representations in the lattice space $\mathscr{L}$ using the pairs of descriptors $(\beta, \rho)$.  When considering all the path connecting any two points $(\beta,\rho), (\beta',\rho') \in \mathscr{L}$, we must consider the path passing through points in  $E = \{(\beta,\rho) \mid \beta\in\mathcal{K}, |\rho|=1, \rho\in\mathcal{P}\}$ for the third equivalence relation.  There are eight such cases:  
\begin{align}
\begin{split}
D_{1}&:(\beta,\rho)\rightarrow(\beta',e^{i\phi'})\rightarrow(\beta',\rho'),\\
D_{2}&:(\beta,\rho)\rightarrow(e^{i\phi'}\beta',-e^{-i\phi'})
\dashrightarrow(\beta',e^{i\phi'})\rightarrow(\beta',\rho'),\\
D_{3}&:(\beta,\rho)\rightarrow(\beta,e^{i\phi})\rightarrow(\beta',\rho'),\\
D_{4}&:(\beta,\rho)\rightarrow(\beta,e^{i\phi})\rightarrow(\beta',e^{i\phi'})\rightarrow(\beta',\rho'),\\
D_{5}&:(\beta,\rho)\rightarrow(\beta,e^{i\phi})\rightarrow(e^{i\phi'}\beta',-e^{-i\phi'})\dashrightarrow(\beta',e^{i\phi'})\rightarrow(\beta',\rho'),\\
D_{6}&:(\beta,\rho)\rightarrow(\beta,e^{i\phi})\dashrightarrow(e^{i\phi}\beta,-e^{-i\phi})\rightarrow(\beta',\rho'),\\
D_{7}&:(\beta,\rho)\rightarrow(\beta,e^{i\phi})\dashrightarrow(e^{i\phi}\beta,-e^{-i\phi})\rightarrow(\beta',e^{i\phi'})\rightarrow(\beta',\rho'),\\
D_{8}&:(\beta,\rho)\rightarrow(\beta,e^{i\phi})\dashrightarrow(e^{i\phi}\beta,-e^{-i\phi})
\rightarrow(e^{i\phi'}\beta',-e^{-i\phi'})
\dashrightarrow(\beta',e^{i\phi'})\rightarrow(\beta',\rho'),
\end{split}
 \label{E:Di}
\end{align}
here $\rightarrow$ indicates the distance between two nodes using $D$ in~(\ref{Dist}), and  $\dashrightarrow$ represents a path of zero length because of the third equivalence relations $\sim_{3}$. The angles $\phi,\phi'$ lie in $[\pi/3,2\pi/3]$. Notice all the involved points other than $(\beta, \rho)$ and $(\beta', \rho')$ are in $E$.   Figure~\ref{digraph} illustrates these paths as a diagram.  Figure~\ref{latticedemo} illustrates the path in the lattice space $\mathscr{L}$ showing the examples of $D$ in green, $D_3$ from (\ref{E:Di}) in blue, and $D_8$ from (\ref{E:Di}) in red.   

\begin{figure}
\centering
\begin{tikzpicture}[x=1.3cm, y=1cm,
    every edge/.style={
        draw,
        postaction={decorate,
                    decoration={markings,mark=at position 0.4 with {\arrow{>}}}  }        } ]
	\vertex (c6) at (3,-1.5) [label=right:{$(e^{i\phi'}\beta',-e^{-i\phi'})$}]{};
	\vertex (c5) at (3,0) [label=right:{$(\beta',e^{i\phi'})$}]{};
	\vertex[fill] (c1) at (-3,1.5) [label=left:{$(\beta,\rho)$}]{};
	\vertex[fill] (c4) at (3,1.5) [label=right:{$(\beta',\rho')$}]{};
	\vertex (c2) at (-3,0) [label=left:{$(\beta,e^{i\phi})$}]{};
	\vertex (c3) at (-3,-1.5) [label=left:{$(e^{i\phi}\beta,-e^{-i\phi})$}]{};
	\path 
		(c1) edge (c5)
		(c1) edge (c6)
		(c2) edge (c4)
		(c2) edge (c5)
		(c2) edge (c6)
		(c3) edge (c4)
		(c3) edge (c5) 
		(c3) edge (c6)
		(c1) edge (c2) 
		(c2) edge [dash dot](c3) 
		(c5) edge (c4) 
		(c6) edge [dash dot](c5);
\end{tikzpicture}
\caption{[Paths through $E$]  This is an illustration of the 8 types of paths, $D_1-D_8$ connecting $(\beta, \rho)$ and $(\beta', \rho')$ via points in $E= \{(\beta,\rho) \mid \beta\in\mathcal{K}, |\rho|=1, \rho\in\mathcal{P}\}$.  Notice all four points other than $(\beta, \rho)$ and $(\beta', \rho')$ are in $E$.  The solid line represents the path length computed by $D$, while the dash line represents the third equivalence relation $\sim_{3}$ (no length added).  }\label{digraph}
\end{figure}
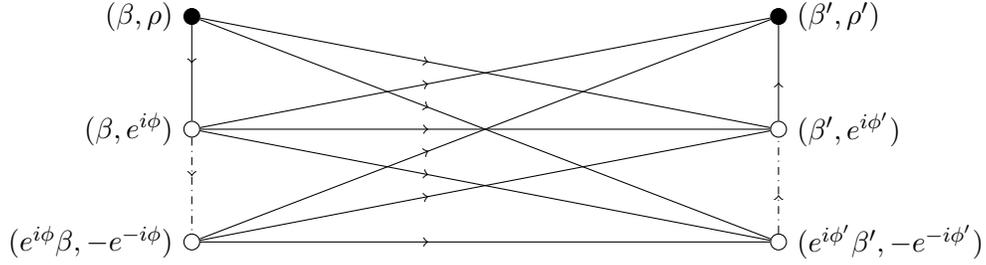 
\begin{figure}
  \centering
\includegraphics[scale=0.35]{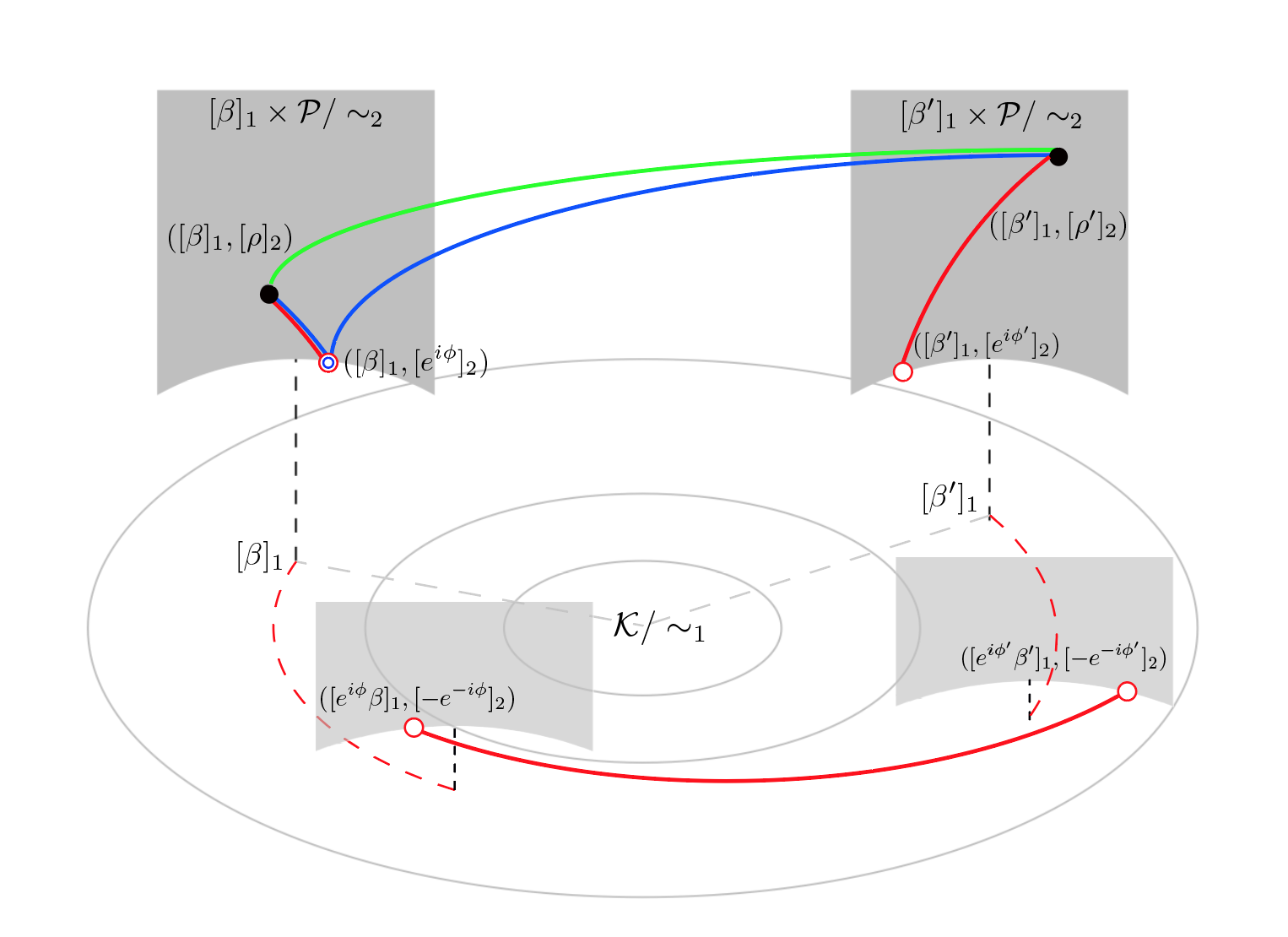}  
\caption{
The lattice space $\mathscr{L}$ is a product space $\mathcal{K}/\sim_{1}\times \mathcal{P}/\sim_{2}$ modulo the points through $E$.  
The distance $d_{\mathscr{L}}((\beta, \rho),(\beta',\rho'))$ is the minimal lengths of the paths. Here the green line shows $D$ in (\ref{Dist}), the blue line is $D_3$, and the red is $D_8$ in (\ref{E:Di}).  Dash lines represent the third equivalence relations (no length).}\label{latticedemo}
\end{figure}

Combining all these distance, \textbf{ the metric $d_{\mathscr{L}}$} between any  two lattices $\Lambda(\beta,\rho)$ and $\Lambda(\beta',\rho')\in\mathscr{L}$ is defined as 
\begin{align}
	d_{\mathscr{L}}((\beta,\rho),(\beta',\rho'))=\min\{D,\min_{\phi,\phi'\in[\pi/3,2\pi/3]}D_{j}(\phi,\phi'), j=1,\dots,8\},\label{dLdef}
\end{align}
here $D$ is from (\ref{Dist}), and  $\{D_{j}(\phi,\phi')\}_{j=1}^{8}$  are from (\ref{E:Di}).   For completeness, we present a pseudo-code for computing $d_{\mathscr{L}}$ in Appendix \ref{A:code}.
This metic $d_{\mathscr{L}}$ is the minimum among all the paths in $\mathcal{K}/\sim_{1}\times \mathcal{P}/\sim_{2}$ connecting any two points $(\beta,\rho)$ and $(\beta',\rho')$, thus $d_{\mathscr{L}}$ is a pseudometric on $\mathscr{L}$. 
In fact, $d_{\mathscr{L}}((\beta,\rho),(\beta',\rho'))=0$ if and only if $[\beta,\rho]=[\beta',\rho']$, hence $d_{\mathscr{L}}$ becomes a metric on $\mathscr{L}$.

\begin{Rem}
Notice that $d_{\mathscr{L}}$ is invariant under translation. It takes inputs from the lattice space $\mathscr{L}$, where only translational lattices are concerned. The visual difference between a lattice and its translated copy can be regarded as a consequence of the boundedness of the image domain, thus it is not intrinsic to the patterns.
\end{Rem}

\subsection{Visual validation of the lattice space $\mathscr{L}$ and metirc $d_{\mathscr{L}}$}

For the purpose of comparison, one may define the following 4-tuple from the classical definition \ref{latdef} of lattice.  For any minimal lattice basis $(b_{1},b_{2})$, 
\begin{align}
(|b_{1}|,|b_{2}|,\theta,\psi):=(|b_{1}|,|b_{2}|,\text{Arg}\,b_{1},\cos^{-1}(\frac{\text{Re}(b_{1}\overline{b}_{2})}{|b_{1}||b_{2}|})),\label{vecpar}
\end{align}
where $\theta$ taking values from $(-\pi/2,\pi/2]$ is the angle of $b_{1}$ to the positive real axis, and $\psi\in(0,\pi]$ is the angle between $b_{1}$ and $b_{2}$. The differences in these parameters also reflect the visual differences in the lattice patterns. 

Figure~\ref{metriccomp} and its table show effects of using the setting of $(b_1,b_2)$.  Comparing a pair of very similar lattices in  (a) $\Lambda_A=\Lambda(12,12.5,10^{\circ},90^{\circ})$, and (b)  $\Lambda_B=\Lambda(12,12.5,-80^{\circ},90^{\circ})$ expressed using the 4-tuples in (\ref{vecpar}),  (\ref{vecpar}) gives a large relative difference for $\theta$, 900\%, while $d_{\mathscr{L}}=0.0816$ gives a small value.   When $|b_{1}|\approx|b_{2}|$, minor numerical errors trigger large relative errors in $\theta$-component due to the equivalence relations. 
The lattices (a) $\Lambda_A$, (c)  $\Lambda_C$, and   (d)  $\Lambda_D$, are more distinguishable, yet the differences in the second and third rows of the table  fail to reflect this.  $d_{\mathscr{L}}$ is more stable and consistent in representing the similarity and differences. 
\begin{figure}
\begin{center}
\begin{tabular}{cccc}
(a) $\Lambda_A$ & (b) $\Lambda_B$ & (c) $\Lambda_C$ & (d) $\Lambda_D$ \\
\includegraphics[width=1.3in]{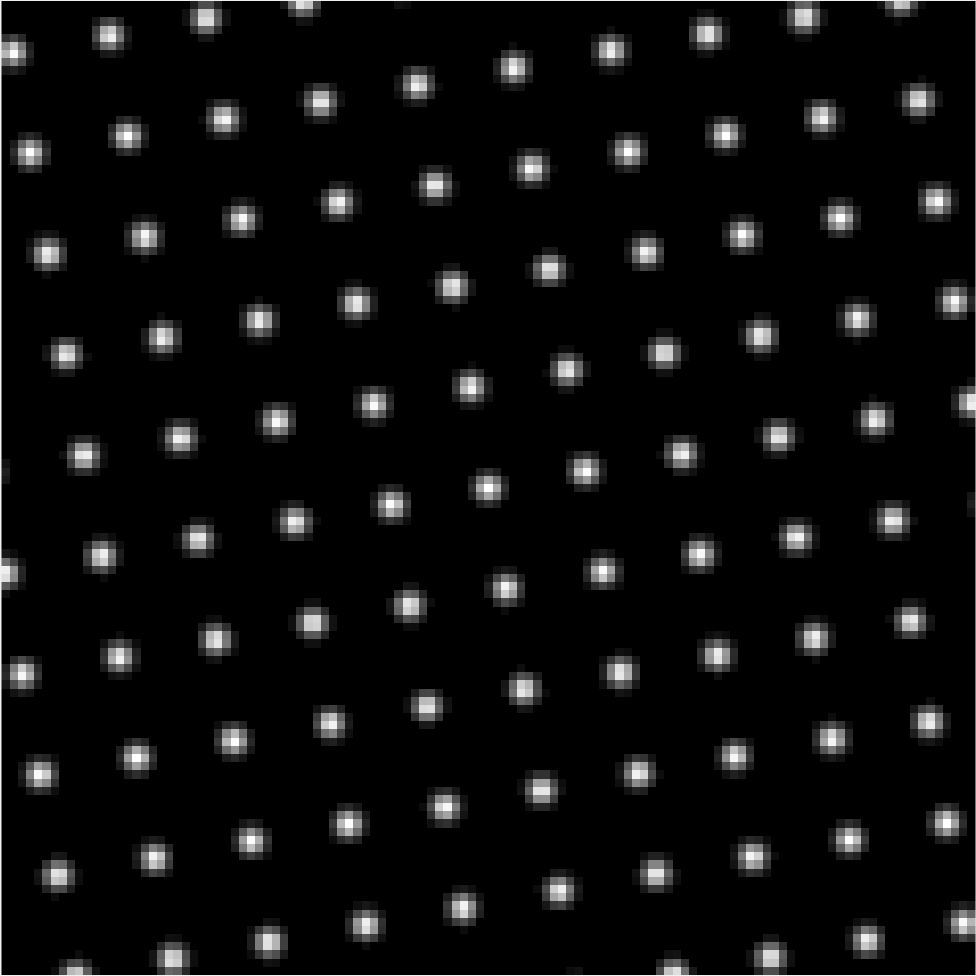} &
    \includegraphics[width=1.3in]{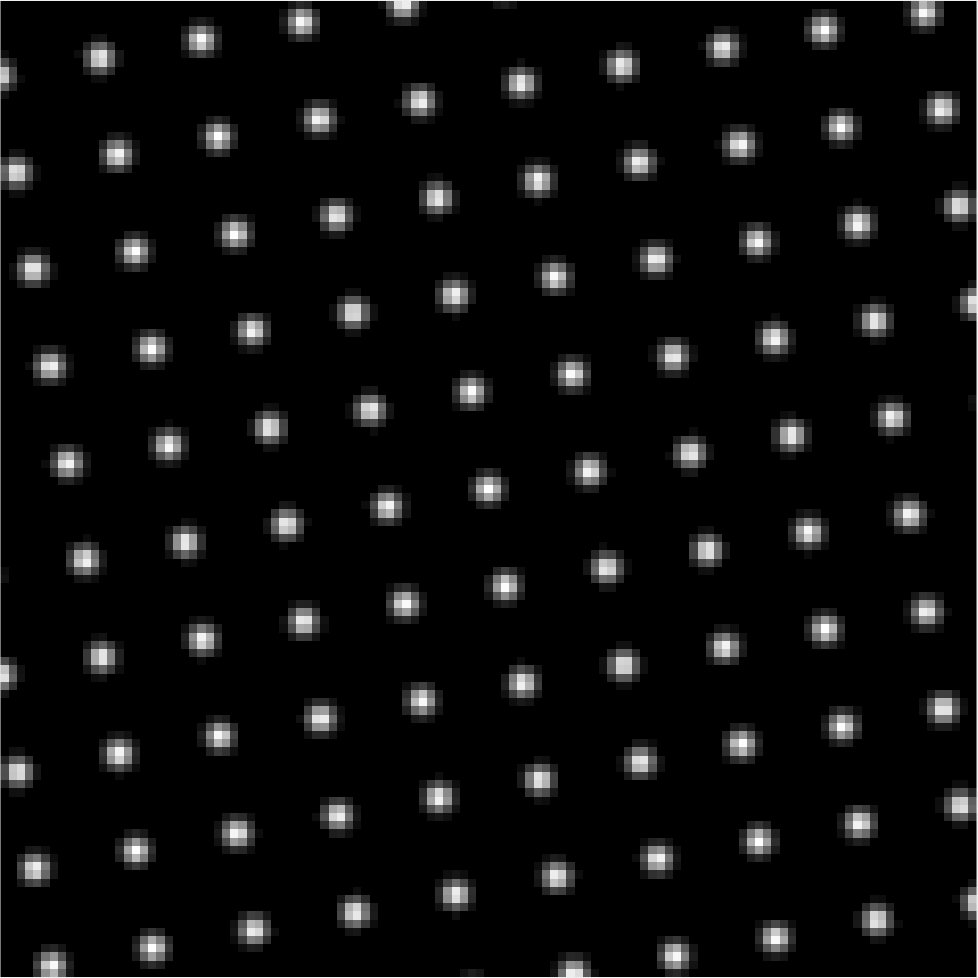} &
    \includegraphics[width=1.3in]{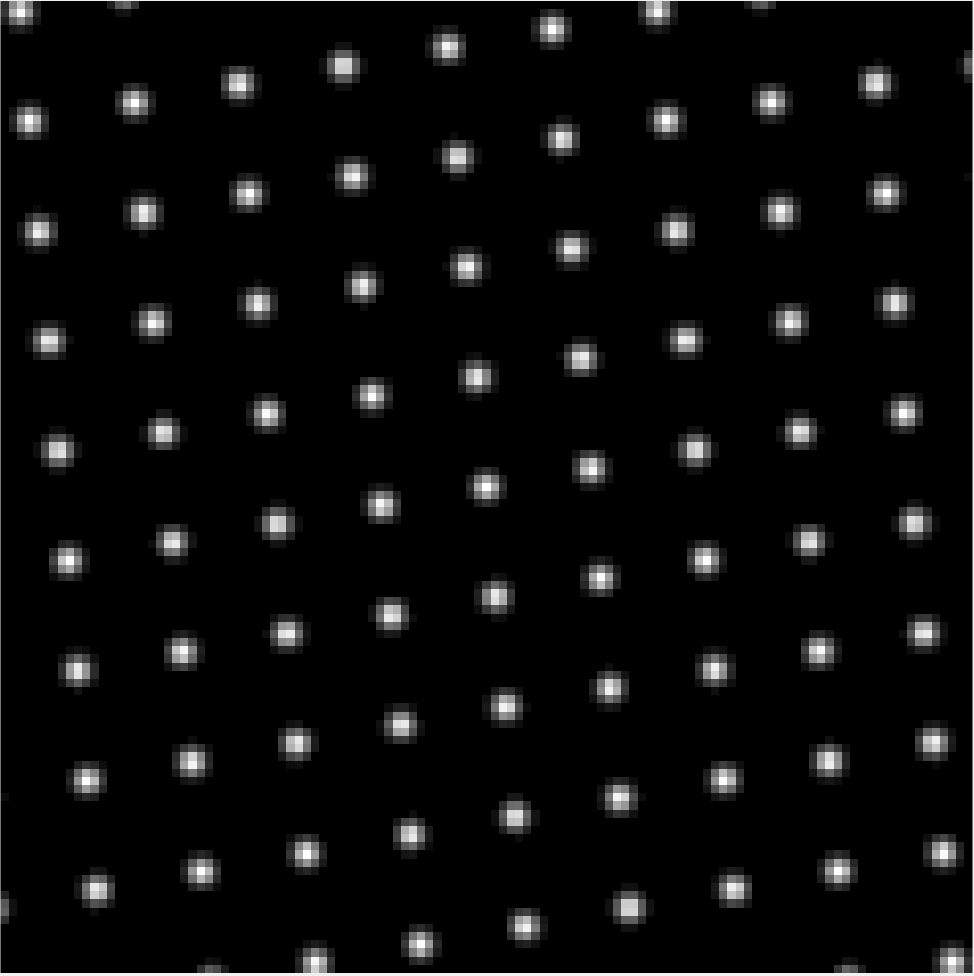} &
    \includegraphics[width=1.3in]{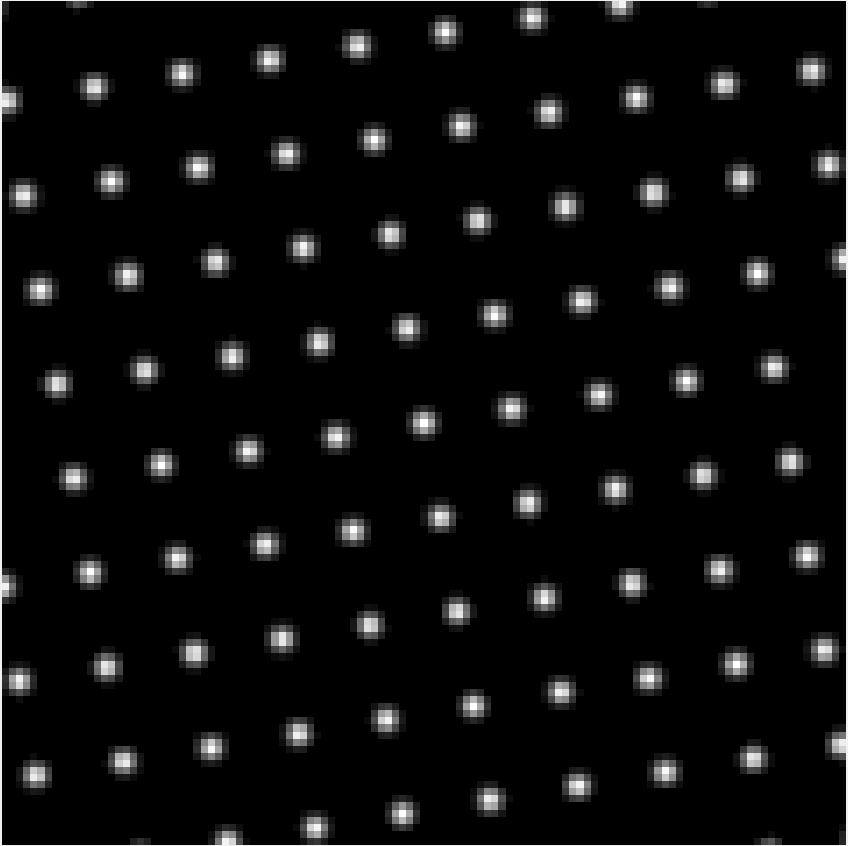} \\
    \end{tabular}
\begin{tabular}{|c||c|c|c|c||c|} \hline 
   Lattice pair  &\multicolumn{4}{|c||}{4-tuple measure system (\ref{vecpar})}& $d_{\mathscr{L}}$ \\
      &$||b_{1}'|-|b_{1}||/|b_{1}|$ & $||b_{2}'|-|b_{2}||/|b_{2}|$ &  $|\theta'-\theta|/|\theta|$ & $|\psi'-\psi|/|\psi|$&
       \\\hline
      $\Lambda_{A},\Lambda_{B}$&$0\%$ & $0\% $&  $900\%$ & $0\%$ &$    0.0816$ \\\hline
	  $\Lambda_{A},\Lambda_{C}$&$8.3333\%$ & $8\% $&  $0\%$ & $5.5556\%$ &$ 0.2401$ \\\hline
	  $\Lambda_{A},\Lambda_{D}$&$4.1667\%$ & $8\% $&  $10\%$ & $1.1111\%$ &$ 0.1200$ \\\hline
    \end{tabular}  
    \end{center}
      \caption{[Metric Comparison] Lattice (a) $\Lambda_A=\Lambda(11.8177 + 2.0838i,-2.1706 +12.3101i)$  and (b) $\Lambda_B=\Lambda(2.0838 -11.8177i,12.3101 + 2.1706i)$ are visually similar.  4-tuple measure shows instability in the values, while $d_{\mathscr{L}}$ give a small value.  The lattices (a), (c) $\Lambda_C=\Lambda(-1.1766 +13.4486i,-2.0838 +11.8177i)$ and (d) are more distinguishable.  While 4-tuple measure  doesn't reflect this consistently comparing the second and third row of the Table, $d_{\mathscr{L}}$ is more stable and consistent in representing the similarity and differences.}\label{metriccomp}
\end{figure}

Figure~\ref{dL} presents more examples of lattice patterns and their pairwise distances.  There are five different lattices patterns shown in (a)-(e).  Comparing lattices $\Lambda_A$ to $\Lambda_B$ or $\Lambda_C$, visually lattice  $\Lambda_C$ seems more different from $\Lambda_A$ than $\Lambda_B$. The corresponding distances $d_{\mathscr{L}}(\Lambda_A, \Lambda_C) = 0.7083>d_{\mathscr{L}}(\Lambda_A, \Lambda_B) = 0.5493$ are consistent with this observation.  Among the lattices, visually $\Lambda_A$ and $\Lambda_D$ seems the most similar and $d_{\mathscr{L}}(\Lambda_A, \Lambda_D) = 0.0203$ is the smallest.  The differences between lattice  $\Lambda_B$ and  $\Lambda_C$, and the differences between lattice  $\Lambda_D$ and  $\Lambda_E$, seems similar, and this is well represented by the distance $d_{\mathscr{L}}(\Lambda_B, \Lambda_C) =  0.4472$ and $d_{\mathscr{L}}(\Lambda_D, \Lambda_E) =0.4472$ being close.  Also the differences between $\Lambda_B$ and  $\Lambda_D$, and the differences between lattice  $\Lambda_C$ and  $\Lambda_E$, are also similar $d_{\mathscr{L}}(\Lambda_B, \Lambda_D) =  0.5293$ and $d_{\mathscr{L}}(\Lambda_C, \Lambda_E) =0.5293$. 
 
\begin{figure} 
\begin{center}
\begin{tabular}{ccccc}
(a) &  (b) & (c) & (d)  & (e)\\
 \includegraphics[width=1.2in]{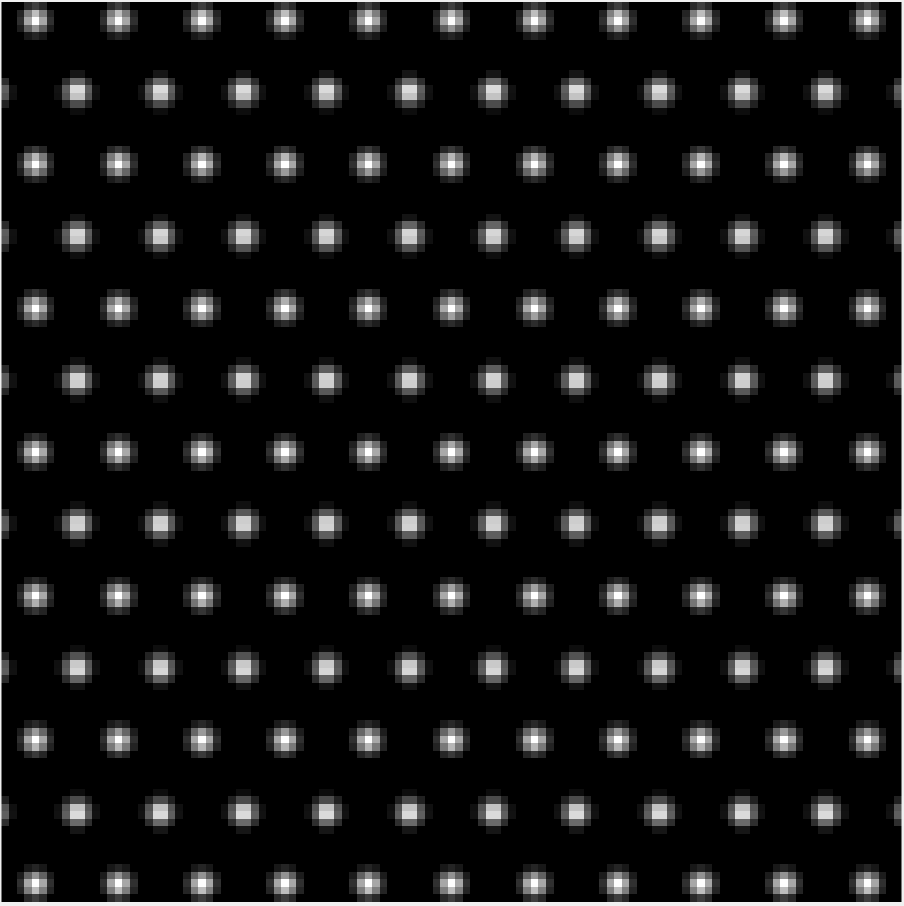}  & 
 \includegraphics[width=1.2in]{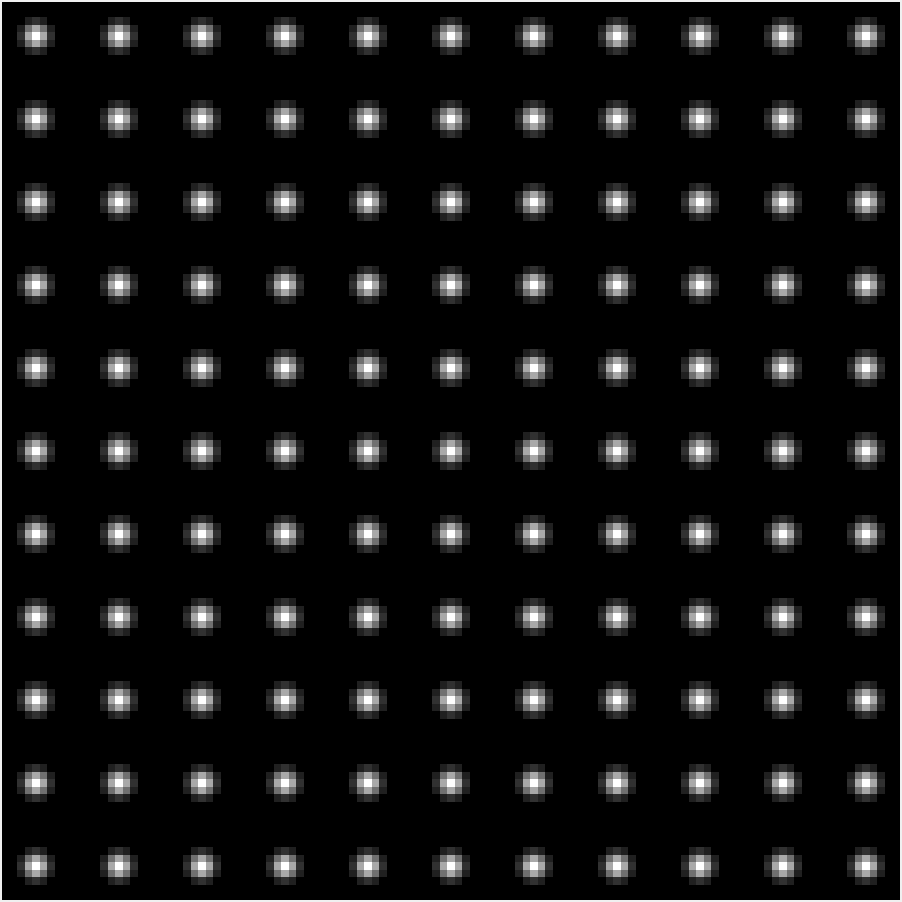} & 
 \includegraphics[width=1.2in]{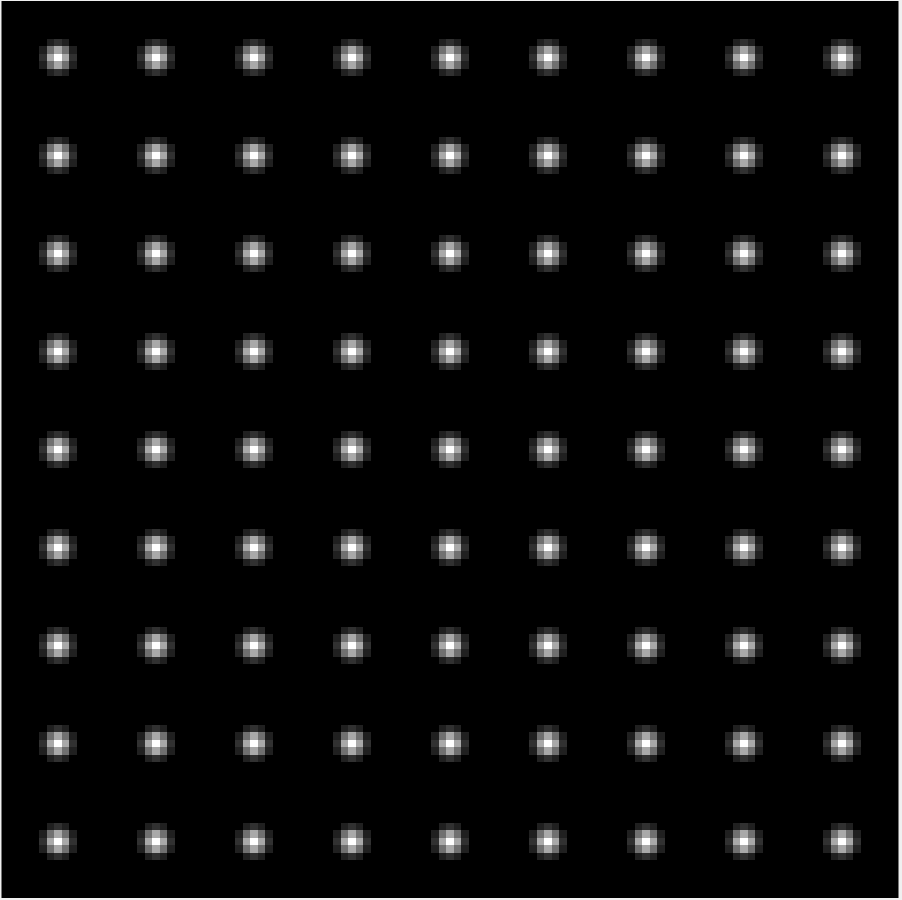} &
 \includegraphics[width=1.2in]{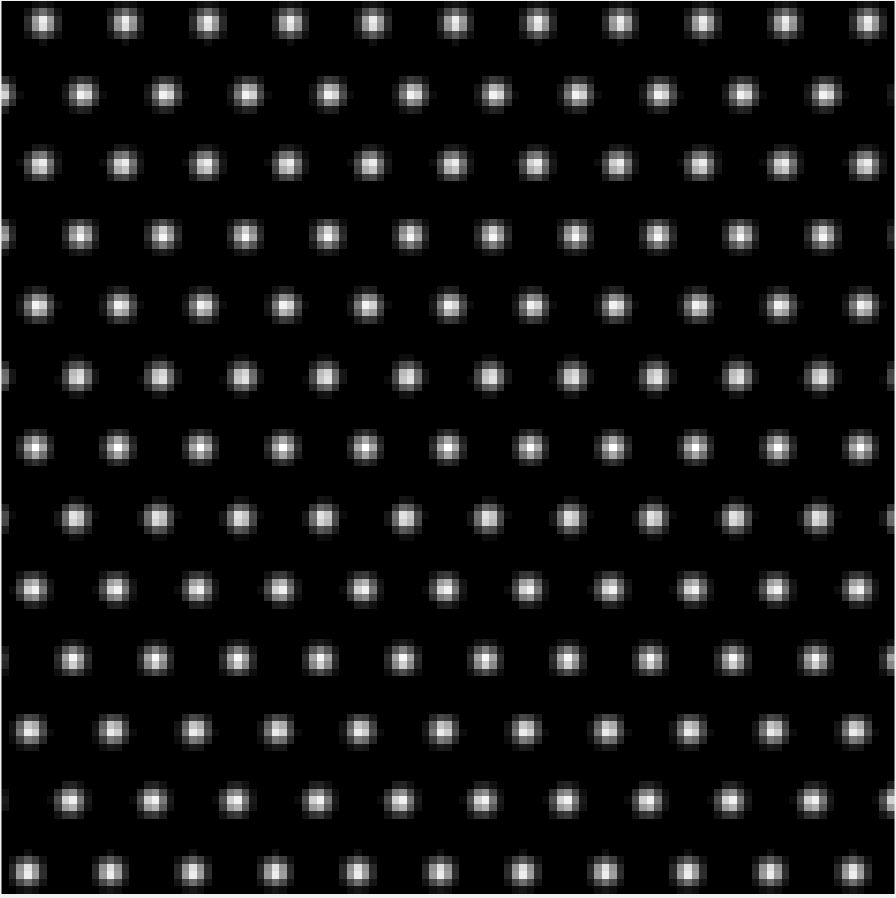}  & 
  \includegraphics[width=1.2in]{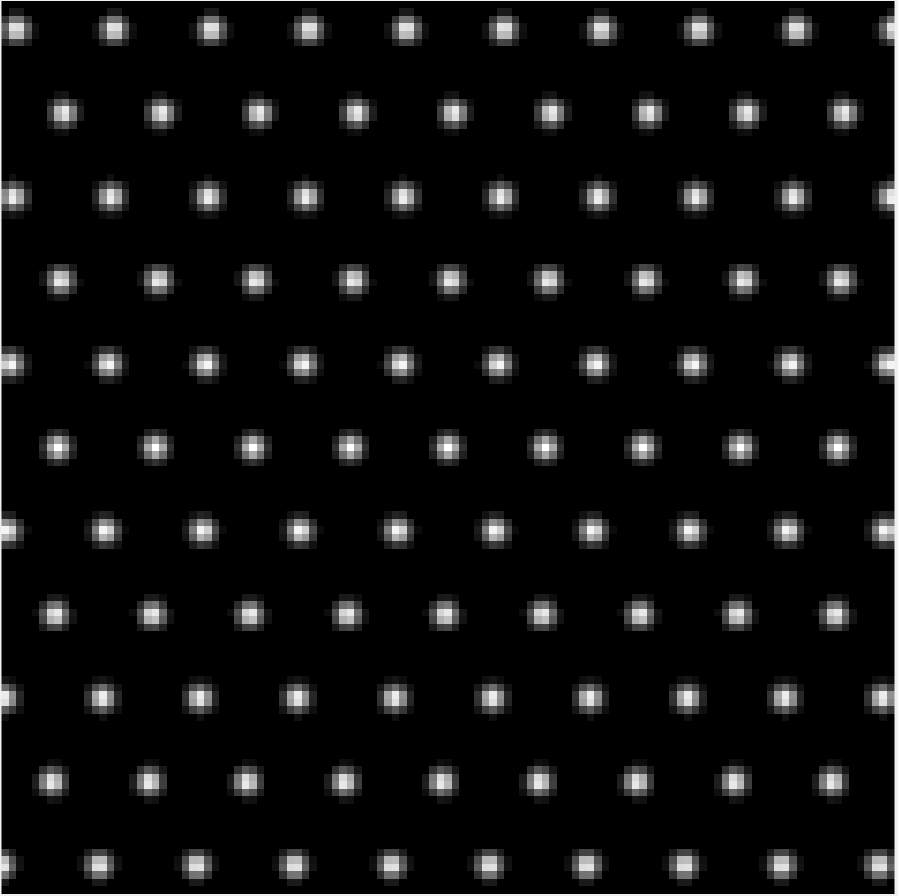} \\
\end{tabular}
\end{center}
  \caption{[Visual effects of $d_{\mathscr{L}}$] Five different lattices: (a) $\Lambda_{A}=\Lambda\langle 11,e^{i\pi/3}\rangle$ (b) $\Lambda_{B}=\Lambda\langle 11,e^{i\pi/2}\rangle$ (c) $\Lambda_{C}=\Lambda\langle 13,e^{i\pi/2}\rangle$ (d) $\Lambda_{D}=\Lambda\langle 11,e^{i61\pi/180}\rangle$ (e)  $\Lambda_{E}=\Lambda\langle 13,e^{i61\pi/180}\rangle$.  Pairwise distances: $d_{\mathscr{L}}(\Lambda_{A},\Lambda_{B})=0.5493$, $d_{\mathscr{L}}(\Lambda_{A},\Lambda_{C})=0.7083$,$d_{\mathscr{L}}(\Lambda_{A},\Lambda_{D})=0.0203$, $d_{\mathscr{L}}(\Lambda_{A},\Lambda_{E})=0.4477$, $d_{\mathscr{L}}(\Lambda_{B},\Lambda_{C})=0.4472$, $d_{\mathscr{L}}(\Lambda_{B},\Lambda_{D})=0.5293$, $d_{\mathscr{L}}(\Lambda_{B},\Lambda_{E})=0.6929$,  $d_{\mathscr{L}}(\Lambda_{C},\Lambda_{D})=0.6929$, $d_{\mathscr{L}}(\Lambda_{C},\Lambda_{E})=0.5293$,  $d_{\mathscr{L}}(\Lambda_{D},\Lambda_{E})=0.4472$. These values correspond well with the visual perception of the lattice differences. }
  \label{dL}
\end{figure}

\section{New Lattice Identification and Separation Algorithm (LISA)}\label{sec:LISA}

In this section, we propose an efficient algorithm to separate each lattice pattern from a superlattice in practice.  First, we present the variational formulation for lattice identification and separation.  Then we introduce our algorithm which does not require any prior knowledge of the number lattice mixture. 

\subsection{Variational Model for Lattice Separation}\label{varsetup}
For a given image with mixture of lattices $U:\Omega\subseteq\mathbb{R}^{2}\to[0,1]$ as in (\ref{imagemodel}), we identify the underlying lattice patterns by minimizing the following energy:
\begin{align}
\min_{K\in\mathbb{N}, \Lambda_{j}\in\mathscr{L},\mu_{j}\in\mathbb{C}}\int_{\Omega}|U-\max_{j=1,\cdots,K}\mathcal{T}_{\mu_{j}}\Lambda_{j}|\,dx\,dy+hK,\label{varprob}
\end{align}
where $dx\,dy$ is the Lebesgue measure on $\mathbb{R}^{2}$, and $h>0$ is a  penalty coefficient. To avoid identifying multiple sub-lattices to approximate a single denser lattice, we suppress the number of different lattice pattern while fitting the mixture to the given image.  See more discussion on sub-lattice in Appendix \ref{A:subL}.

For a fixed $K$, this energy is balancing two competing terms. Using $|a-b|=a+b-2\min(a,b)$ for any $a,b\in\mathbb{R}$, the minimization (\ref{varprob}) becomes:
\begin{align}
\min_{\Lambda_{j}\in\mathscr{L},\mu_{j}\in\mathbb{C}}\{\int_{\Omega}U-\min(U,\max_{j=1,\dots,K}\mathcal{T}_{\mu_{j}}\Lambda_{j})\,dx\,dy+\int_{\Omega}\sum_{j=1}^{K}\mathcal{T}_{\mu_{j}}\Lambda_{j}-\min(U,\max_{j=1,\dots,K}\mathcal{T}_{\mu_{j}}\Lambda_{j})\,dx\,dy\}.\label{varprob2}
\end{align}
These integrals are equivalent to counting particles: if $U$ and $\mathcal{T}_{\mu_{j}}\Lambda_{j}$, $j=1,\dots, K$ denote the sets of particles they contain respectively, then we have the following correspondences:  
\begin{align*}
\int_{\Omega}U-\min(U,\max_{j=1,\dots,K}\mathcal{T}_{\mu_{j}}\Lambda_{j})\,dx\,dy&\Longleftrightarrow U-U\bigcap\bigcup_{j=1}^{K}\mathcal{T}_{\mu_{j}}\Lambda_{j},\\
\int_{\Omega}\max_{j=1,\dots,K}\mathcal{T}_{\mu_{j}}\Lambda_{j}-\min(U,\max_{j=1,\dots,K}\mathcal{T}_{\mu_{j}}\Lambda_{j})\,dx\,dy&\Longleftrightarrow \bigcup_{j=1}^{K}\mathcal{T}_{\mu_{j}}\Lambda_{j}-U\bigcap\bigcup_{j=1}^{K}\mathcal{T}_{\mu_{j}}\Lambda_{j}.
\end{align*}
The sets on the right hand sides can be further expressed as
\begin{align*}
\bigcap_{j=1}^{K}\big(U\bigcap(\mathcal{T}_{\mu_{j}}\Lambda_{j})^{c}	\big),~\text{and}~\bigcup_{j=1}^{K}\big(\mathcal{T}_{\mu_{j}}\Lambda_{j}\bigcap U^{c}\big).
\end{align*}
The problem (\ref{varprob2}) is thus equivalent to:
\begin{align}
\min_{\Lambda_{j}\in\mathscr{L},\mu_{j}\in\mathbb{C}}&\{\underbrace{\int_{\Omega}U-\max_{j=1,\dots,K}(\min(U,\mathcal{T}_{\mu_{j}}\Lambda_{j}))\,dx\,dy}_\text{under-fitting}+\underbrace{\int_{\Omega}\max_{j=1,\dots, K}(\mathcal{T}_{\mu_{j}}\Lambda_{j}-\min(U,\mathcal{T}_{\mu_{j}}\Lambda_{j}))\,dx\,dy}_\text{over-fitting}\}.
\label{varprob3}	
\end{align}
The first term in the objective function is to measure the remaining intensities of $U$ after points being extracted by $K$ lattices, i.e., the under-fitting. The second term evaluates the total excessive intensities of these $K$ lattices, i.e., the over-fitting. As $K$ increases, the under-fitting decreases. If we control the over-fitting to be $0$, i.e., every lattice candidate has no extra lattice points, then by adding more layers, (\ref{varprob2}) reaches the minimum. Therefore, we solve~(\ref{varprob}) using a greedy strategy, which leads to LISA.

\subsection{Lattice Identification and Separation Algorithm (LISA)}

\begin{table}
\noindent\rule{\textwidth}{1pt}
\textbf{Lattice Identification and Separation Algorithm (LISA)} \\
\noindent\rule{\textwidth}{0.5pt}
\textbf{Inputs:} \begin{enumerate}\vspace{-0.3cm}
\item $U$: given gray scale image with particles;\vspace{-0.3cm}
\item $J$: a parameter to control the number of lattice candidates;\vspace{-0.3cm}
\item (\textit{Optional}) $K$: number of iteration for the optional step 3.
\end{enumerate}
Let $j=1$. While TRUE:\\
\textbf{Step 1.}  Compute Fourier transform of $U$ on polar coordinate.  Collect the local maximum of spectrum surface of  $U$, $C_{j}=\{x_{1},x_{2},\cdots,x_{M}\}$  within $J$ connected components.\\
\textbf{Step 2.} For every pair $(x_{k},x_{l})\in C_{j}$, $k\neq l$, construct a lattice pattern and compute the translation  $\mu_{k,l}$ to get $\mathcal{T}_{\mu_{k,l}}\Lambda_{(k,l)}$.  Take $\mathcal{T}_{\mu_{j}}\Lambda_{j}=\arg\min_{k,l=1,\dots,M;k\neq l}\,\mathcal{E}(\mathcal{T}_{\mu_{k,l}}\Lambda_{(k,l)})$ as in (\ref{eval}). \\
\textbf{Step 3.} (\textit{Optional}) correction of $\mathcal{T}_{\mu_{j}}\Lambda_{j}$ using K iteration. \\
\textbf{Step 4.} For the identified optimal $\mathcal{T}_{\mu_{j}}\Lambda_{j}$, if $\text{mean}(U-\mathcal{T}_{\mu_{j}}\Lambda_{j})<0.01$, terminate the algorithm. Otherwise, set $U= \mathcal{F}(U-\mathcal{T}_{\mu_{j}}\Lambda_{j})$, $j=j+1$ and repeat. \\
\rule{\textwidth}{1pt}
\caption{Lattice Identification and Separation Algorithm}\label{LISAtabel}
\end{table} 

\begin{figure}
\centering
\includegraphics[width=6in]{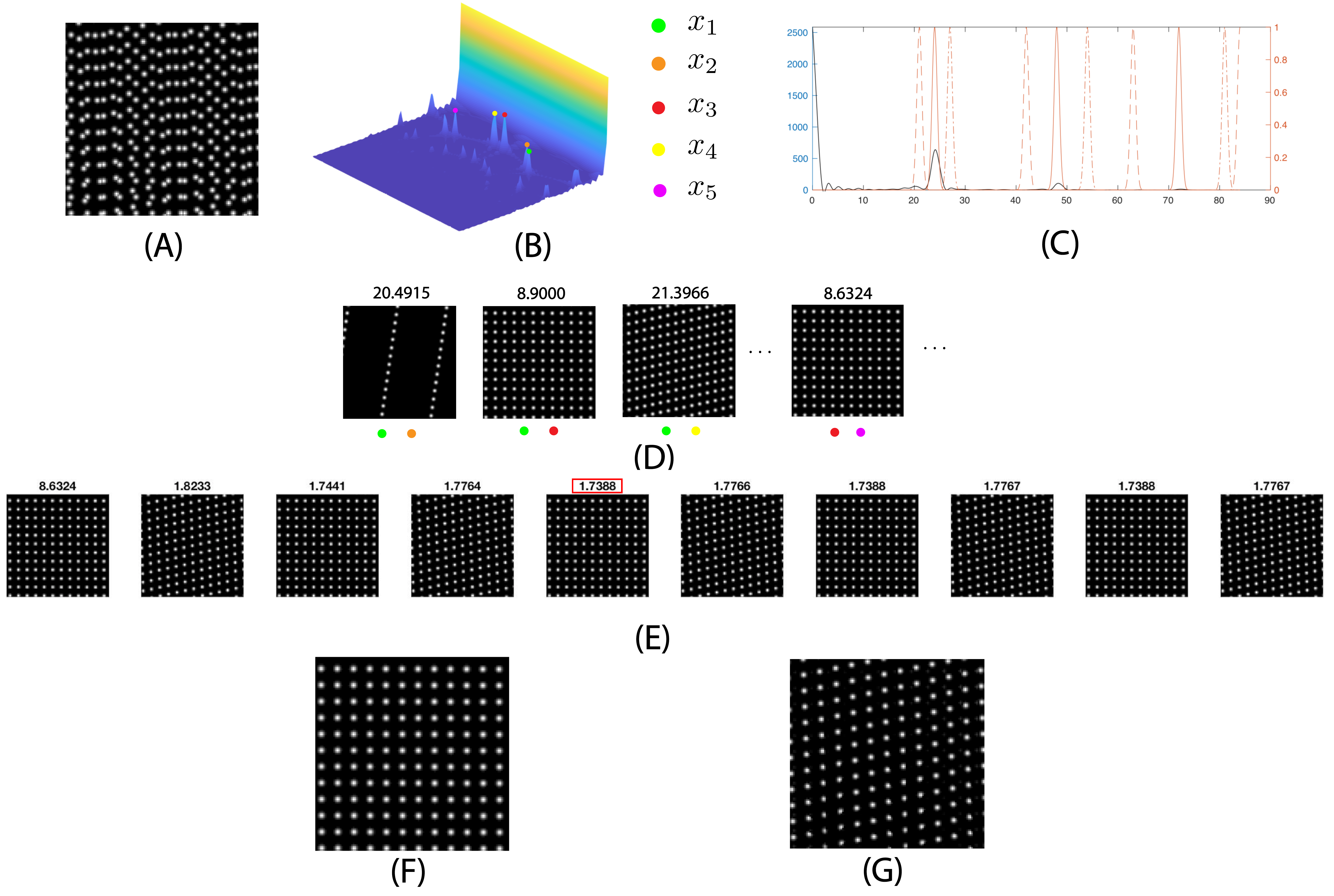}
\caption{[Steps of LISA] (A) An image processed by $\mathcal{F}$ (\ref{Hdef}). \textbf{Step1:} (B) the spectrum surface via polar coordinate, and the high responses  $x_{1},\cdots, x_{5}$ above a threshold $J$. (C) peak location refinement via repeating Gaussian impulses.  
\textbf{Step 2:} (D) generates lattice candidates $\mathcal{T}_{\mu_{k,l}}\Lambda_{(k,l)}$, $k,l=1,2,\cdots, 5, k\neq l$ for each pair of high peaks, and the energy  (\ref{eval}) is computed for each candidate.  Pick $(x_{3},x_{5})$ to be the optimal $\mathcal{T}_{\mu_{1}}\Lambda_{1}$, for it has the lowest energy.  
\textbf{Step 3} (optional correction step): (E) Among $\mathcal{T}_{\mu_{1}}\Lambda_{1}^{(i)}$, $i=1$ to $K$ ($K=10$) $\mathcal{T}_{1}\Lambda_{\mu_{1}}^{(5)}$ has the lowest score, hence replaces $\mathcal{T}_{\mu_{1}}\Lambda_{1}$.  
\textbf{Step 4:} (F) The optimal lattice $\mathcal{T}_{\mu_{1}}\Lambda_{1}$ for this iteration. (G) The remainder image.  The error is above the accuracy criteria $0.0710>0.01$, goes to next iteration.}	
\label{LISA}
\end{figure}

The outline of the algorithm is presented in  Table \ref{LISAtabel} and a demonstration of the workflow is in Figure~\ref{LISA}.  

In real applications,  there are inhomogeneities in the particle sizes, shapes, and intensities, which complicates the identification. After background denoising (e.g., Otsu's method~\cite{Otsu}) if necessary, we process the image by replacing each local maximum in the image with a uniform size Gaussian PSF $G_{\sigma}$. We denote this processing as an operator $\mathcal{F}$:
\begin{align}
	\mathcal{F}(U) = G_{\sigma}*\delta(|\nabla U|).\label{Hdef}
\end{align}
For low or medium resolution, we apply the Gaussian approximation method to calibrate the peak locations as necessary (also see~\cite{algoTEST,Qua,GR}). Let $U(x,y)$ be a discrete local maxima, i.e., $U(x,y)\geq U(x',y')$, for $x'=x\pm 1$ and $y'=y\pm 1$. The approximated coordinate for the real local maxima becomes:
\begin{align*}
\widehat{x} &= x-\frac{\log(U(x+1,y))-\log(U(x-1,y))}{2(\log(U(x+1,y))+\log(U(x-1,y))-2\log(U(x,y)))}, \\	
\widehat{y} &= y-\frac{\log(U(x,y+1))-\log(U(x,y-1))}{2(\log(U(x,y+1))+\log(U(x,y-1))-2\log(U(x,y)))}.
\end{align*}

In Step 1, we compute the Radon transform by a B-spline convolution-based Radon transform  proposed by Horbelt et al.~\cite{splineRadon}. The result is a 1D signal for each projecting angle, upon which we apply the standard 1D FFT. The collection of these 1D spectra form the 2D Fourier transform of the image on the polar coordinate (see Theorem~\ref{FourierSlice}), see Figure~\ref{LISA}(b). For computational efficiency, we focus on peaks with sufficient heights, e.g. $x_1,\dots,x_5$.  This height threshold is picked such that, above it, there are $J$ connected components of the power spectrum, i.e., $J=5$ in Figure~\ref{LISA}(b).  
To achieve sub-pixel precision,  the distance of a peak response to the origin is adjusted by perturbation along the radial direction.  Figure~\ref{LISA}(c) demonstrates this process: consider trains of Gaussian impulses  placed periodically along the radial direction (periodicities perturbed around the peak distance to the origin), and choose the one with the most overlap with the signal to be the adjusted distance.  

In Step 2, each pair of local maxima on the spectrum surface corresponds to a lattice candidate in the image domain. Figure~\ref{LISA}(d) shows 4 examples of such combinations.  Fourier transform of a lattice $\Lambda(b_{1},b_{2})$ in the image domain is a lattice in the frequency domain, called its reciprocal lattice $\Lambda(\omega_{1},\omega_{2})$. The formula transferring the basis  $(\omega_{1},\omega_{2})$ in frequency domain to the basis $(b_{1},b_{2})$ in image domain is:
\begin{align*}
\begin{cases}b_{1}=(\omega_{2}\times[0,0,1]^{T})/(|\omega_{1}\times\omega_{2}\cdot[0,0,1]^{T}|)\\
b_{2}=([0,0,1]^{T}\times\omega_{1})/(|\omega_{1}\times\omega_{2}\cdot[0,0,1]^{T}|)
\end{cases}.	
\end{align*}

The translation for each candidate, Figure~\ref{LISA}(e), is then identified by the maximum of the cross-correlation function between the candidate and the original image. To evaluate the lattice candidates, we use the following energy:
\begin{align}
\mathcal{E}(\mathcal{T}_{\mu}\Lambda):=\underbrace{||\mathcal{F}(U-\mathcal{T}_{\mu}\Lambda)\odot\mathcal{F}(U)||_{2}}_\text{under-fitting}+\gamma\underbrace{|\frac{\#\mathcal{T}_{\mu}\Lambda}{\#\mathcal{F}(\mathcal{T}_{\mu}\Lambda\odot U)+\varepsilon}-1|}_\text{over-fitting},~ \Lambda\in\mathscr{L},\mu\in\mathbb{C}. \label{eval}
\end{align}
The optimal candidate has minimal energy. In~(\ref{eval}), $U$ denotes the original image, $\#\cdot$ represents counting the number of particles, and $\odot$ is element-wise multiplication of matrices. The difference from the subtraction is truncated to non-negative parts. $\gamma>0$ is a penalty coefficient (we set $\gamma =10$), and $\varepsilon>0$ is a small constant to avoid division by $0$ (we set $\varepsilon=1\times 10^{-8}$). 

Note that~(\ref{eval}) is closely related to the energy~(\ref{varprob3}). The first component in~(\ref{eval}) measures the portion of particles not covered by the lattice candidate, i.e., the under-fitting. A smaller value means that more particles in the image lie on the lattice points of $\mathcal{T}_{\mu}\Lambda$. We normalize the remainder $U-\mathcal{T}_{\mu}\Lambda$ to make it comparable with $\mathcal{F}(U)$. The element-wise multiplication with $\mathcal{F}(U)$ prevents new points generated from incomplete particle extraction. The second term in~(\ref{eval}) compares the ratio between the number of lattice points of the candidate, and the slots filled with particles from the image, i.e., the over-fitting. 

Step 3, Figure~\ref{LISA}(e), is similar to a sampling procedure with replacement. This step is optional, yet when the number of underlying lattices is large, it improve the accuracy of identification. As to be shown in subsection~\ref{depLISA}, superposing lattices complicates the power spectrum, hence early identification is affected the most. Incorrect early extraction further yields unstable identification for the remaining lattices. This optional step correct these aspects, and proceeds iteratively. We first set $t=1$, and $\mathcal{T}_{\mu_{j}}\Lambda_{j}^{(1)}=\mathcal{T}_{\mu_{j}}\Lambda_{j}$.  For $t=1,\cdots,K$, from the optimal lattice candidate $\mathcal{T}_{\mu_{j}}\Lambda_{j}^{(t)}$, compute the remainder $\mathcal{F}(U-\mathcal{T}_{\mu_{j}}\Lambda_{j}^{(t)})$ as in (\ref{Hdef}). Then iterate Step 1 and Step 2 on $\mathcal{F}(U-\mathcal{T}_{\mu_{j}}\Lambda_{j}^{(t)})$ to find the next optimal lattice $\mathcal{T}_{\mu_{j}}\Lambda_{j}^{(t+1)}$.  This is the red boxed one in Figure~\ref{LISA}(e). Update $\mathcal{T}_{\mu_{j}}\Lambda_{j}=\arg\min_{t=1,\cdots,K} \mathcal{E}(\mathcal{T}_{\mu_{j}}\Lambda_{j}^{(t)})$.  This optimal one is Figure~\ref{LISA}(f). 

In Step 4, the optimal candidate $\mathcal{T}_{\mu_{j}}\Lambda_{j}$ is subtracted from the original image, and the difference is truncated, only non-negative values remain. Figure~\ref{LISA}(g) shows the remainder.  We compute the average intensity of the residual $U-\mathcal{T}_{\mu_{j}}\Lambda_{j}$. Insufficient intensity terminates the algorithm (the threshold is set to be $0.01$); otherwise, we preprocess the residual using $\mathcal{F}$ replacing the original image and then repeat Step 1--4.

\subsection {Analytical properties of LISA: Superlattice and Spectrum Surface} \label{depLISA}

We describe the close relation between LISA and geometric features of the superlattice. Assuming no remainder term in the image representation~(\ref{imagemodel}), the Fourier transform of a superlattice image is:
 \begin{align}
 \hat{U}(\xi)=\hat{G}_{\sigma}(\xi)\sum_{j=1}^{N}\frac{\Lambda^{*}_{j}(\xi)}{\det\Lambda_{j}}\exp(-i2\pi\xi\cdot\mu_{j})\label{duallattice},~\xi\in\mathbb{R}^{2},
 \end{align}
 where $\xi$ represents the frequency coordinate, $\det\Lambda_{j}$ is the fundamental volume of $\Lambda_{j}=\Lambda\langle\beta_j,\rho_j\rangle$ computed by $\text{Im}(\overline{\beta_{j}}\beta_{j}\rho_{j})$, and $\Lambda_{j}^{*}$ denotes the reciprocal lattice impulse on the frequency domain corresponding to $\Lambda_{j}$, which is expressed in the lattice space by: 
\begin{align*}
 	\Lambda_{j}^{*} =[\hat{\beta}_{j},\hat{\rho}_{j}]:=[\frac{\beta_{j}\exp(-i\pi/2)}{\det\Lambda_{j}},\rho_{j}]\in\mathscr{L}.
\end{align*} 


Formula~(\ref{duallattice}) implies that the Fourier transform of a superlattice image is a mixture of complex lattices modified by three factors: the centered Gaussian $\hat{G}_{\sigma}$, the fundamental volumes $\text{det}\,\Lambda_{j}$, and the translations of the original lattices $\mu_{j}$, $j=1,2,\dots,N$. Without these modulations, every lattice with two basis vectors in the image corresponds to two peaks on the power spectrum. Notice that $\Lambda_{j}^{*}(\xi)=1$ if and only if $1/|\xi|$ is a period of $\Lambda_{j}$ along the direction of $\xi$, $\forall\xi\in\mathbb{R}^{2}$ and $\forall j =1,\cdots,N$. Hence, lattices $\{\Lambda_{j}\}_{j=1}^{N}$ can be identified with correct combinations of the peaks.\par
Gaussian PSF and fundamental volumes of lattices complicate the problem. First, independent of the positions of the superlattice particles, a centered Gaussian $\hat{G}_{\sigma}$ globally dampens the power spectrum. If $|\xi|$ is small, $\hat{G}_{\sigma}(\xi)$ has little influence on the power spectrum, and if $|\xi|$ is large, $\hat{G}_{\sigma}(\xi)$ decreases the value at $\xi$. Second, the radius of particles controls the rate of radial decay of the power spectrum surface. Large frequency components are preserved if the particles of the superlattice have a small radius, as the standard deviation $\sigma$ is small. Third, fundamental volumes of the original lattices affect the power spectrum. The magnitudes of a pair of peaks on the spectrum surface associated with the lattices with smaller fundamental volumes are augmented, and those with larger fundamental volumes are decreased. This coincides with our experience that denser lattices are easier to be recognized. LISA tends to find lattices with smaller particles and smaller fundamental volumes.\par
Relative translations of the lattice layers have a more delicate influence on the power spectrum surface. Translation in spatial domain results in a phase change in the frequency domain, and it has no effect on the power spectrum if there is only one lattice. When multiple lattices are superposed, frequencies along the same direction will interact with each other. Suppose for some $1< m\leq N$, $\Lambda^{*}_{1}(\xi)=\cdots=\Lambda^{*}_{m}(\xi)=1$ and $\Lambda^{*}_{j}(\xi)=0$ for $j=m+1,\cdots,N$, then $\hat{U}(\xi)$ is a sum of $m$ complex numbers, whose magnitude varies based on directions of $\mu_{1},\cdots,\mu_{m}$. An extreme case is that, if $\Lambda_{1}=\Lambda_{2}$, $\mu_{1}=-\mu_{2}\neq 0$, and there exists an $\xi$ such that $\Lambda_{1}^{*}(\xi)=1$ and $\xi\cdot\mu_{1}\neq 0$, then $|\hat{U}(\xi)|=0$. See Figure~\ref{relative} for an example. LISA  detects potential lattices, even though the reciprocal lattices are incomplete and the reciprocal bases are not minimal. If any basis of the reciprocal lattice remains high response in the power spectrum, LISA will consider it as a candidate to be evaluated. 
\begin{figure} 
 \begin{center}
\begin{tabular}{cccc}
(a) & (b)  & (c)  & (d) \\
    \includegraphics[height=1.3in]{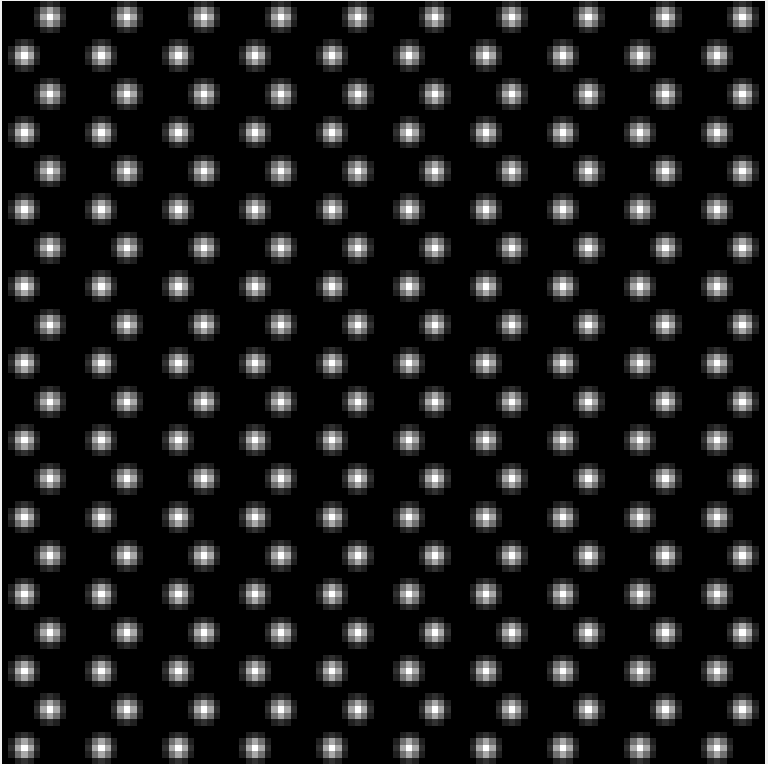} &
    \includegraphics[height=1.3in]{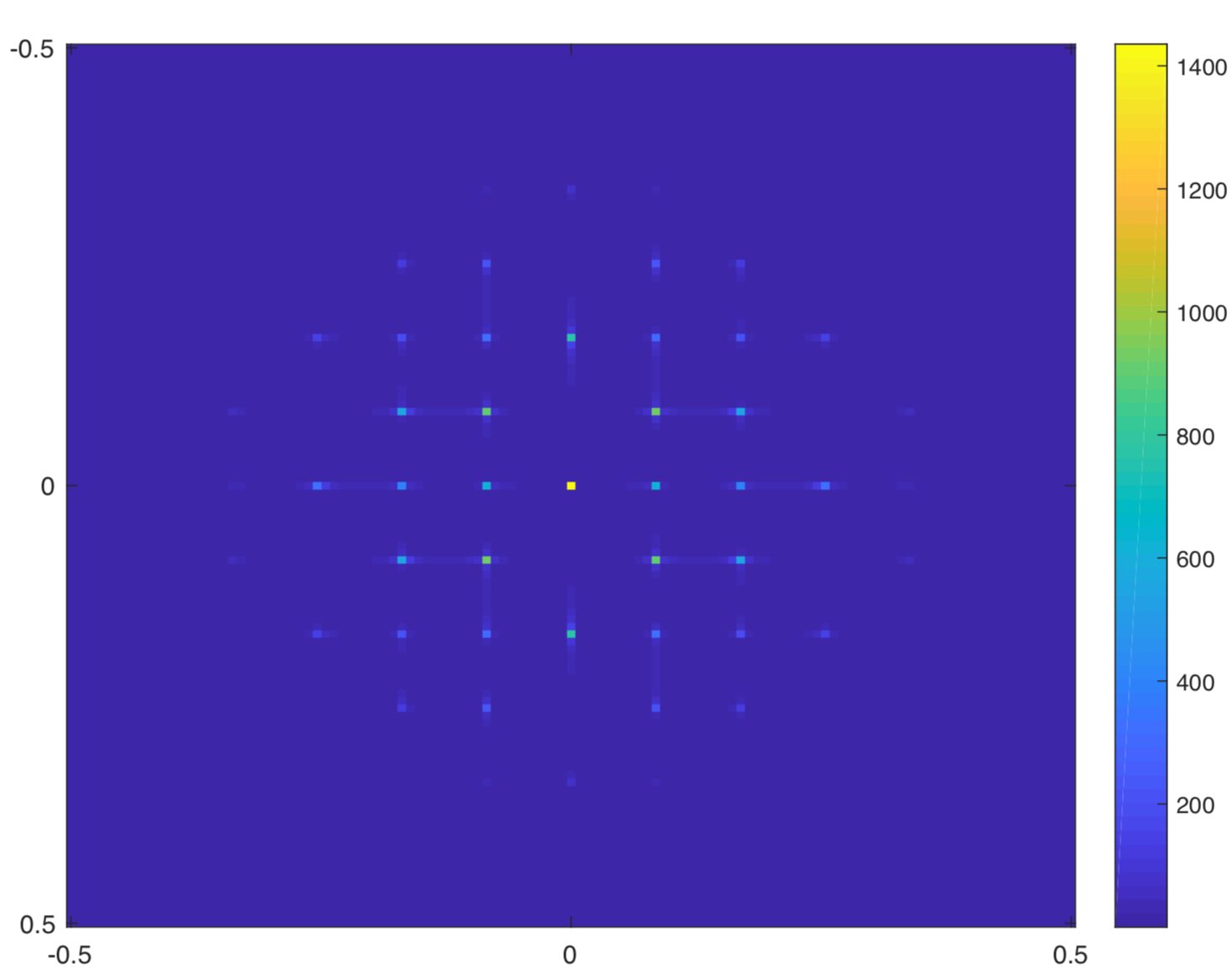} & 
    \includegraphics[height=1.3in]{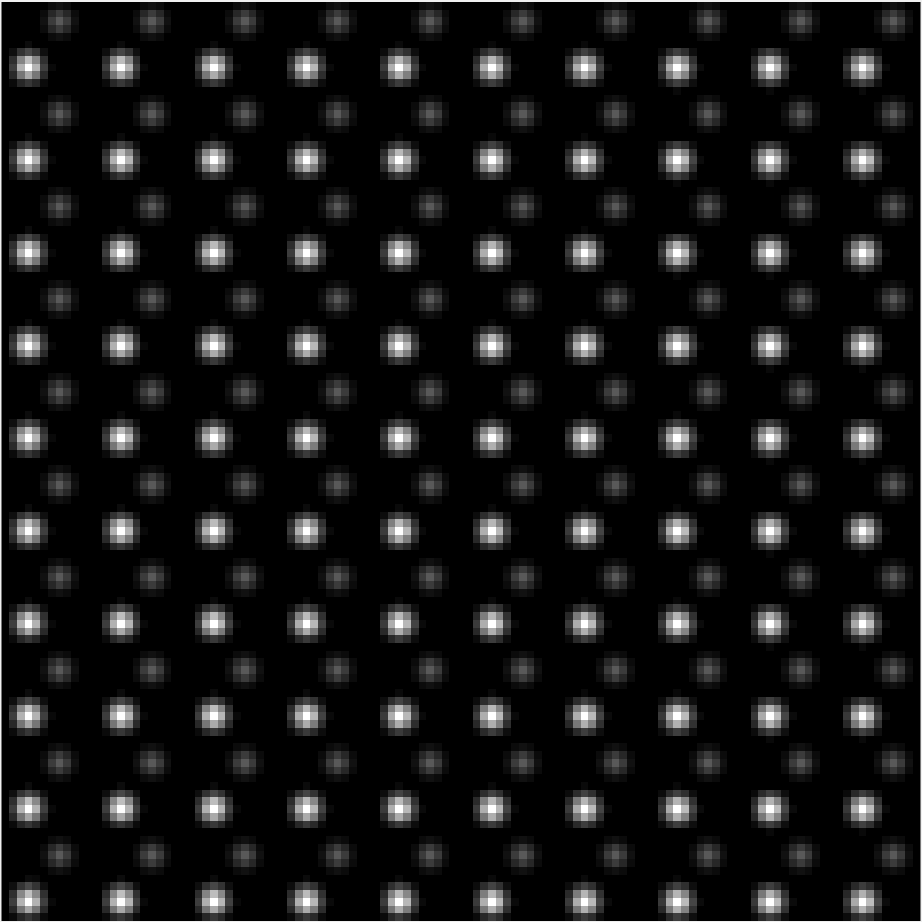} &
    \includegraphics[height=1.3in]{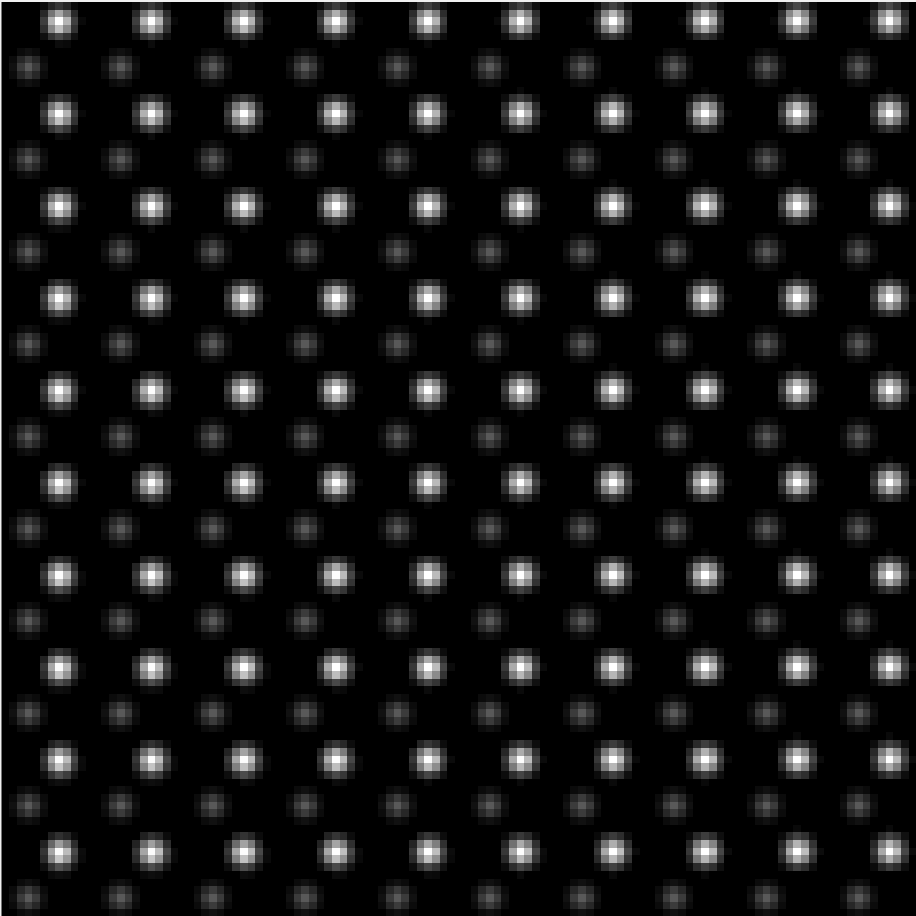} 
 \end{tabular} 
 \end{center}   
  \caption{[Relative translation affects spectrum surface] Two lattices $\mathcal{T}_{4-3i}\Lambda\langle 12,i\rangle$ and $\mathcal{T}_{-4+3i}\Lambda\langle 12,i\rangle$ are superposed in (a). Due to the relative translations, there are missing peaks in the spectrum surface as shown in (b). From this incomplete reciprocal lattice, LISA recovers the lattices in (c) and (d).}\label{relative}
\end{figure}

\subsection{Robustness of LISA against Gaussian Perturbation}

Lattice distorted by Gaussian perturbation can model the atomic configuration in crystal-melt interface~\cite{crystalmelt}.  We modify (\ref{duallattice}) to consider such cases. For simple notations, only one unshifted lattice is assumed, and it is easy to extend to the general formula for multiple lattices with translations. Ignoring the remainders in~(\ref{imagemodel}), the image $\tilde{U}$ of a lattice $\mathcal{T}_{0}\Lambda(b_{1},b_{2})$ with perturbed particles can be expressed as:
\begin{align*}
&\tilde{U}(x,y)=\sum_{k_{1},k_{2}\in\mathbb{Z}}G_{\sigma}*\delta(k_{1}b_{1}+k_{2}b_{2}+\Delta x_{k_{1},k_{2}}+i\Delta y_{k_{1},k_{2}}-x-iy),
\end{align*}
with $(\Delta x_{k_{1},k_{2}},\Delta y_{k_{1},k_{2}})\in\mathbb{R}^{2}$ denoting the perturbation on the particle parameterized by $(k_{1},k_{2})$ in the lattice. The Fourier transform of $\tilde{U}$ is
\begin{align*}
\hat{G}_{\sigma}(\xi)\sum_{k_{1},k_{2}\in\mathbb{Z}}\exp(-2\pi i\phi_{k_{1},k_{2}}(\xi)),
\end{align*}
where $\phi_{k_{1},k_{2}}(\xi)=(\Delta x_{k_{1},k_{2}}+i\Delta y_{k_{1},k_{2}}+k_{1}b_{1}+k_{2}b_{2})\cdot\xi$. We assume that the perturbations are independent and identically distributed Gaussian vectors with  uncorrelated coordinates, that is  $(\Delta x_{k_{1},k_{2}},\Delta y_{k_{1},k_{2}})\sim\mathcal{N}(0,\Sigma)$ where $\Sigma =\begin{bmatrix}
s^{2}&0\\
0& s^{2}	
\end{bmatrix}
$, $s>0$ constant, $\forall (k_{1},k_{2})\in\mathbb{Z}^{2}$. This implies that for any $\xi$ in the frequency domain, 
\begin{align*}
\phi_{k_{1},k_{2}}(\xi)\sim\mathcal{N}((k_{1}b_{1}+k_{2}b_{2})\cdot\xi,s^{2}|\xi|^{2}).
\end{align*}
Some observations are immediate. First, for a single lattice, perturbations only alter the phases. Interactions among multiple lattices can still change the magnitude of the power spectrum as discussed in subsection~\ref{depLISA}. Second, $\mathbb{E}[\phi_{k_{1},k_{2}}(\xi)]$ depends on the angle between $k_{1}b_{1}+k_{2}b_{2}$ and  $\xi$. In particular, perturbations have a stronger effect on non-lattice points than lattice points. If $\xi$ is reciprocal to the lattice point $k_{1}b_{1}+k_{2}b_{2}$, then they are perpendicular, thus the average perturbation is $0$. Finally, with fixed $s$, the standard deviation of $\phi_{k_{1},k_{2}}(\xi)$ only depends on $|\xi|$. When $|\xi|$ is small, or equivalently when we are approximating relatively large periods, the Fourier transform of the perturbed lattice is almost the same as that of the unperturbed one. In Figure~\ref{gaussian}, lattice points are shifted by Gaussian perturbation with various standard deviations, and the low-frequency components maintain high responses. LISA is robust against Gaussian perturbation with bounded standard deviation, and the detection of medium-sized-lattices is effective. When the standard deviation is large, LISA identifies the correct lattices, yet the extraction may need other techniques, e.g., the nearest particles to the lattice candidate are identified.
\begin{figure}
 \begin{center}
\begin{tabular}{ccc}
(a) & (b)  & (c) \\
    \includegraphics[scale=0.21]{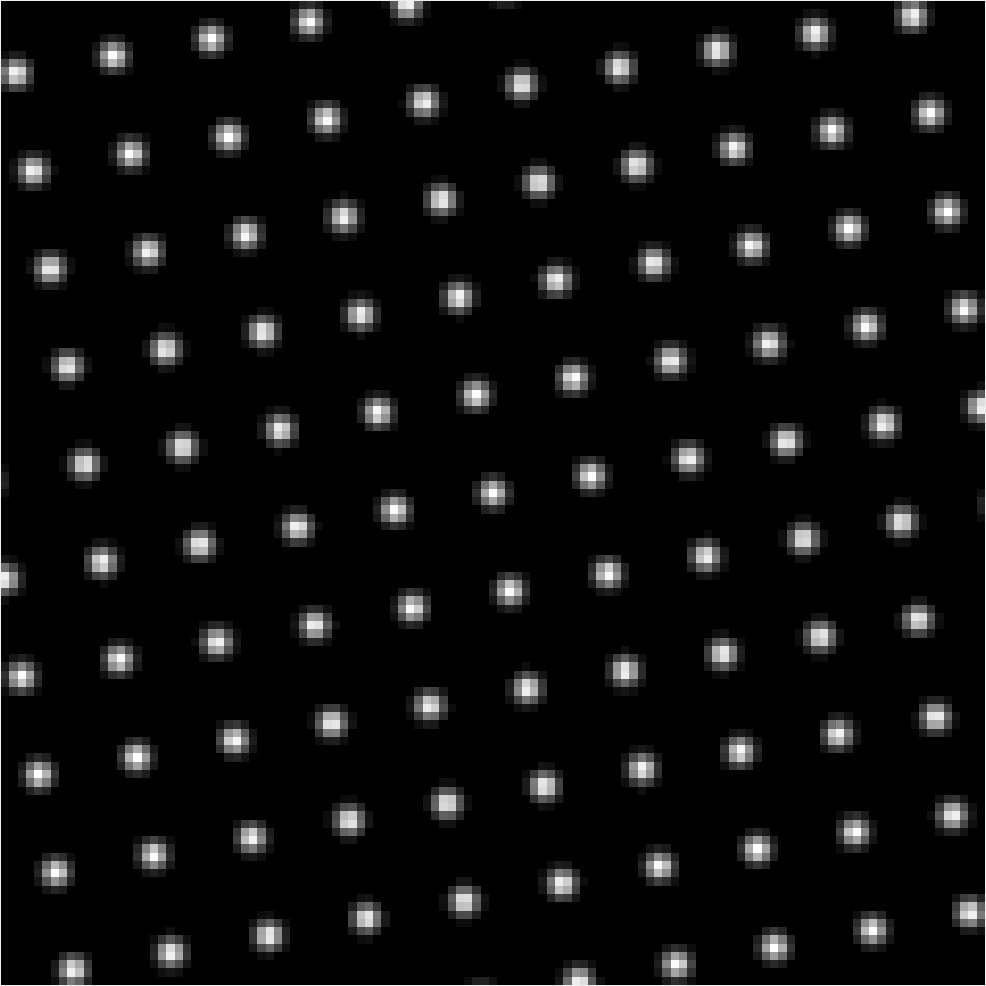} &
    \includegraphics[scale=0.2]{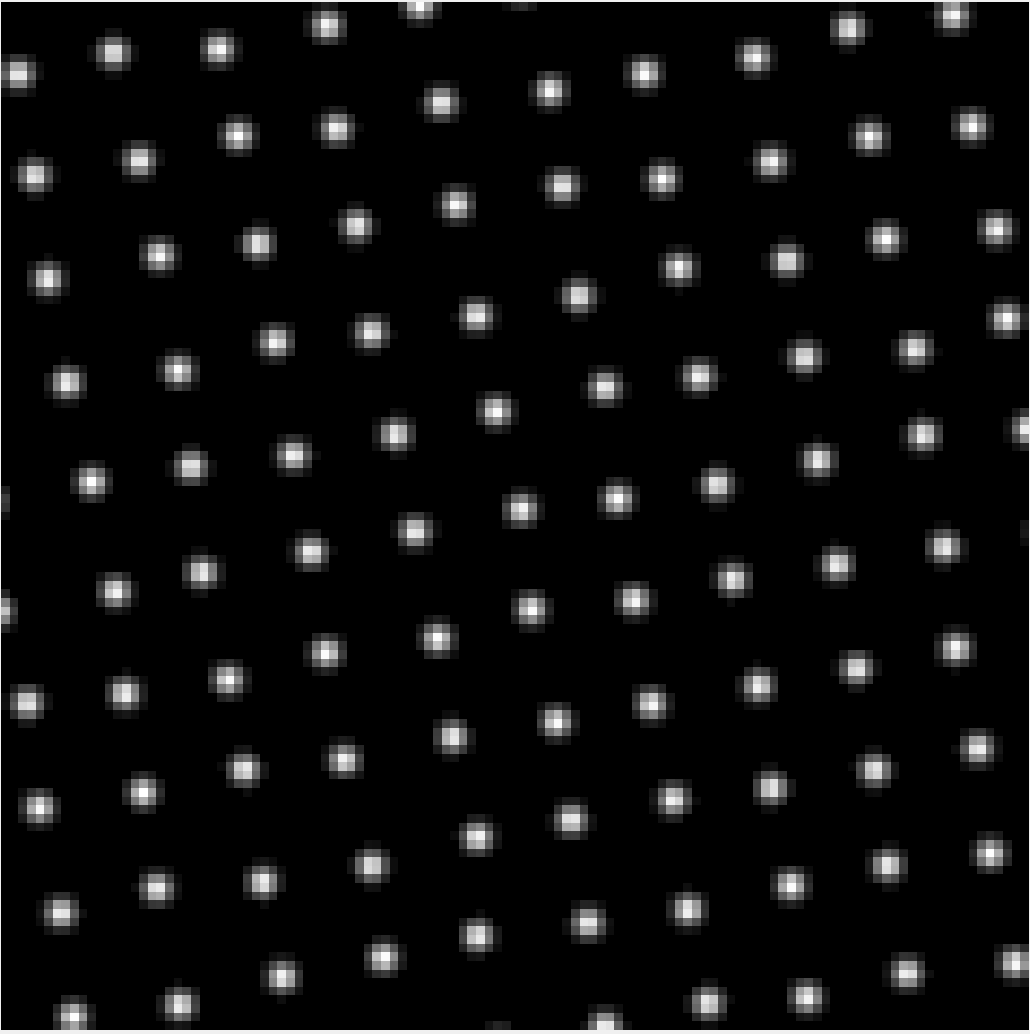} &
    \includegraphics[scale=0.21]{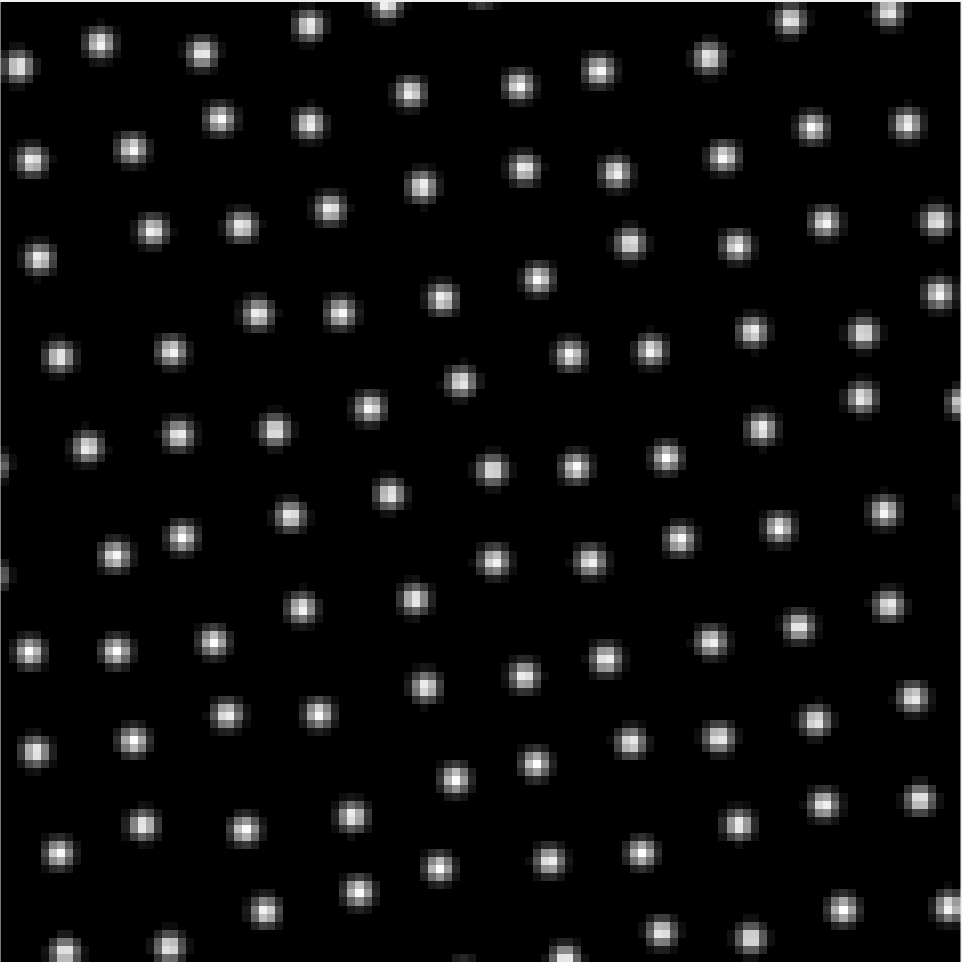} \\
(d) & (e) & (f) \\
    \includegraphics[scale=0.16]{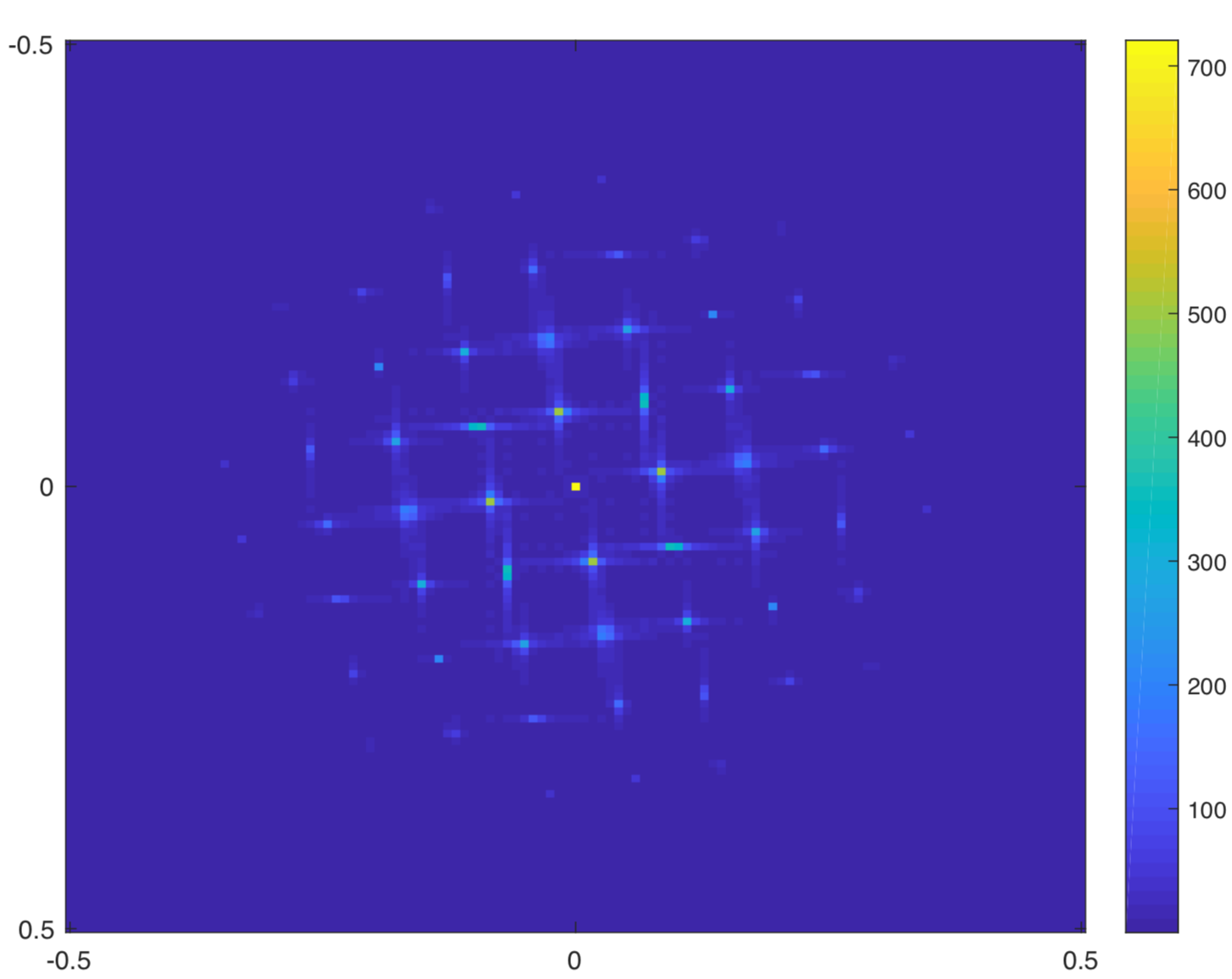} &
    \includegraphics[scale=0.16]{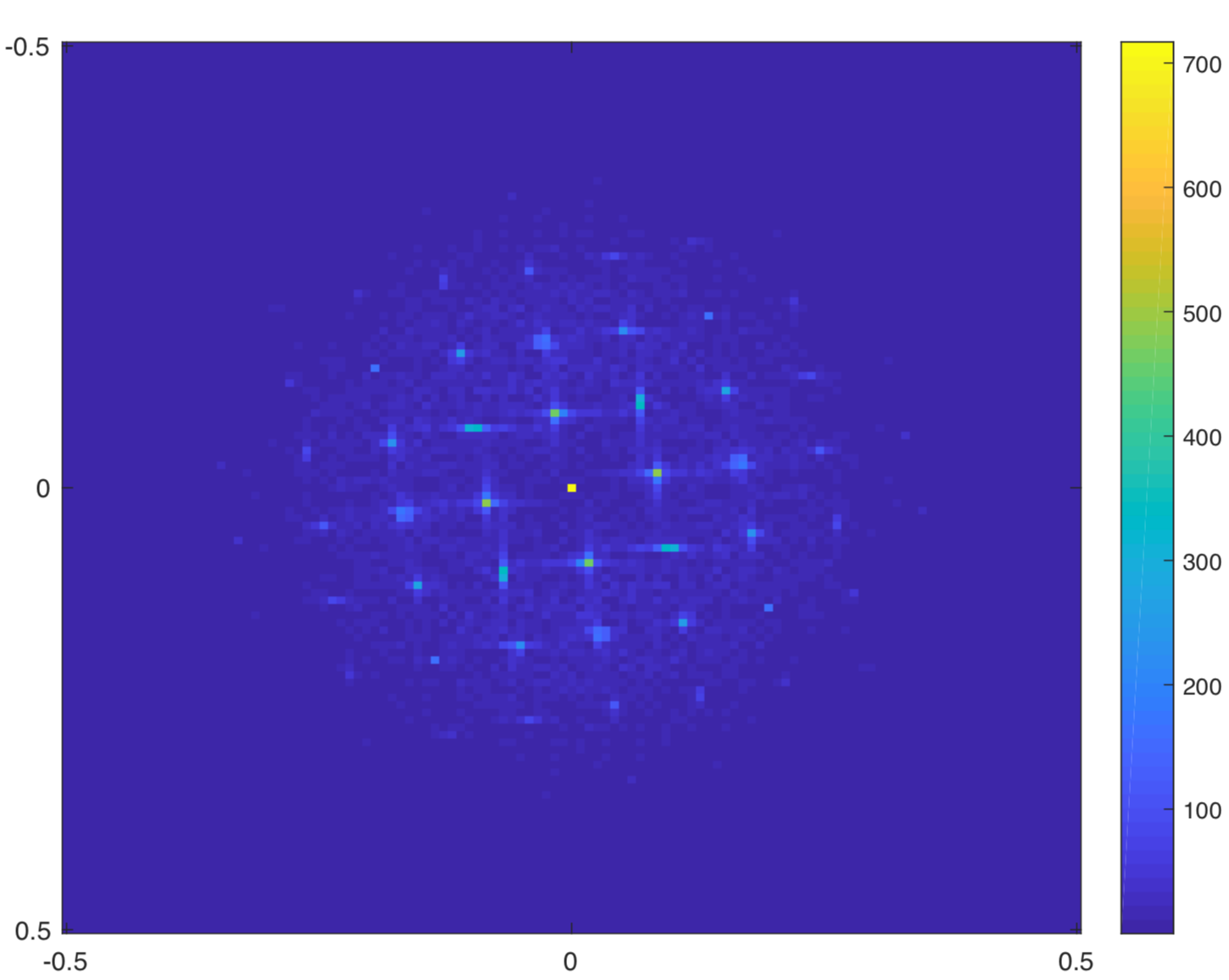} &
    \includegraphics[scale=0.16]{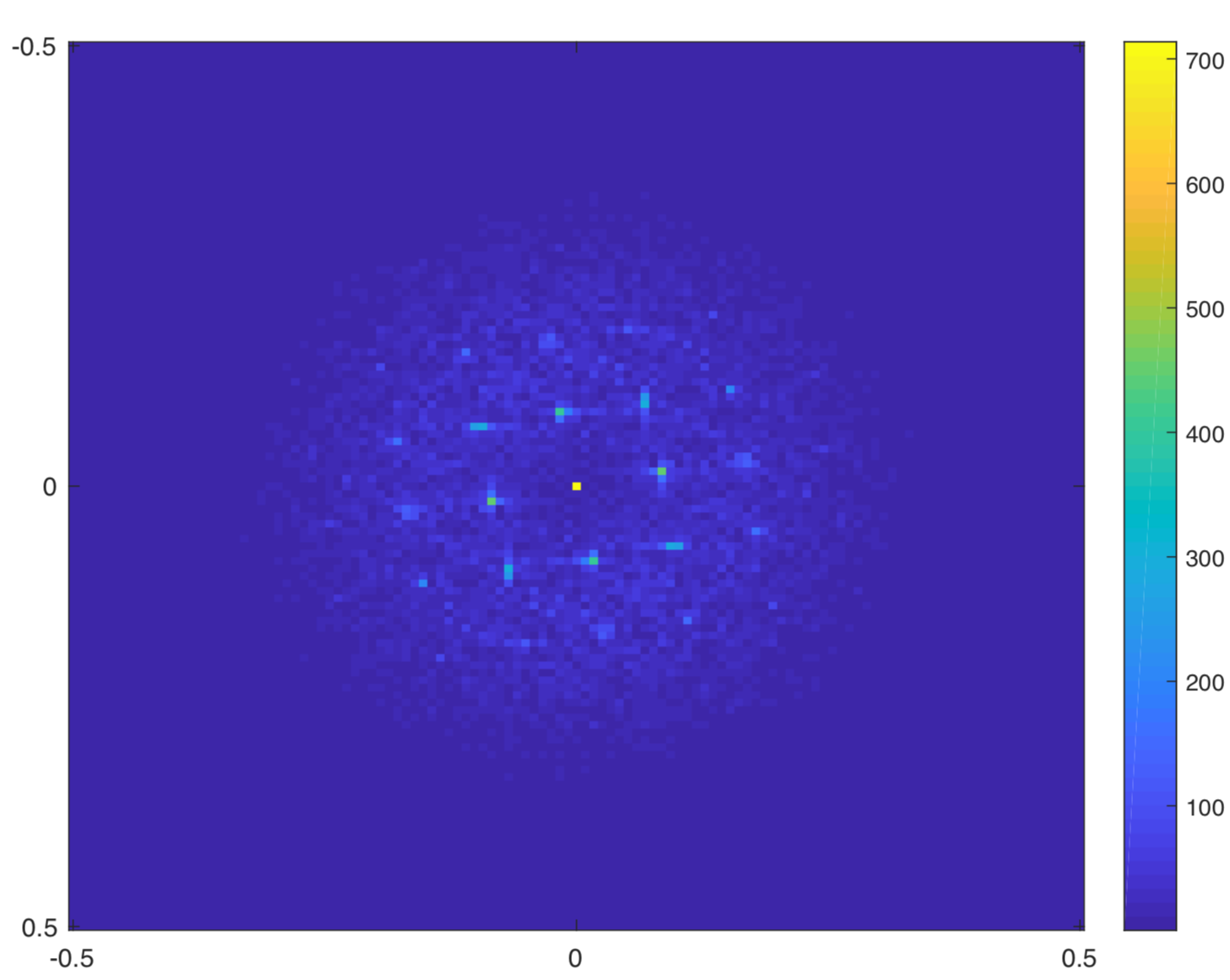} 
 \end{tabular}
 \end{center}
 \caption{[LISA is robust against Gaussian perturbation] In the first column, a single lattice $\mathcal{T}_{0}\Lambda\langle 12,e^{i\pi/18}\rangle$ is shown in (a) with its power spectrum surface in (d). A centered Gaussian perturbation is applied with standard deviation (b) $s = 0.5$ and (c) $s = 1$, and their power spectrums are displayed in (e) and (f) respectively. Notice that in the frequency domain, the reciprocal bases away from the origin are smeared by noises, but those near the origin remain high responses. The first lattices identified by LISA in (b) and (c) are robust against the perturbation, and their distances to (a) are $0.0046$ and  $0.0081$ respectively.}\label{gaussian}
\end{figure} 
\section{Numerical Experiments of LISA} \label{sec:num}

We present various numerical results.  
The radius of each particle is around $2.5\sim3$ pixels, and we choose $2.7$, for visualization.  The performance of LISA is evaluated both visually and numerically by computing the distances between the identified and the real patterns in the lattice space. For the choice of parameters, we fix $J=6$ and $K=10$. 
 
Figure~\ref{syn1} shows a typical example of LISA. The given image is a superlattice composed of three distinguishable lattices. LISA successfully extracts all the underlying lattices, one after another. For better comparison, results (b)--(d) display each identified lattice (bright white) overlaid on a dimmer original image in (a). For each layer, the identified lattice and true lattice shows a small difference.
\begin{figure}
\begin{center}
\begin{tabular}{cccc}
(a) Original image& (b) $d_{\mathscr{L}}(\Lambda,\hat{\Lambda})=0.0044$ & (c) $d_{\mathscr{L}}(\Lambda,\hat{\Lambda})=0.0093$ & (d) $d_{\mathscr{L}}(\Lambda,\hat{\Lambda})=0.0586$  \\
    \includegraphics[width=1.4in]{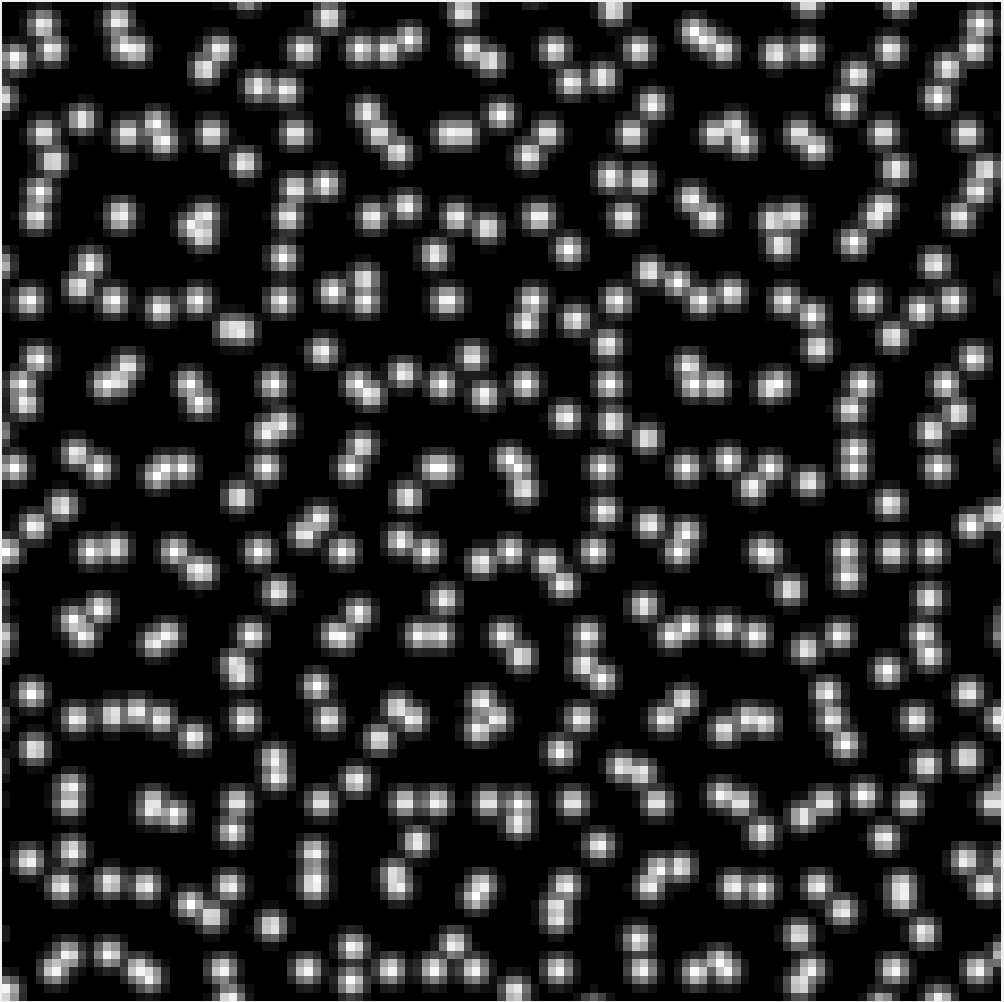}  &
    \includegraphics[width=1.4in]{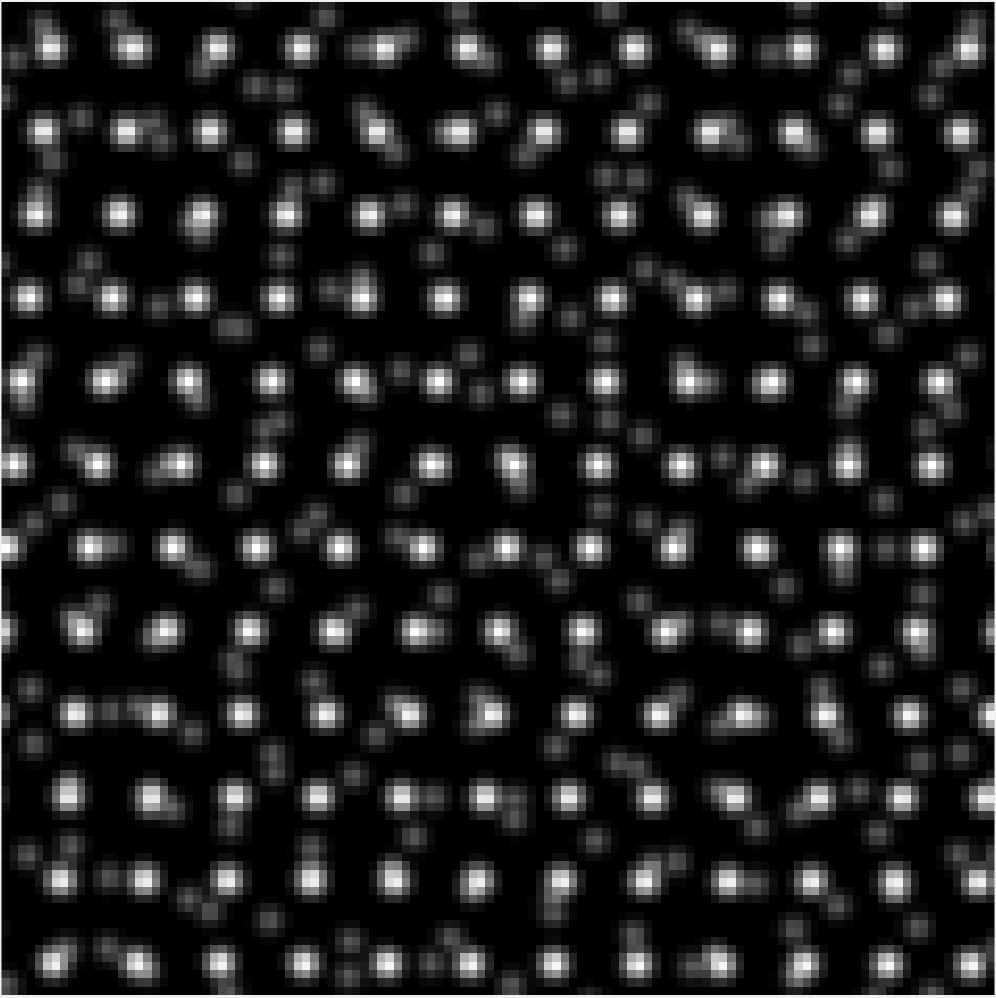} & 
    \includegraphics[width=1.4in]{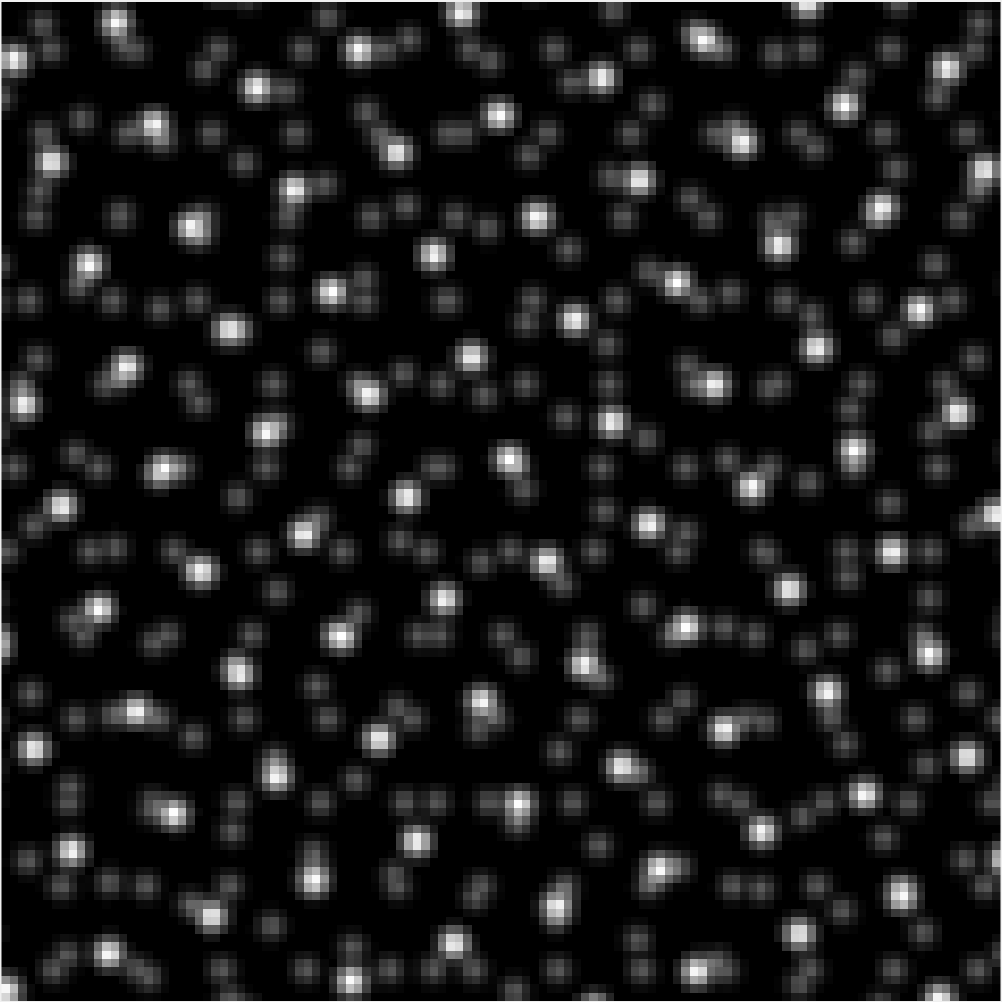} &
    \includegraphics[width=1.4in]{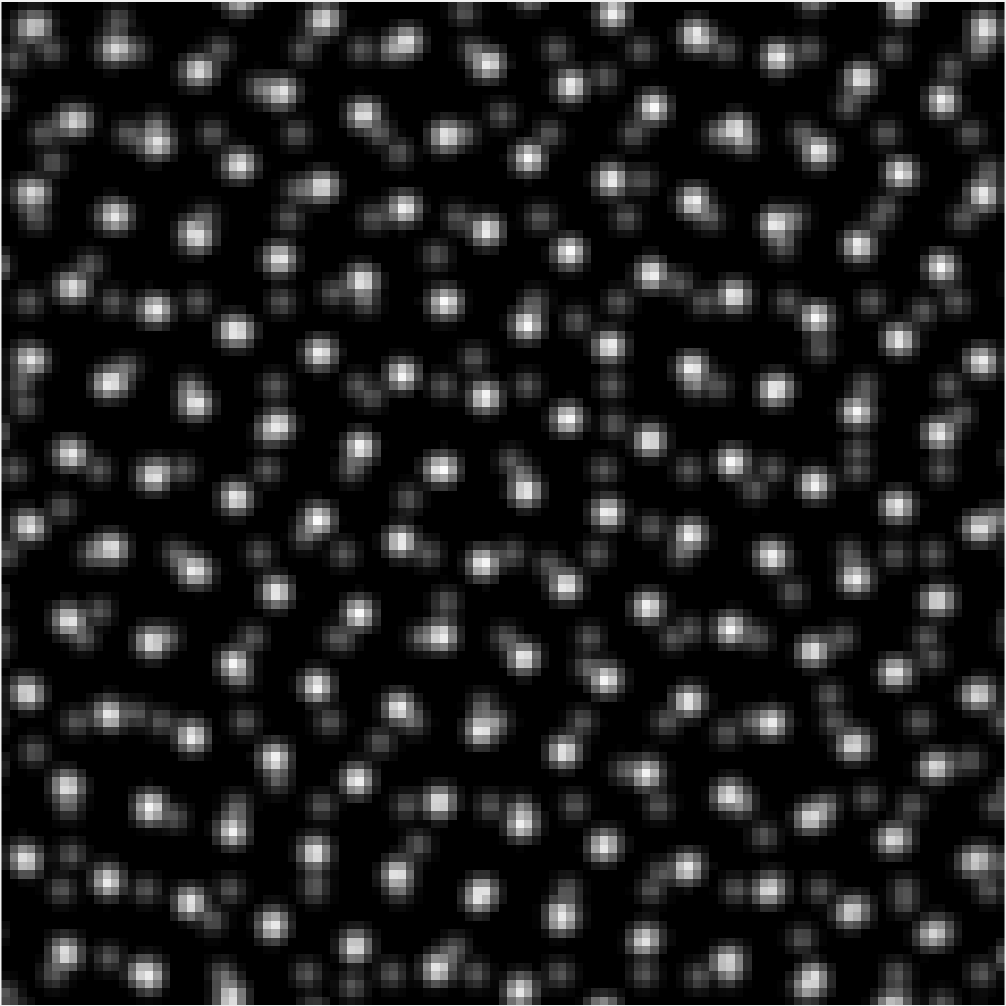} \\
\end{tabular} 
 \end{center}
\caption{[Typical example of LISA] Identification of three lattices in (a). (b) is the extracted pattern for $\mathcal{T}_{2-4i}\Lambda\langle -9.9927 + 0.0315i, 1.0014e^{i17\pi/36}\rangle$, (c) is the extracted pattern for $\mathcal{T}_{-7-4i}\Lambda\langle-4.4820 +12.1815i,i\rangle$ . (d) is the extracted pattern for $\mathcal{T}_{1-5i}\Lambda\langle-4.9898-8.5389i, 1.0298e^{i7\pi/12}\rangle$. The metric value above the images shows the comparison with the true lattice and the recovery. }\label{syn1}
\end{figure}

Figure~\ref{syn2} shows a more complicated mixture where the given image (a) seems almost random. 
The more layers of lattices there are, the more complicated the separation becomes.  Randomly clustered particles,  curve-like segments, and highly inhomogeneous texton regions present challenges. 
LISA identified five different lattice patterns from image (a), without any prior knowledge, and the recovered lattice patterns show high precision.  The metric values below the images show the comparison with the true lattices, which are all less than $0.03$.   Also, notice that the identified lattice patterns (c) and (f) are very similar. Their lattice distance in the lattice space is $0.0340$. LISA  is able to distinguish small differences since in the power spectrum surface, periodic structures are more easily identified as strong responses. 
\begin{figure}
\begin{center}
\begin{tabular}{cccc}
(a) Original image& (b) $d_{\mathscr{L}}(\hat{\Lambda},\Lambda)=0.0224$ & (c) $d_{\mathscr{L}}(\hat{\Lambda},\Lambda)=0.0053$\\
\includegraphics[scale=0.2]{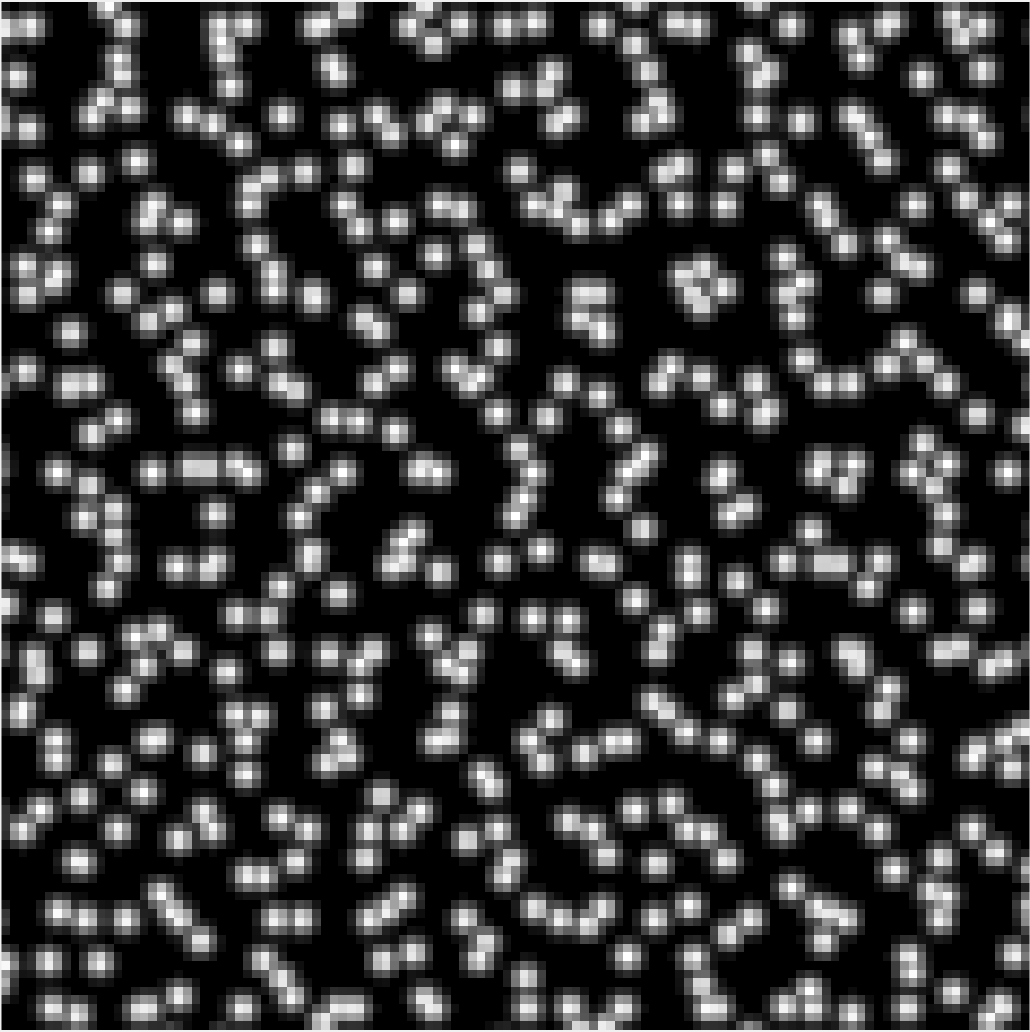} &
\includegraphics[scale=0.2]{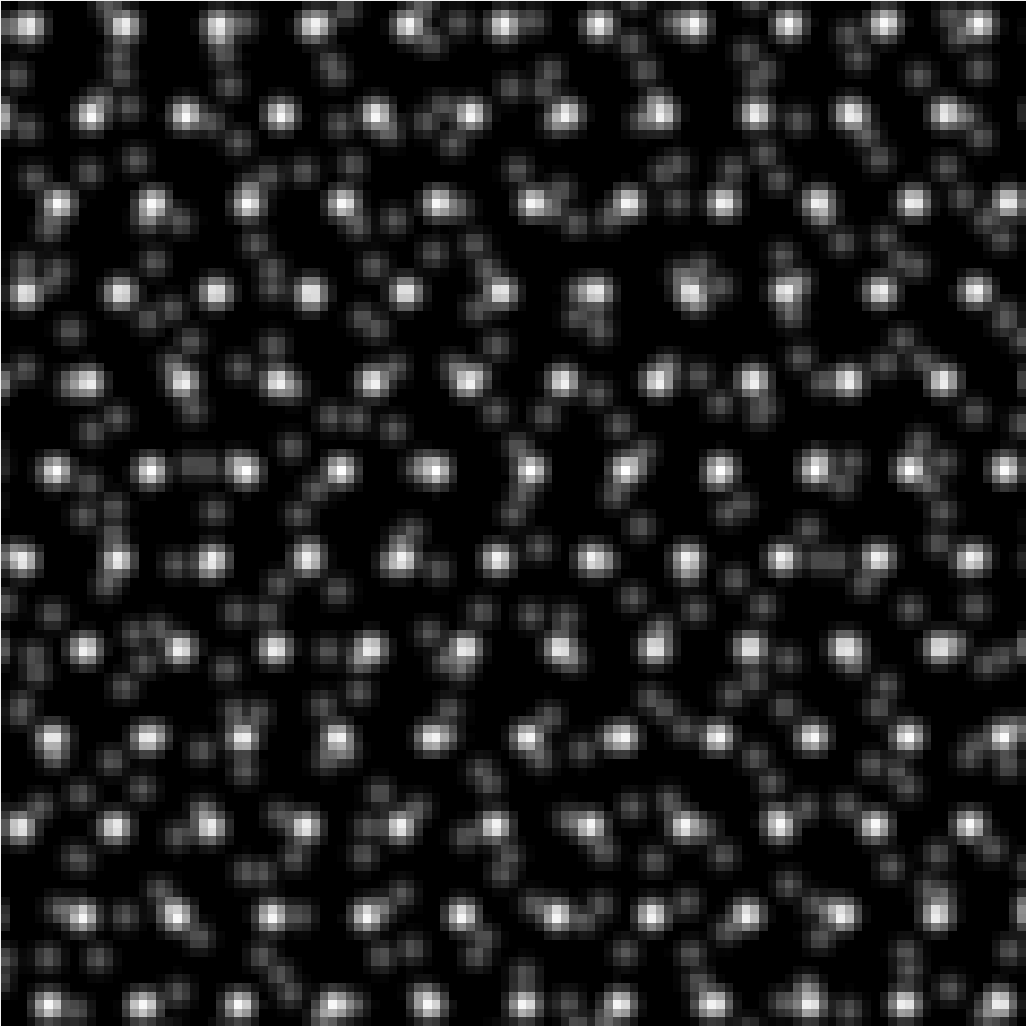} & 
\includegraphics[scale=0.2]{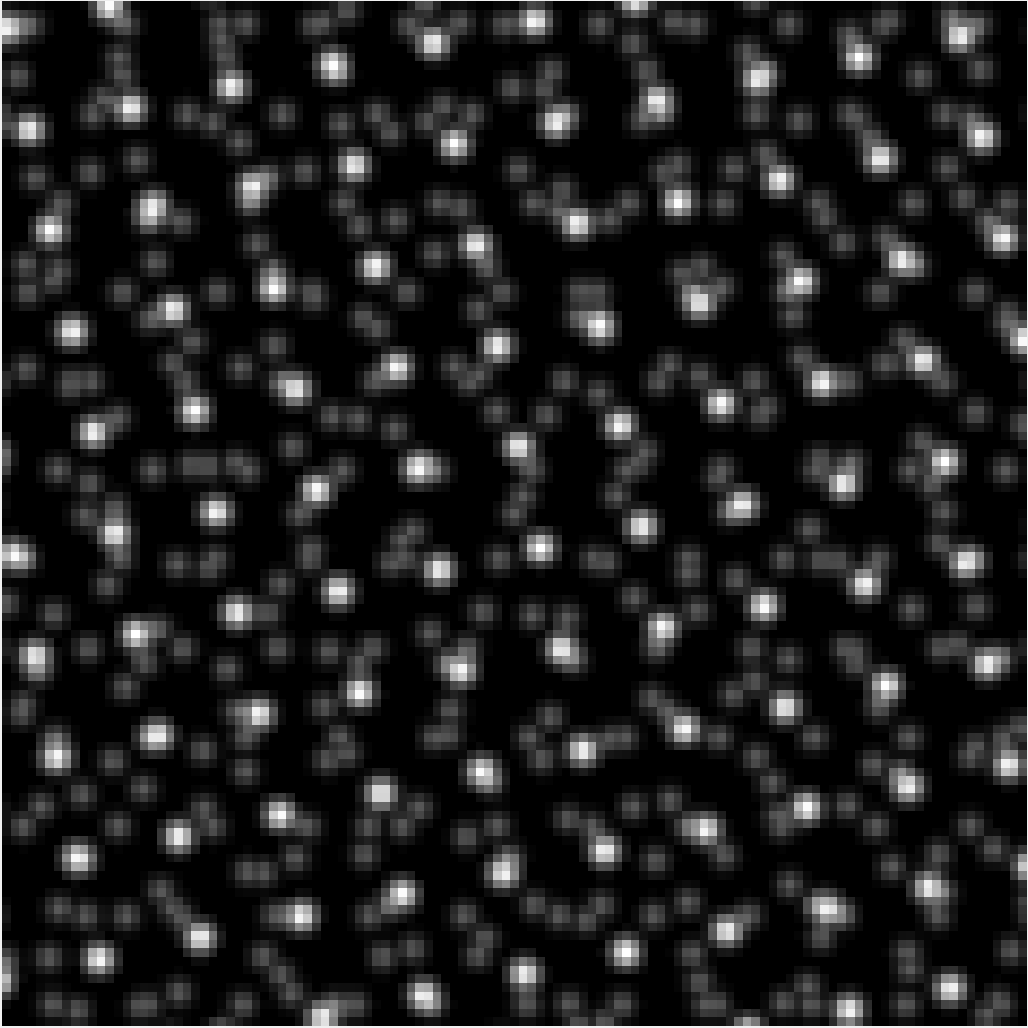} \\
(d) $d_{\mathscr{L}}(\hat{\Lambda},\Lambda)=0.0065$ & (e) $d_{\mathscr{L}}(\hat{\Lambda},\Lambda)= 0.0067$  & (f) $d_{\mathscr{L}}(\hat{\Lambda},\Lambda)=0.0025$  \\
     \includegraphics[scale=0.2]{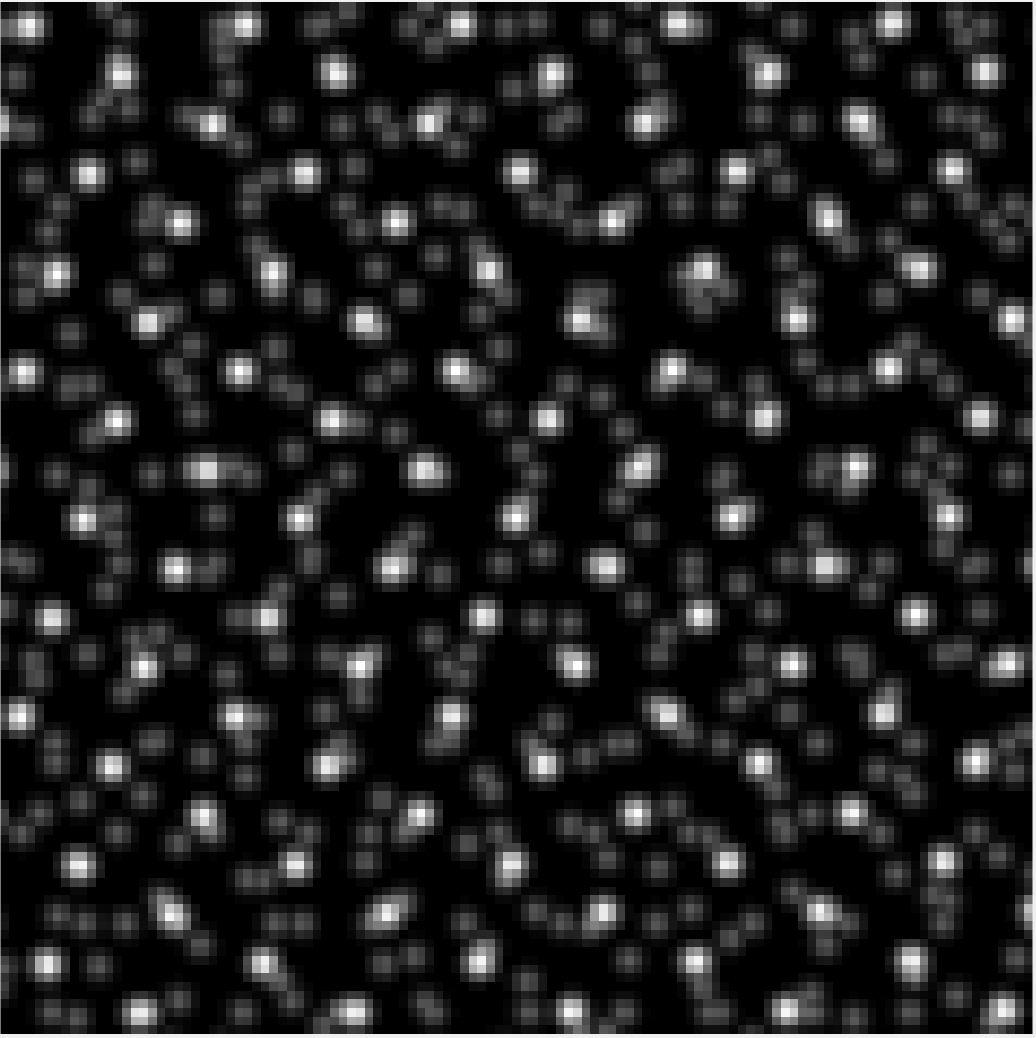} &
    \includegraphics[scale=0.2]{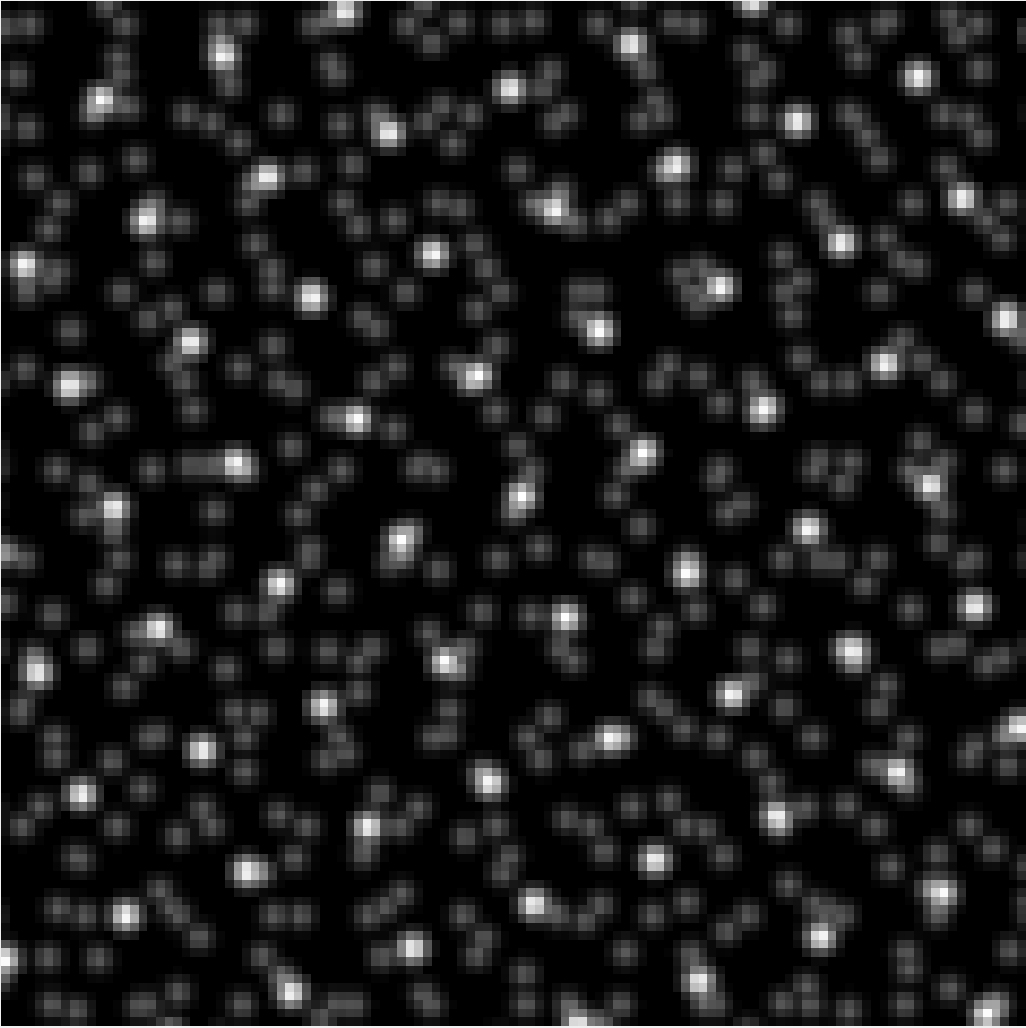} &
        \includegraphics[scale=0.2]{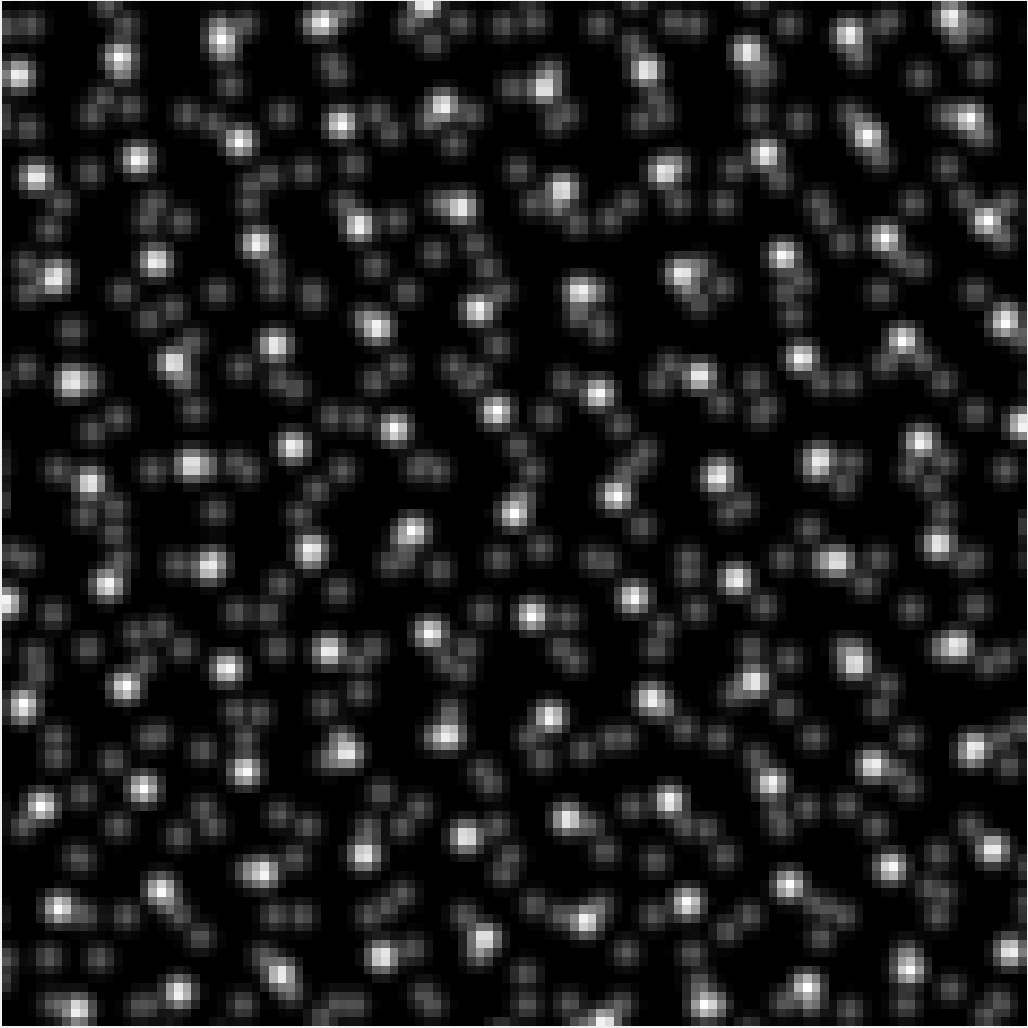}\\
\end{tabular} 
 \end{center}
\caption{Identification of five lattices from the given image (a) using LISA.  (b) The extracted pattern $\mathcal{T}_{2-5i}\Lambda
\langle 11, e^{i7\pi/18}\rangle$, (c)  $\mathcal{T}_{3+4i}\Lambda\langle 11.7378 + 2.4949i ,i\rangle$,  (d) $\mathcal{T}_{0}\Lambda\langle3.7082 +11.4127i,e^{4\pi/9}\rangle$, (e)  $\mathcal{T}_{1-2i}\Lambda\langle 14.0954 + 5.1303i,i\rangle$, and  (f)  $\mathcal{T}_{0}\langle 11.8177 + 2.0838i, i\rangle$.  LISA inspects the superlattice in frequency domain and avoids complexities in the image domain. Notice that all the metric value $d_{\mathscr{L}}(\hat{\Lambda},\Lambda)$ comparing the recovered lattice with the true lattice are very small. }
  \label{syn2}
\end{figure}

The new lattice representation and the metric are independent to the translation of lattice pattern. Figure~\ref{shiftexp} presents the effect of LISA concerning the translation.  There are two groups of lattices mixed in the given image (a), and  each group has two lattices differing from each other only by translation $\mu$.  By cross-correlation function (in Step 2 of LISA), the underlying four lattices are extracted sequentially.  We notice that even if particles are located close to each other, LISA is able to identify the underlying lattices. 
 \begin{figure}
\begin{center}
\begin{tabular}{ccc}
(a) Original image & & (b) $d_{\mathscr{L}}(\hat{\Lambda},\Lambda)=0.0406$ \\
    \includegraphics[scale=0.2]{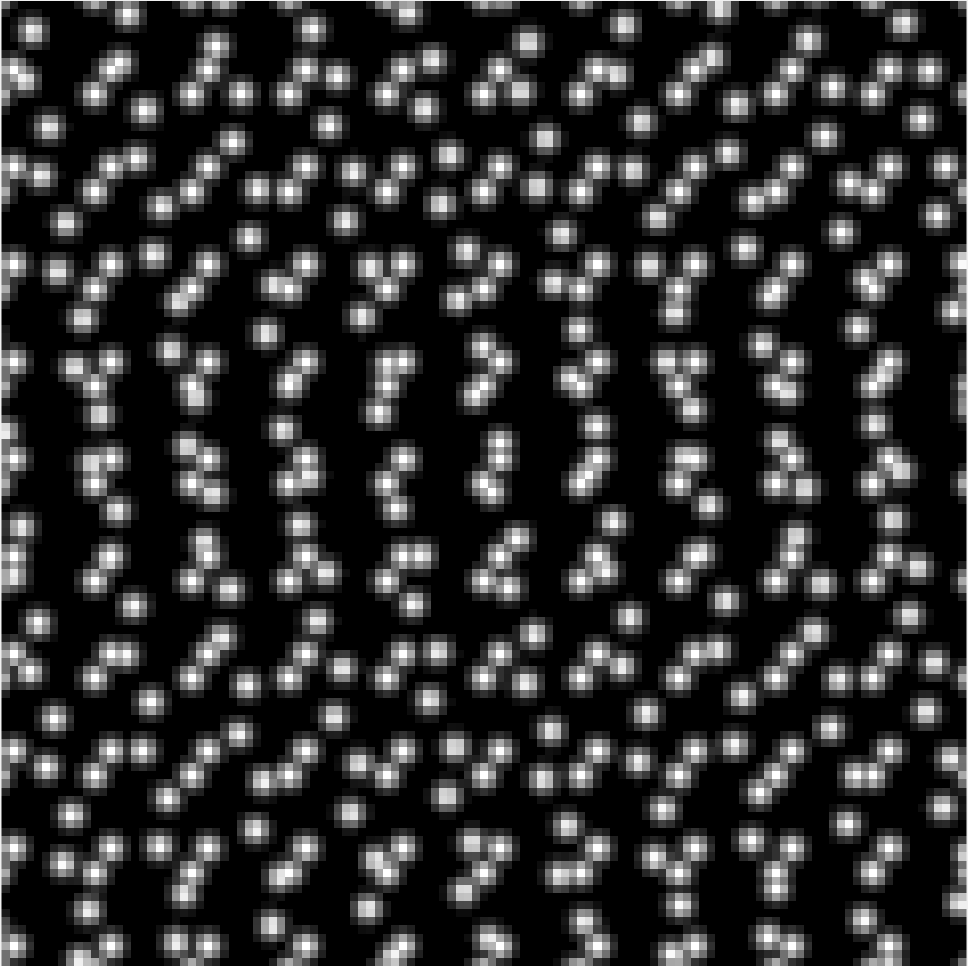} & &
    \includegraphics[scale=0.2]{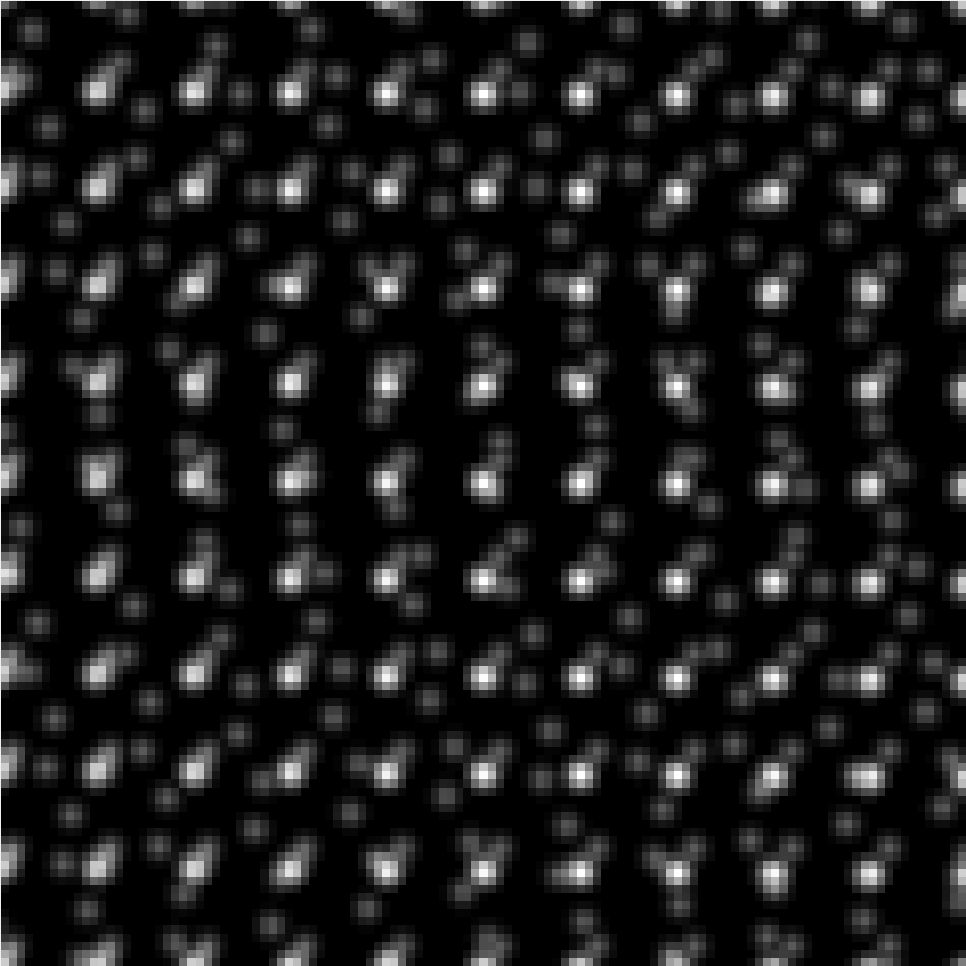} \\
(c) $d_{\mathscr{L}}(\hat{\Lambda},\Lambda)= 0.0095$ & (d) $d_{\mathscr{L}}(\hat{\Lambda},\Lambda)=0.0053$ &
(e) $d_{\mathscr{L}}(\hat{\Lambda},\Lambda)=0.0074$ \\
    \includegraphics[scale=0.2]{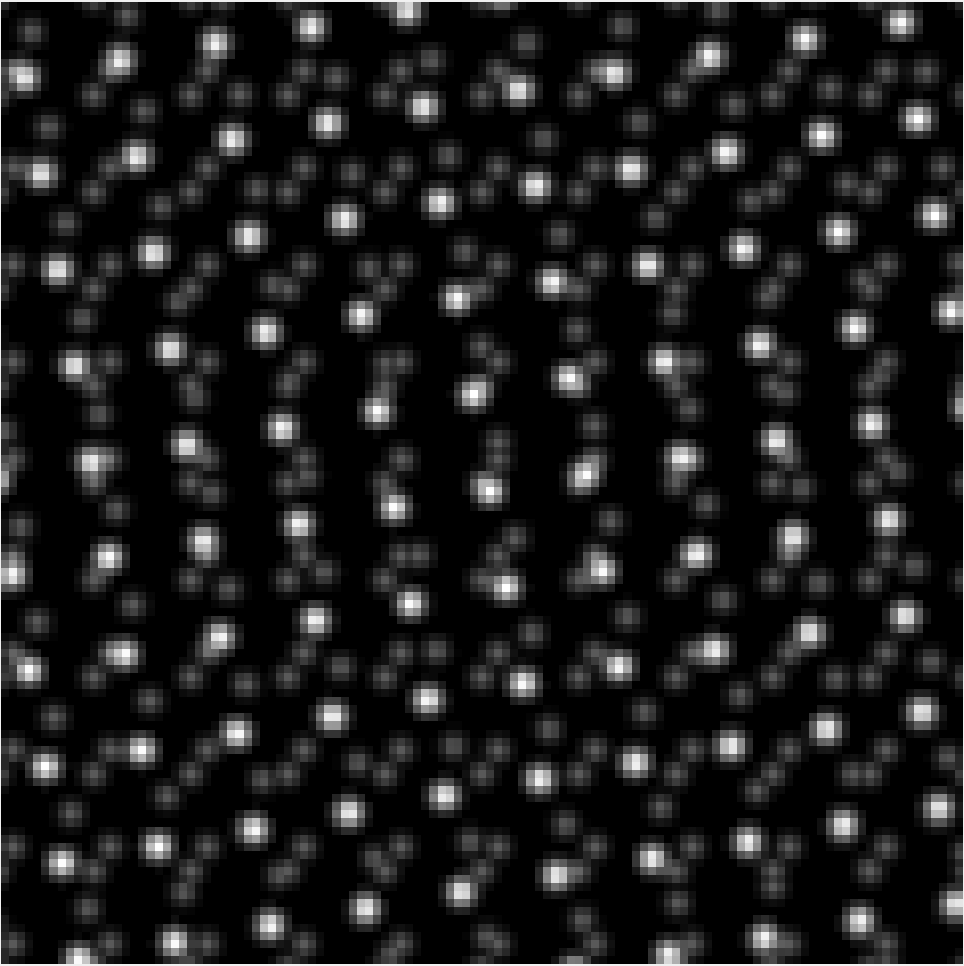} &
    \includegraphics[scale=0.2]{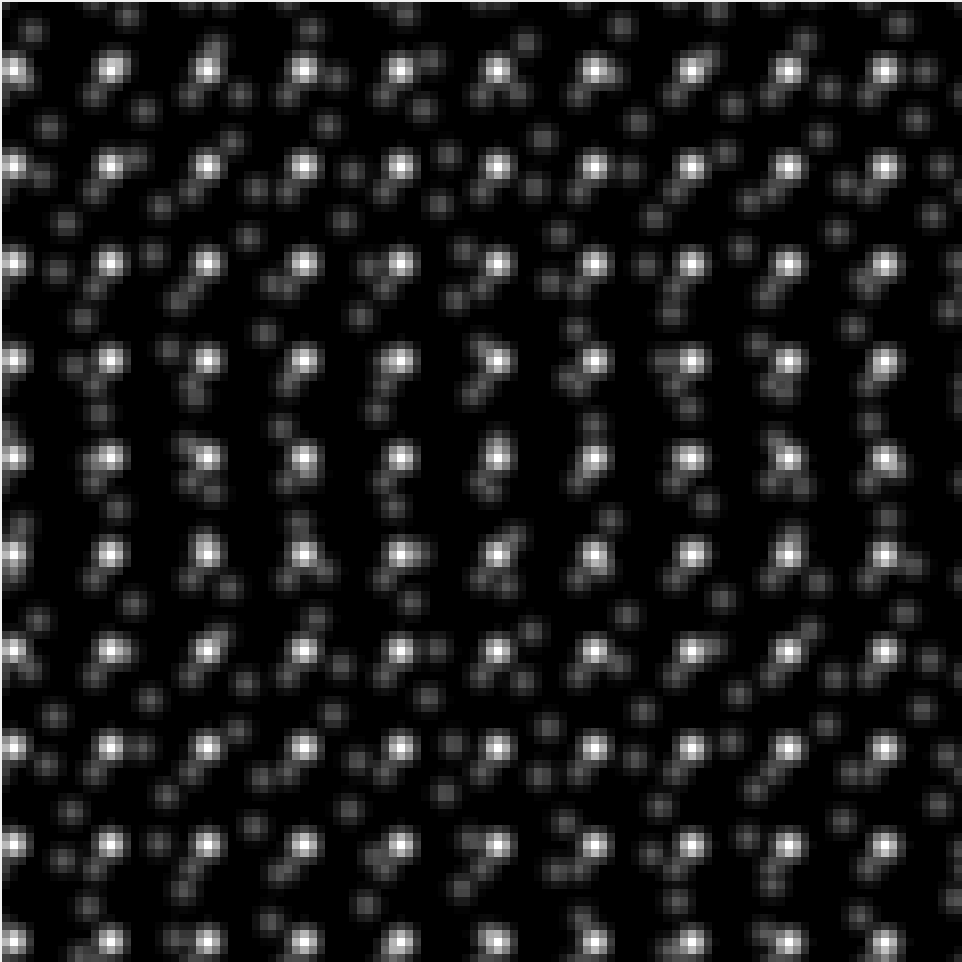} &
    \includegraphics[scale=0.2]{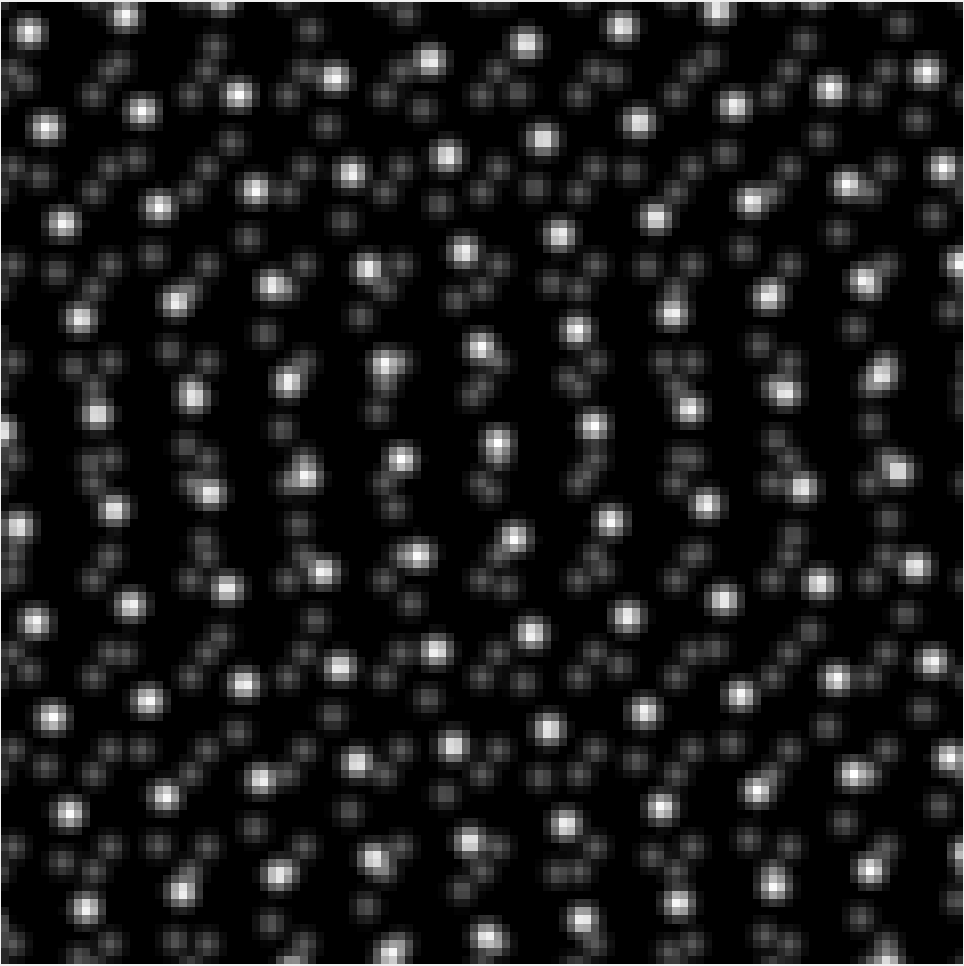} \\
\end{tabular} 
 \end{center}  
 \caption{[Translation] Identification of four lattices in (a). Lattices identified in (b) $\mathcal{T}_{0}\Lambda\langle 12,i\rangle$ and (d) $\mathcal{T}_{2-3i}\Lambda\langle 12,i\rangle$ only differ by translation, so do those in (c) $\mathcal{T}_{1+i}\Lambda\langle 11.8177 + 2.0838i,i\rangle$ and (e) $\mathcal{T}_{2-5i}\Lambda\langle11.8177 + 2.0838i,i\rangle$.}
  \label{shiftexp}
\end{figure}
In Figure~\ref{synclose}, a different set of translational lattices form a grand lattice pattern, whose lattice points are replaced by three dots, each of which belongs to a different lattice. Such configuration results in many L-shapes in the image~\cite{MSBP}. This local ambiguity presents no confusion for LISA since it observes the image globally in the frequency domain. The sensitivity of LISA to the distance between particles is affected by the size of the particle of lattice candidates.
 \begin{figure} 
 \begin{center}
\begin{tabular}{cccc}
(a)  Original image & (b) $d_{\mathscr{L}}(\hat{\Lambda},\Lambda)=0.0072$  & (c) $d_{\mathscr{L}}(\hat{\Lambda},\Lambda)= 0.0155$ & (d) $d_{\mathscr{L}}(\hat{\Lambda},\Lambda)=0.0113$ \\
    \includegraphics[height=1.3in]{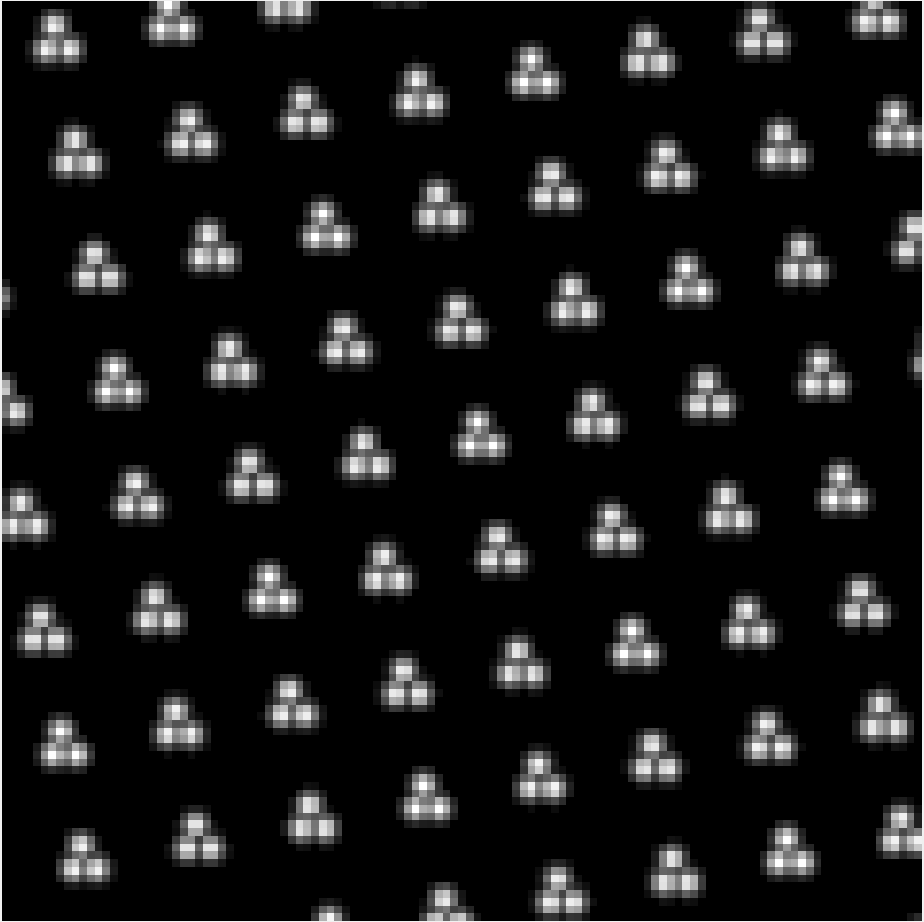}  &
    \includegraphics[height=1.3in]{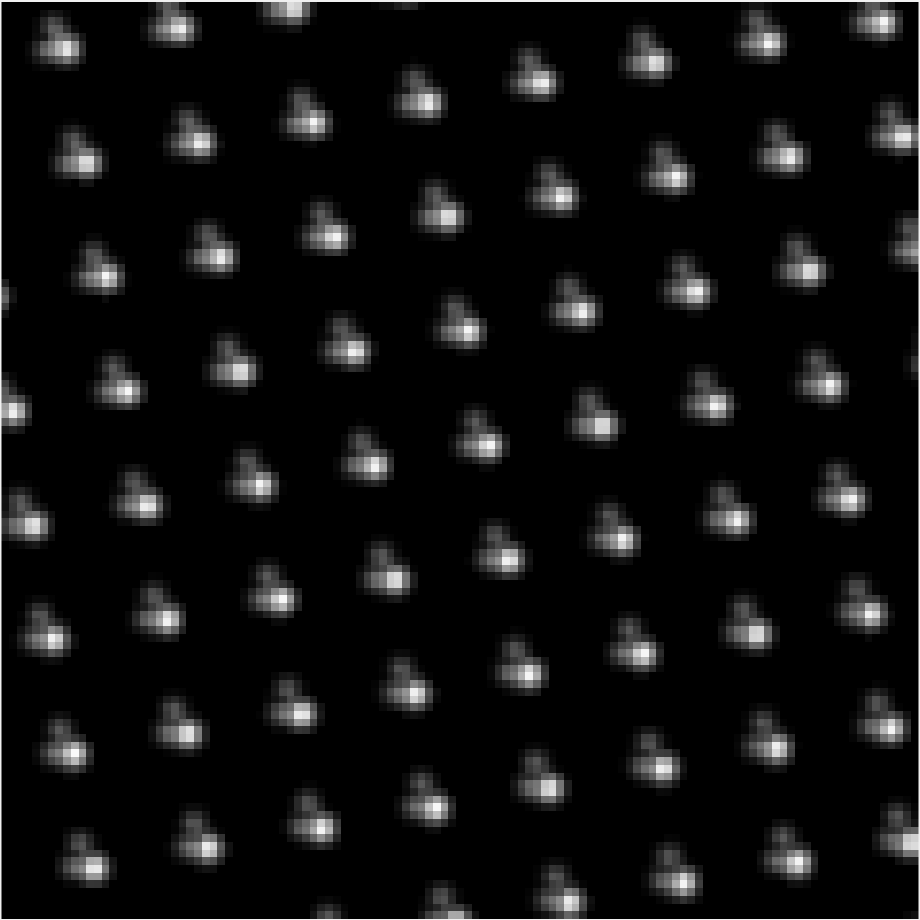} & 
    \includegraphics[height=1.3in]{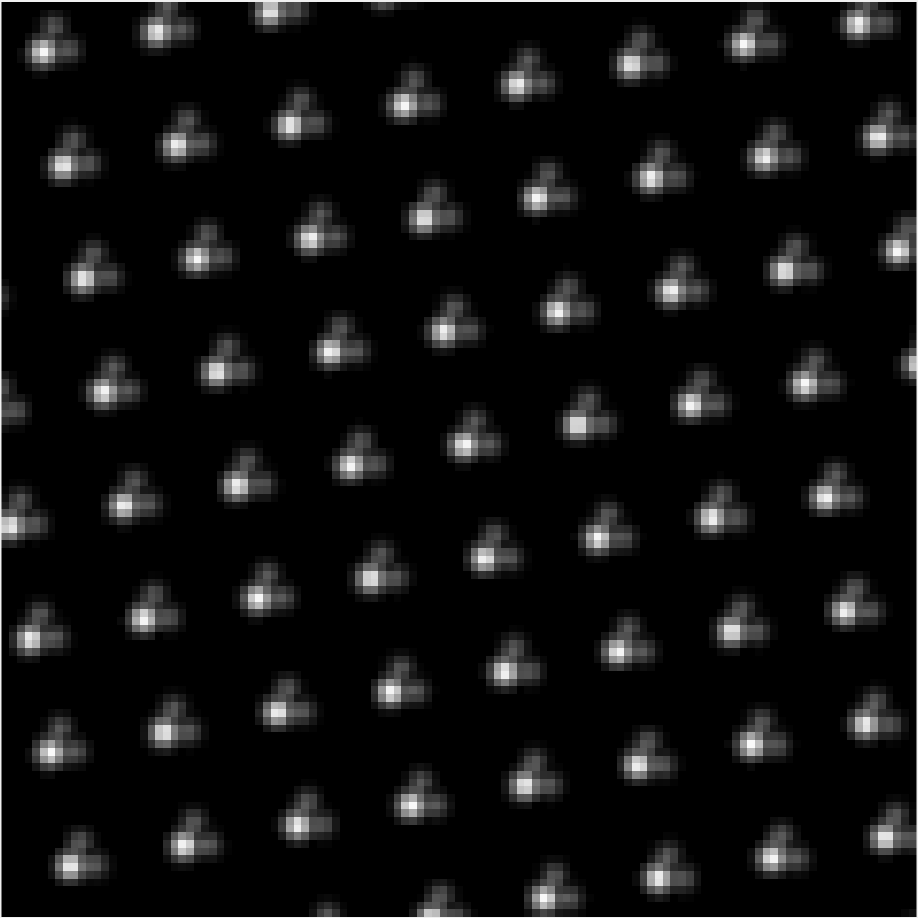} &
    \includegraphics[height=1.3in]{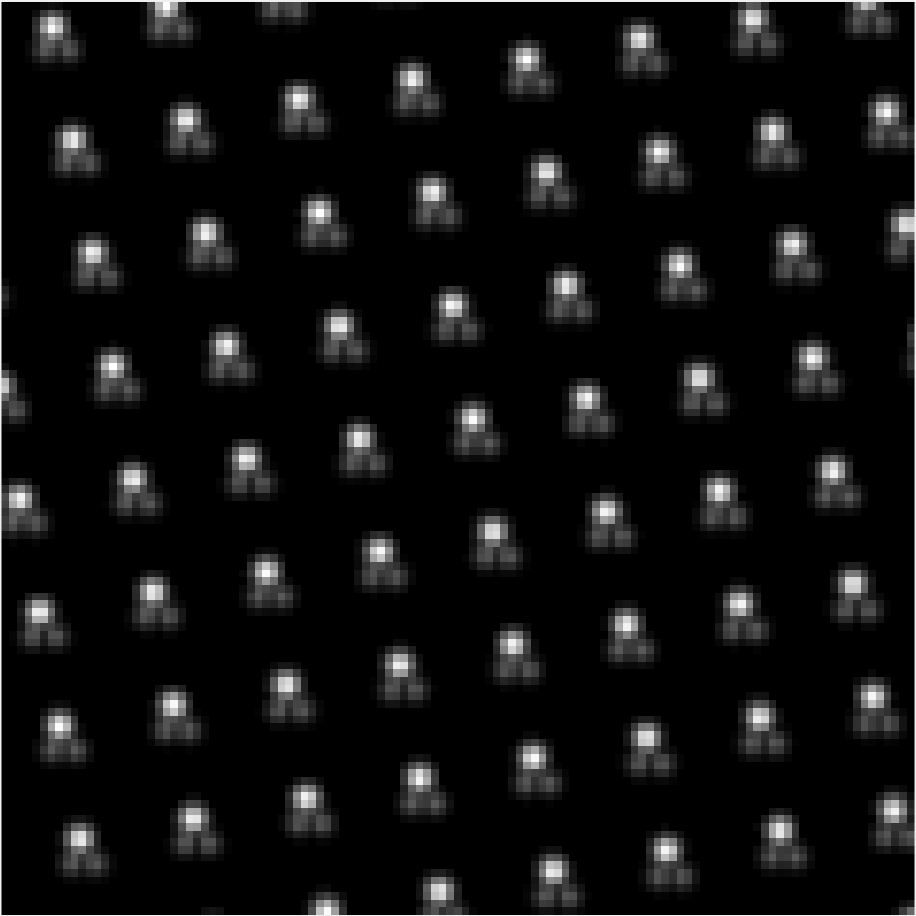} \\
 \end{tabular} 
 \end{center}   
 \caption{[Close translated particles] Identification of three lattices in (a). They are shifted from $\mathcal{T}_{0}\Lambda\langle14.7721 + 2.6047i,i\rangle$ by (b) $4-2i$ (c) $1-2i$. (d) $2-5i$. These relative translations push particles close, and generate a lattice pattern whose lattice points are composed of three dots. LISA successfully distinguishes them with high precision as indicated by $d_{\mathscr{L}}(\hat{\Lambda},\Lambda)$ values above each lattice.}
  \label{synclose}
\end{figure}

In practice, only a portion of the lattice patterns may be present.  For example, Figure~\ref{completeexp} (a) is a superlattice composed by a completely presented lattice (b) and a partially displayed lattice (c), where $50\%$ of the particles are missing. The incompleteness modifies the power spectrum by convolving the reciprocal lattice of (c) with the Fourier transform of the lower triangular shape, resulting in weaker responses. Hence, LISA identifies the complete lattice (b) first, then reveals the incomplete lattice (c). Notice that LISA identifies the complete lattice pattern corresponding to (c), instead of the incomplete lattice image (c). This is shown in (e), where the identified pattern extends to the upper triangular region. A simple improvement of this visual presentation would be directly comparing the identified pattern with the original image, which is shown in (f). We also experiment in situations where $70\%$ of the particles from one of the lattices are missing, and LISA recognizes the incomplete lattices successfully. An obvious upper bound for the number of missing particles is that the average intensity of the incomplete lattice must remain at least $0.01$ by the terminating condition of LISA.  
\begin{figure} 
 \begin{center}
\begin{tabular}{ccc}
(a) Original image & (b) Complete lattice $\Lambda_{1}$ & (c) Partial lattice $\Lambda_{2}$ \\
\includegraphics[height=1.3in]{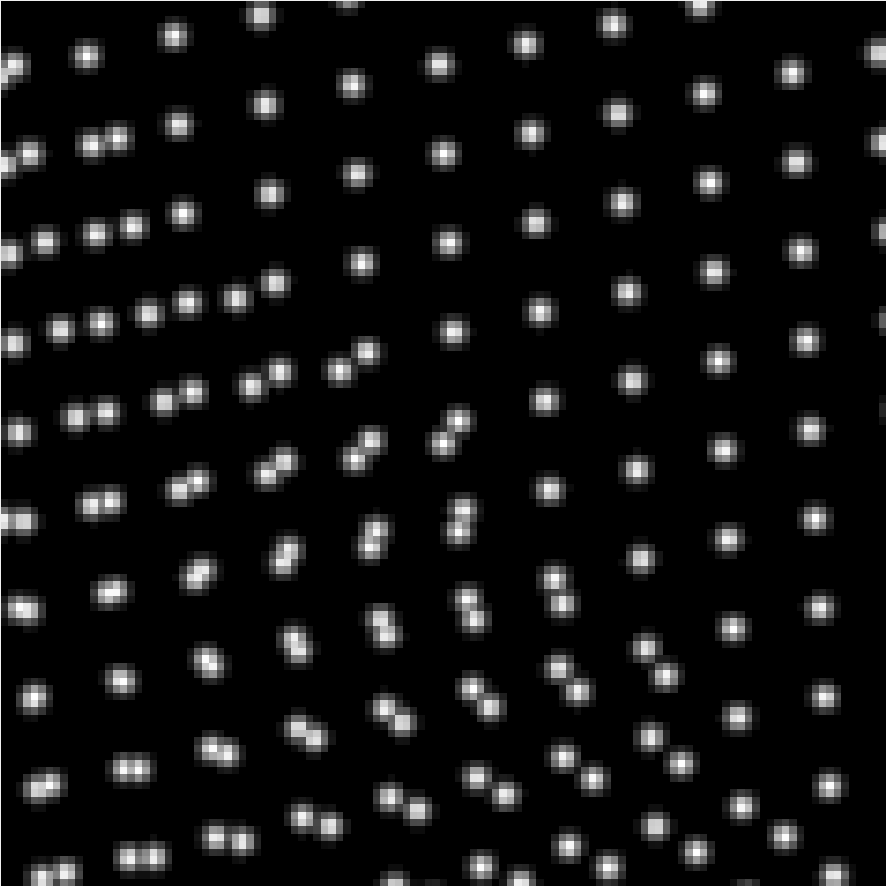} &
 \includegraphics[height=1.3in]{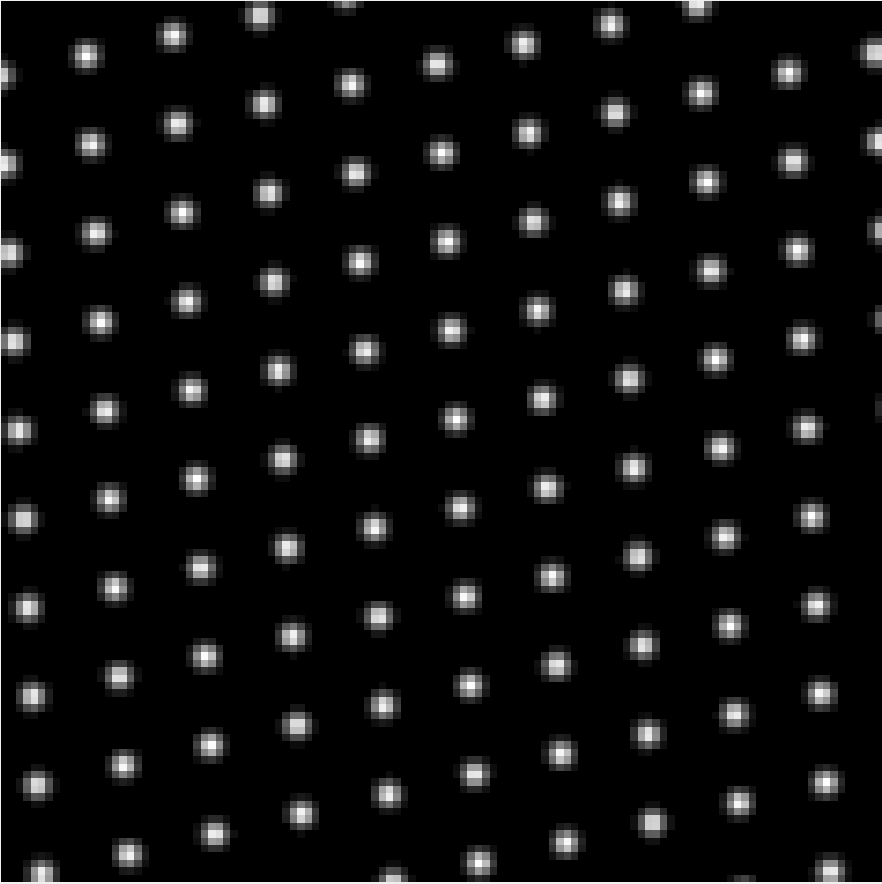} &
 \includegraphics[height=1.3in]{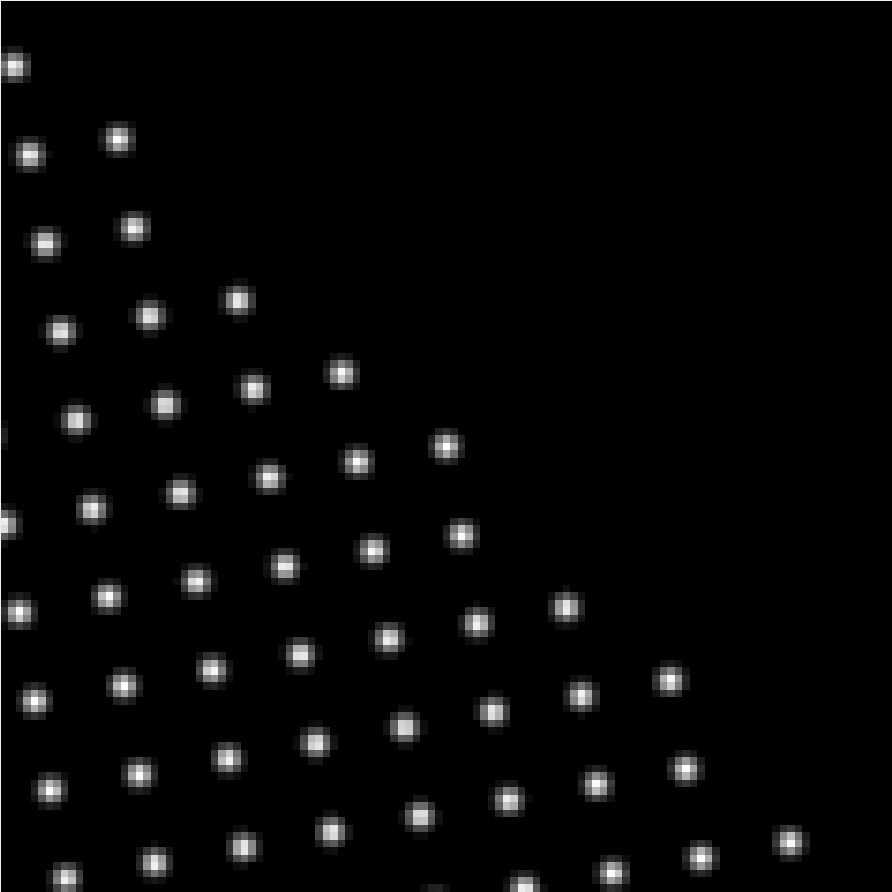} 
\\
 (d) $d_{\mathscr{L}}(\hat{\Lambda}_{1},\Lambda_{1})=0.0056$ & (e) $d_{\mathscr{L}}(\hat{\Lambda}_{1},\Lambda_{1})=0.0036$ & (f) Improved presentation \\ 
\includegraphics[height=1.3in]{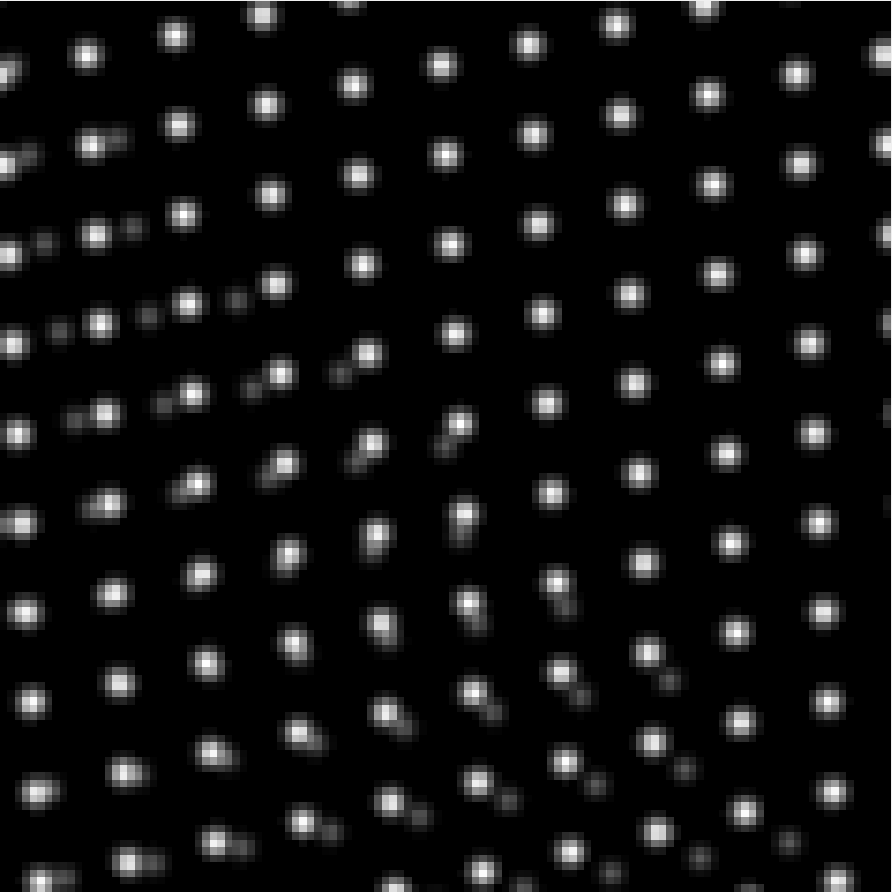} &
\includegraphics[height=1.3in]{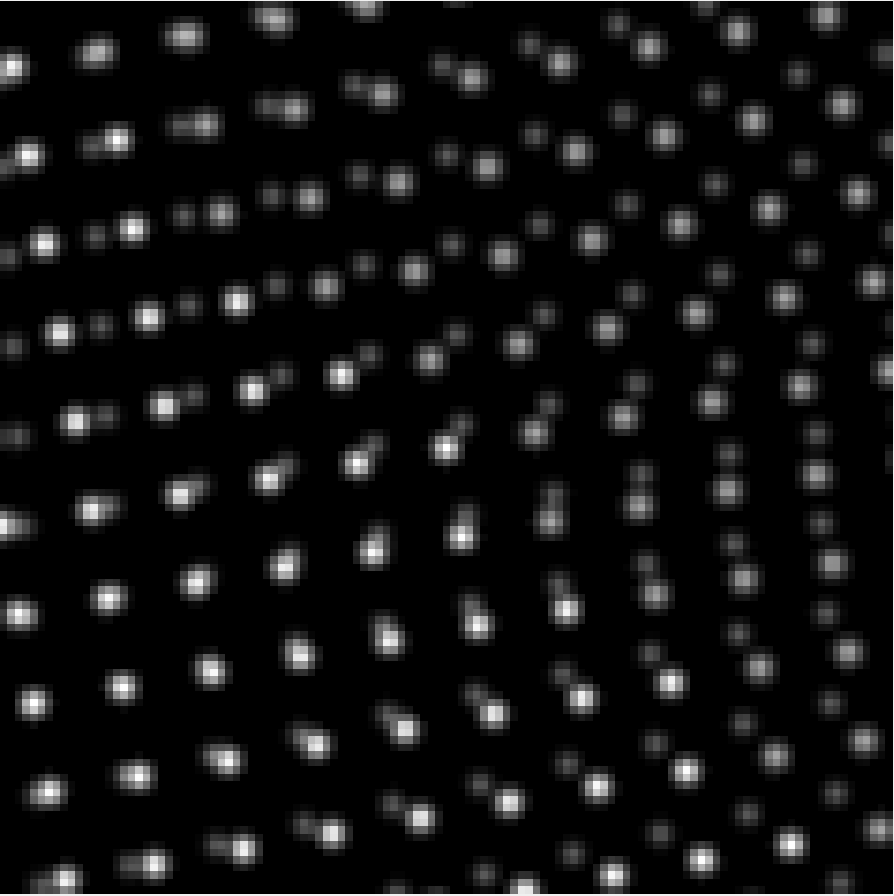} &
\includegraphics[height=1.3in]{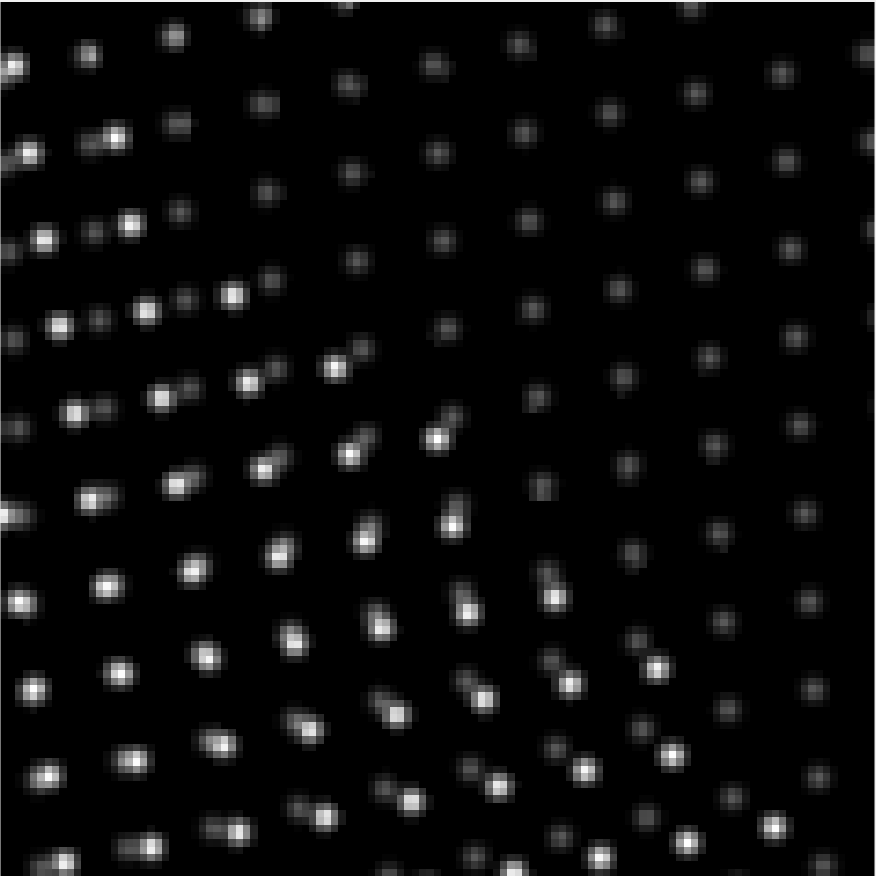}
 \end{tabular} 
 \end{center}   
  \caption{[Incomplete lattice] The superlattice in (a) is a mixture of (b) $\mathcal{T}_{0}\Lambda\langle 11.6924 + 2.6994i,e^{i4\pi/9}\rangle$ and (c) $\mathcal{T}_{2-3i}\Lambda\langle11.8177 + 2.0838i,i\rangle$, which is incomplete. (d) and (e) are the identified patterns by LISA. By directly comparing (e) with (a), (f) improves the visual presentation.}
  \label{completeexp}
\end{figure}

The evaluation method~(\ref{eval}) proposed in section~\ref{sec:LISA} considers  the density restriction, i.e., the overfitting term.  Figure~\ref{scoreDemo} illustrates this importance. Image (a) is generated with two lattices presenting a region of moir\'{e} pattern in the center. Without the density restriction in (\ref{eval}), image (b) is identified while with the density restriction image (c) is the identified lattice.  Comparing the lattice (b) and (c), shown in (d), (b) is dense and (c) is almost a sub-lattice of (b). Although (b) seems to have identified more points (according to (d)), comparing the lattice (b) with the given image (a), in fact, many points are not identified --- shown in (e). In this case, lattice (c) was one of the underlying lattices. In the frequency domain, large-scale moir\'{e} patterns tend to produce strong responses on power spectrum surface.  Lattice candidates associated with these high responses are excessively dense, and they partially coincide with the moir\'{e} patterns in the given image, which causes the next-level-identification unstable. Hence the density restriction in~(\ref{eval}) makes LISA robust against possible moir\'{e} patterns.
\begin{figure}
\begin{center}
\begin{tabular}{ccc}
(a)  & (b) & (c) \\
    \includegraphics[height=1.3in]{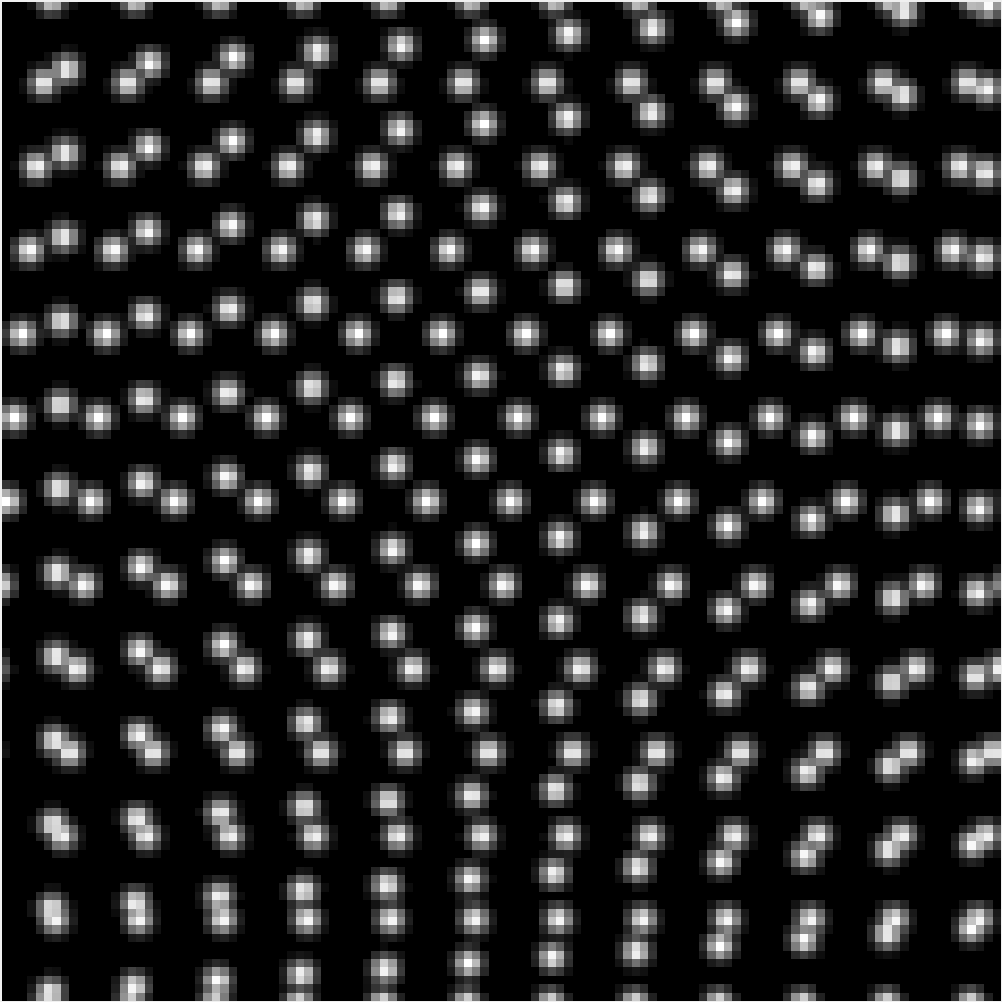} &
    \includegraphics[height=1.3in]{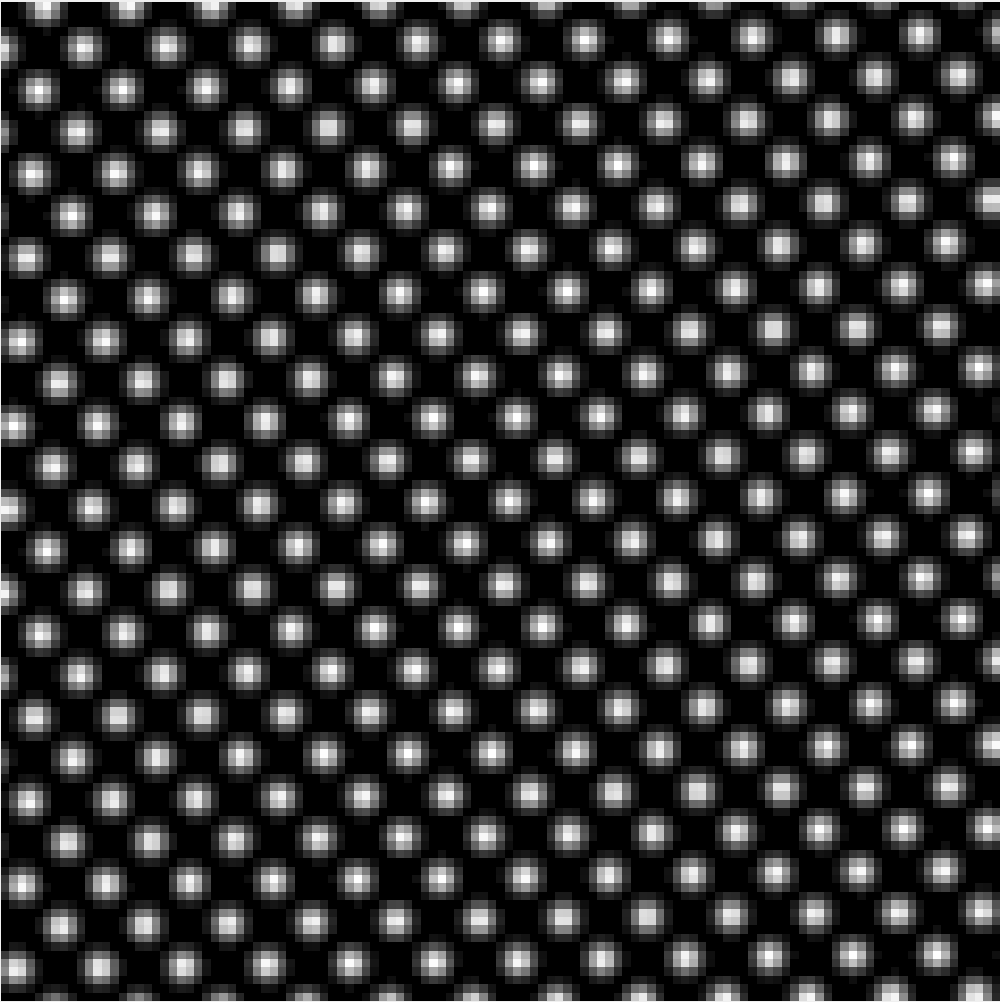} &
    \includegraphics[height=1.3in]{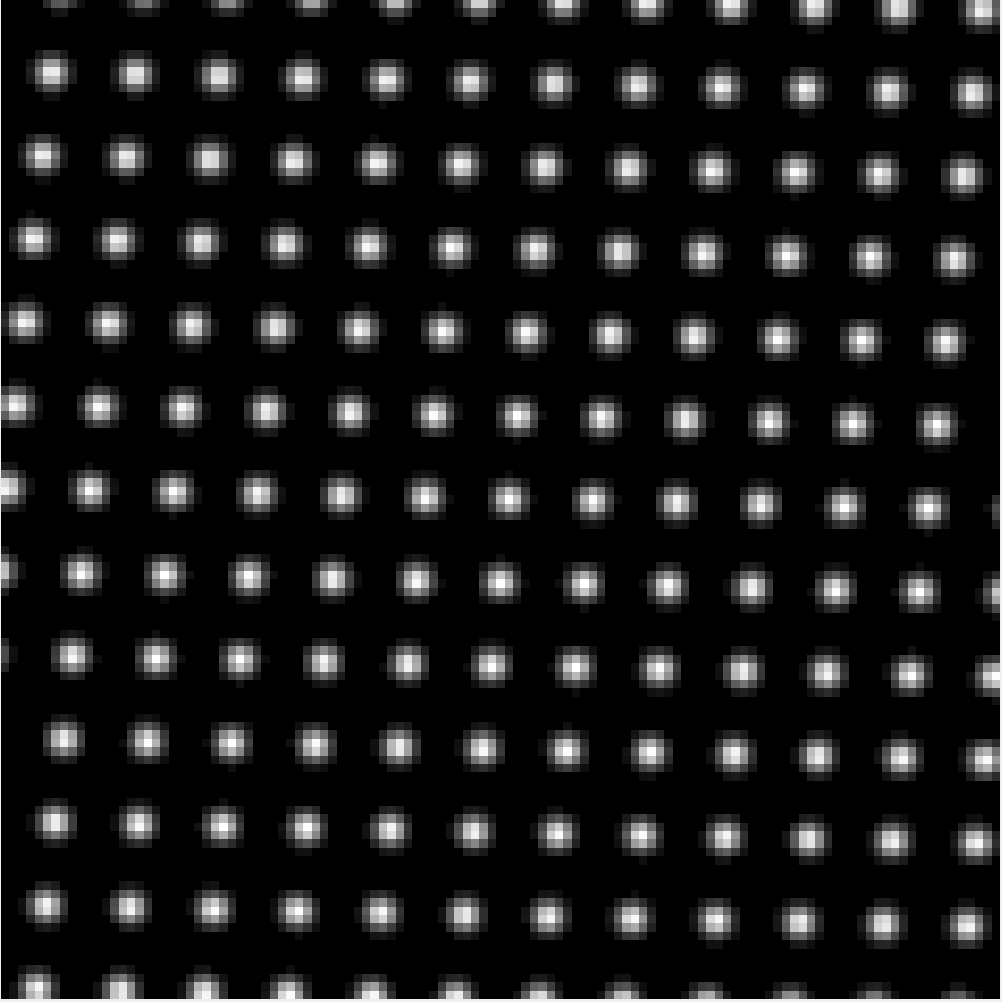}\\
(d) &  & (e) \\
     \includegraphics[height=1.3in]{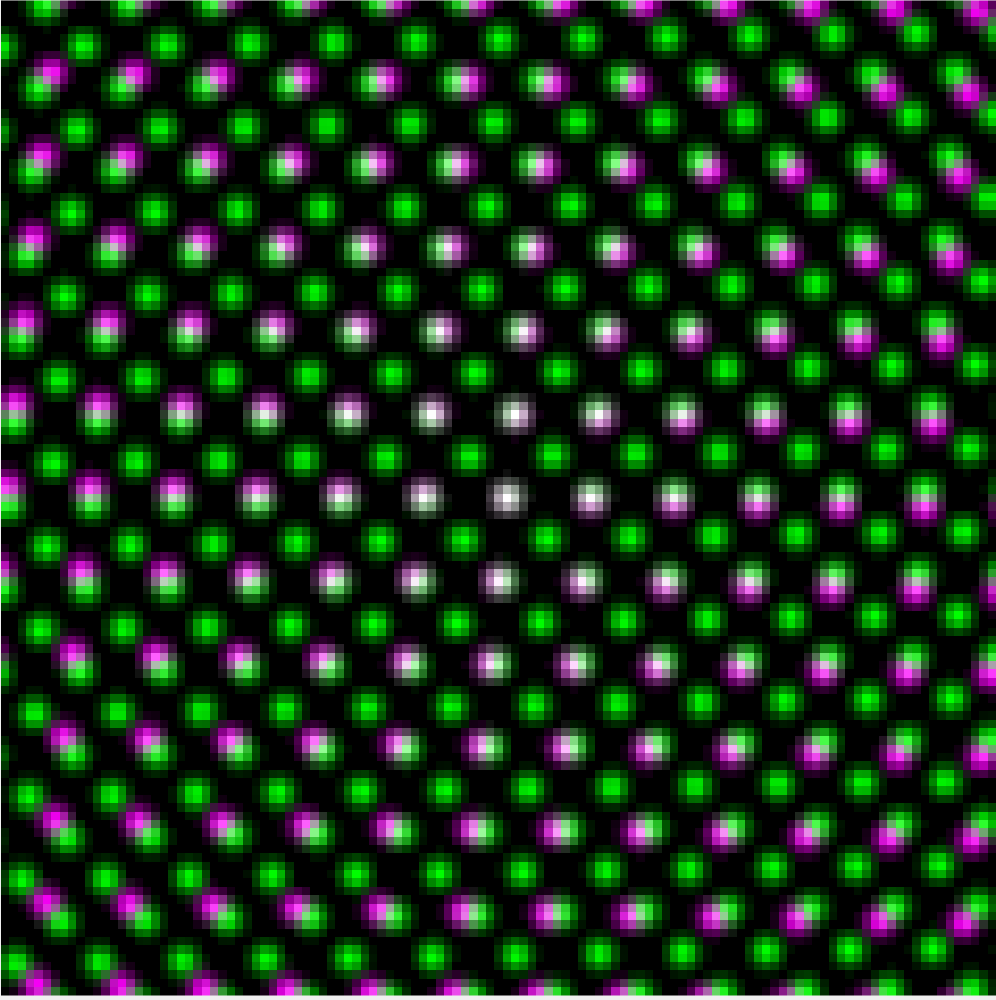} & &
    \includegraphics[height=1.3in]{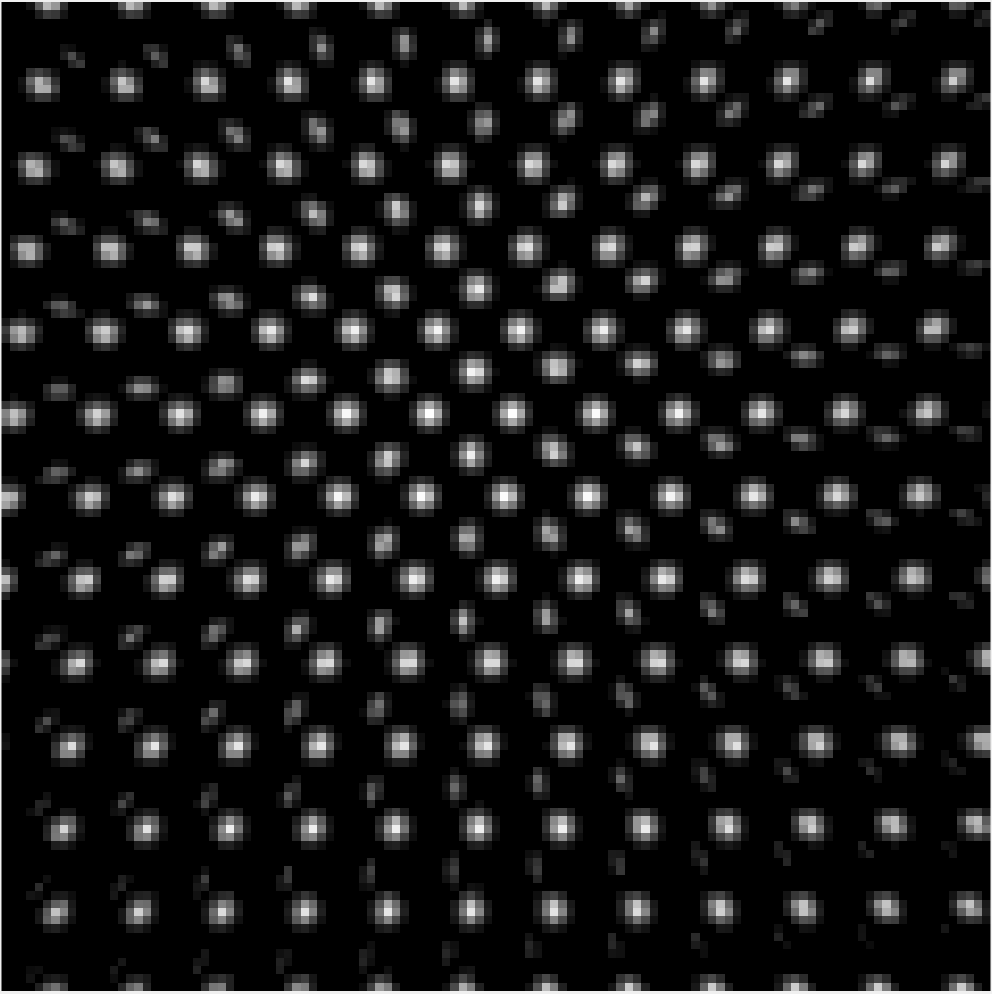} 
 \end{tabular} 
 \end{center}   
  \caption{[Importance of density restriction] (a)  Given image $I$ generated by two lattices $\mathcal{T}_{2-10i}\Lambda\langle10, e^{i17\pi/36}\rangle$, and $\mathcal{T}_{-3+5i}\Lambda\langle 9.9756 + 0.6976i,e^{i17\pi/36}\rangle$. If the score~(\ref{eval}) does not have the second term, we obtain a dense lattice $\mathcal{T}\tilde{\Lambda}$ in (b). With the density restriction, we get $\mathcal{T}\hat{\Lambda}$ in (c) which is the correct lattice pattern. (d) compares patterns in (b) and (c) where white pixels are $\mathcal{T}\hat{\Lambda}\cap \mathcal{T}\tilde{\Lambda}$, the green are $\mathcal{T}\tilde{\Lambda}-\mathcal{T}\hat{\Lambda}$, and the red are $\mathcal{T}\hat{\Lambda}-\mathcal{T}\tilde{\Lambda}$. (e) displays $\min\{\mathcal{T}\tilde{\Lambda},I\}$.}
  \label{scoreDemo}
\end{figure}

Superposed lattices can present various attractive patterns, and the formation simply involves scaling and rotating. Hexagonal lattices,   which share shape descriptors $\rho=\pm1/2+i\sqrt{3}/2$, produce the most variation. This is because, in the lattice space, their equivalent classes have the most elements, i.e., they are more symmetrical than the other lattices. In Figure~\ref{syn31}, 4 hexagonal lattices with $\beta = 10,12,13$ and $15$ are superposed, which displays a flake-like pattern. 
In Figure~\ref{syn33}, a flower-like pattern is formed by 4 hexagonal lattices with identical scale descriptor norm $|\beta|=11$ but different inclination angles: $53^{\circ}$, $-53^{\circ}$, $143^{\circ}$ and $-143^{\circ}$.

\begin{figure} 
\begin{center}
\begin{tabular}{cccc}
(a) Original image &&&\\
\includegraphics[scale=0.21]{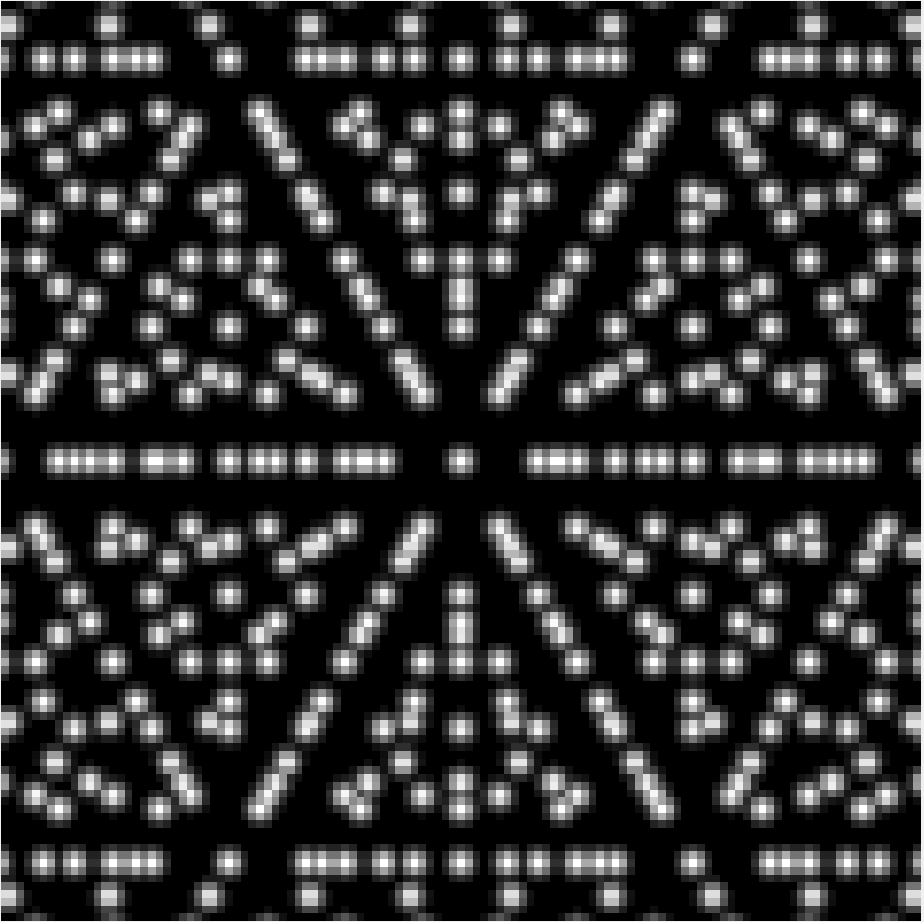}  & &&\\
 (b) $d_{\mathscr{L}}(\hat{\Lambda},\Lambda)=0.0120$ &
(c) $d_{\mathscr{L}}(\hat{\Lambda},\Lambda)=0.0230$ & 
(d) $d_{\mathscr{L}}(\hat{\Lambda},\Lambda)=0.0062$ &
(e) $d_{\mathscr{L}}(\hat{\Lambda},\Lambda)=0.0143$ \\
\includegraphics[scale=0.21]{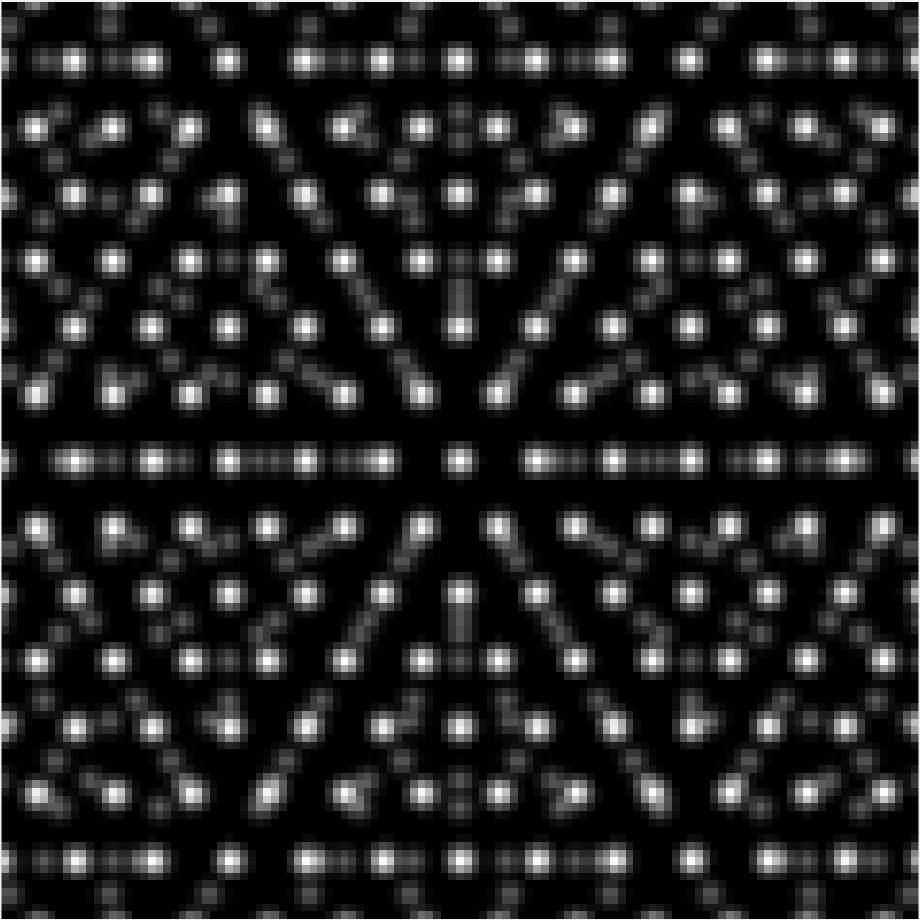} &
\includegraphics[scale=0.21]{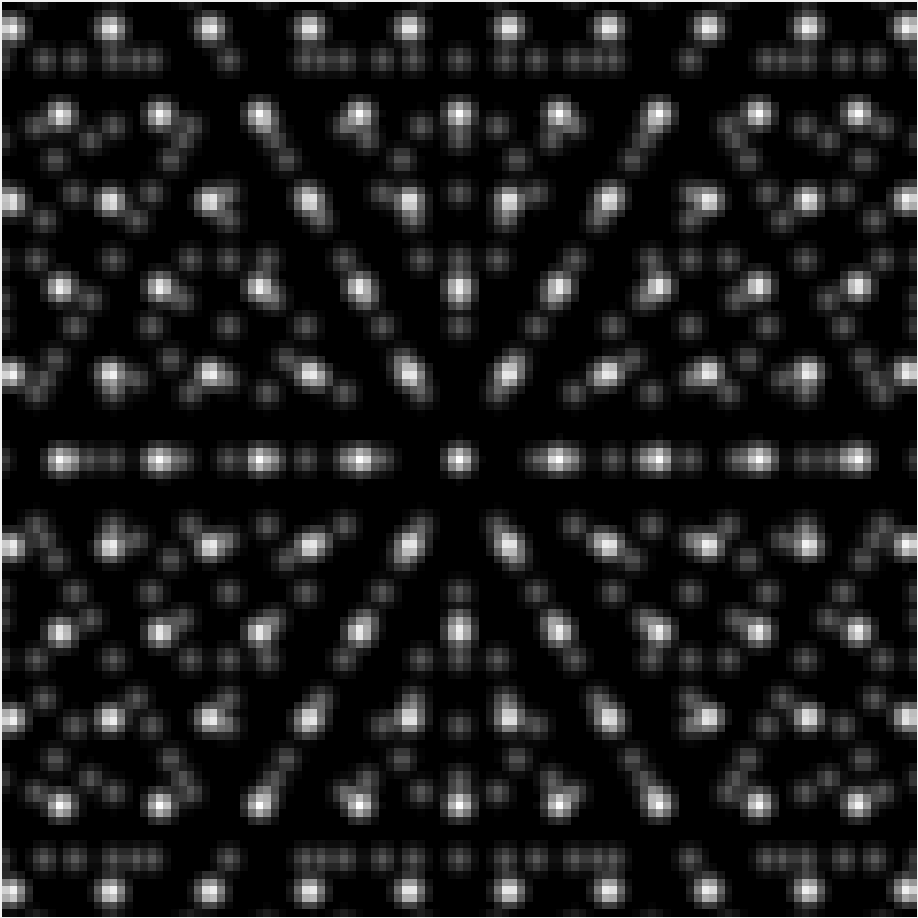} &
\includegraphics[scale=0.21]{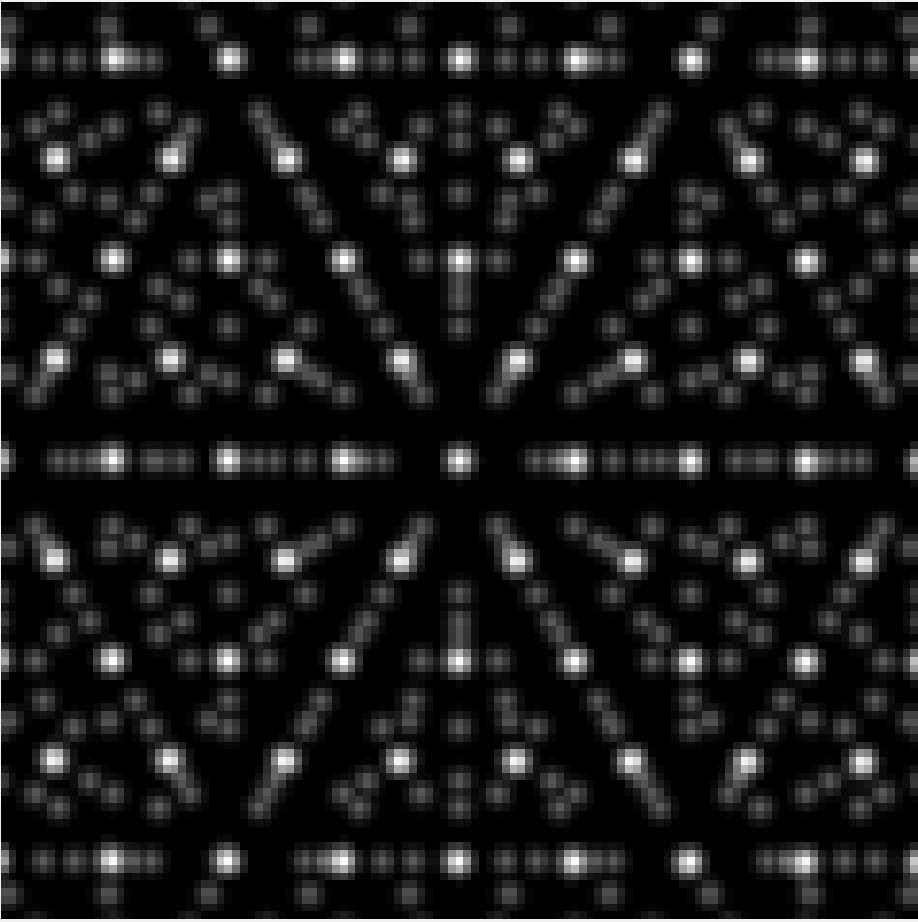} &
\includegraphics[scale=0.21]{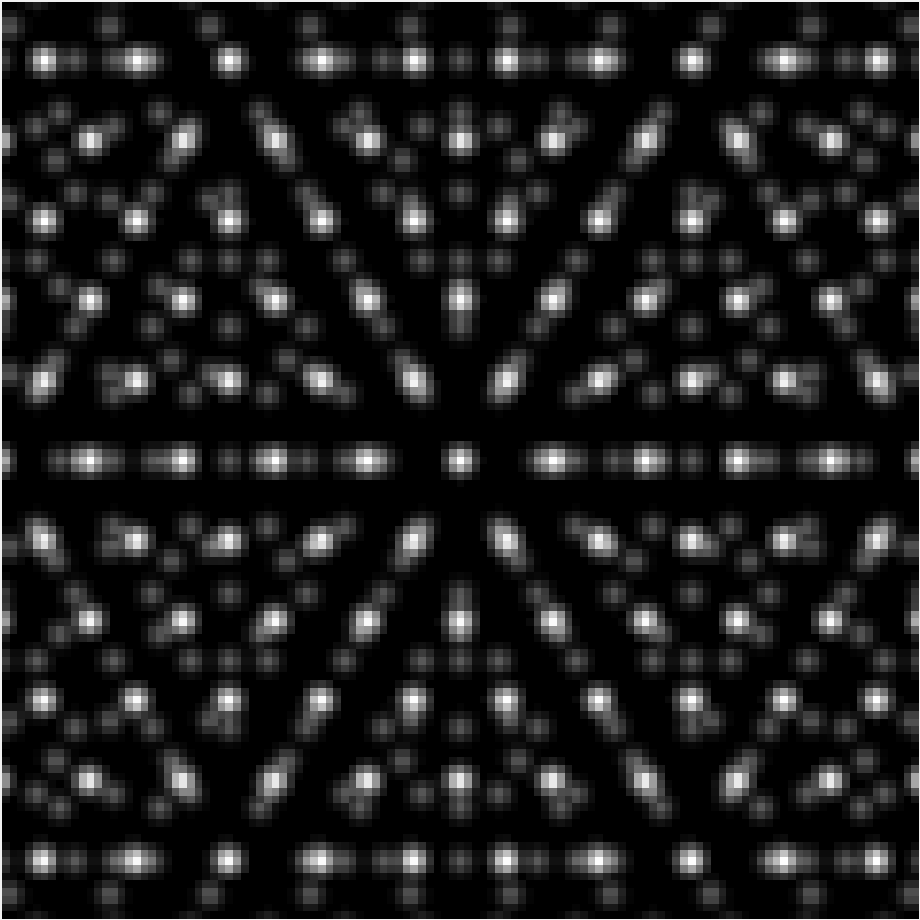} 
 \end{tabular} 
 \end{center}   
  \caption{[Flake-like pattern generated by lattices] Identification of four lattices in the flake-like pattern in (a). The hexagonal lattices have scale descriptor $\beta$ equal to (b) $10$ (c) $13$ (d) $15$ and (e) $12$. }
  \label{syn31}
\end{figure}

   

\begin{figure} 
\begin{center}
\begin{tabular}{cccc}
(a) Original image & &&\\
    \includegraphics[scale=0.177]{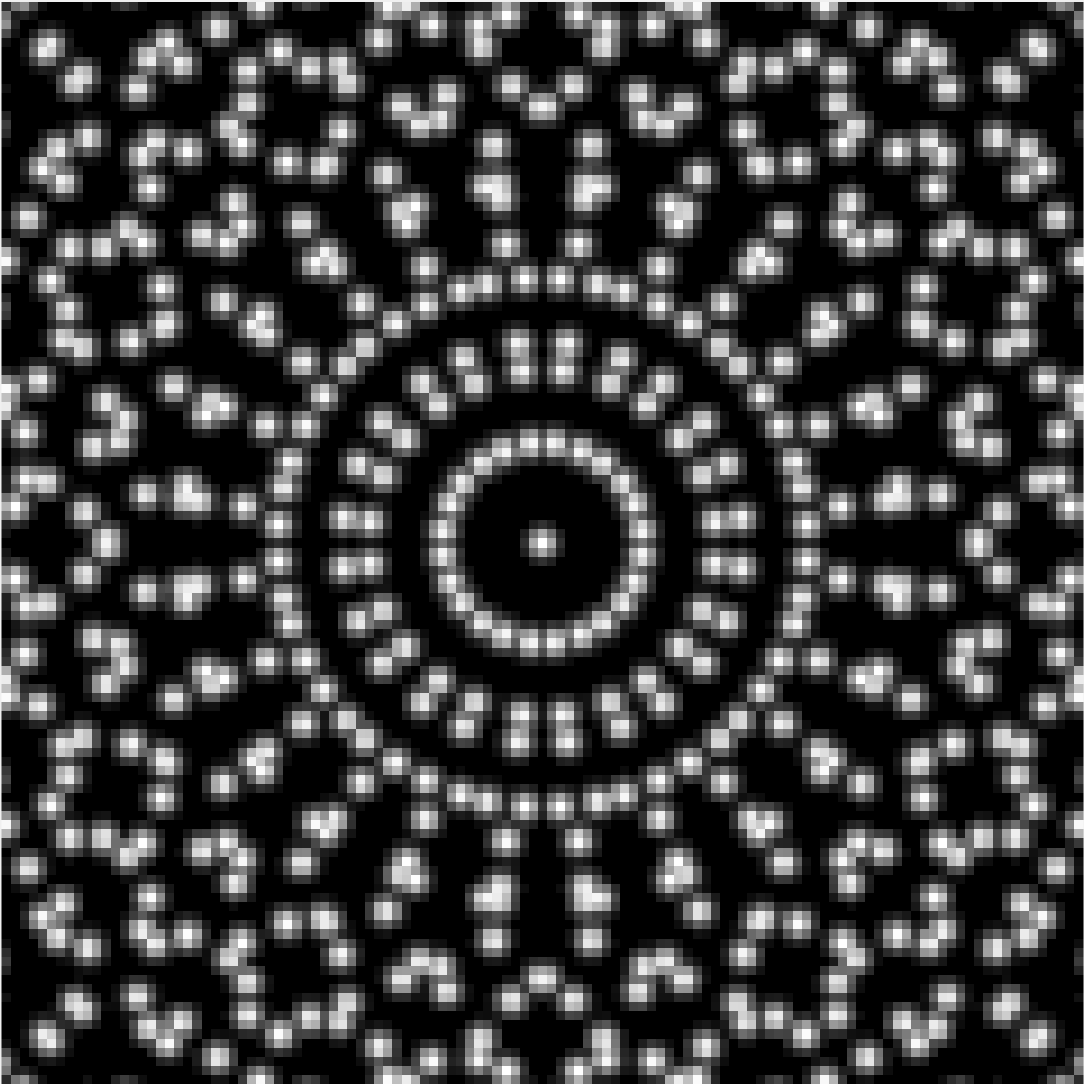} &&&\\
(b) $d_{\mathscr{L}}(\hat{\Lambda},\Lambda)=0.0152$ &
(c) $d_{\mathscr{L}}(\hat{\Lambda},\Lambda)=0.0092$ & 
(d) $d_{\mathscr{L}}(\hat{\Lambda},\Lambda)=0.0058$ &
(e) $d_{\mathscr{L}}(\hat{\Lambda},\Lambda)=0.0041$ \\
    \includegraphics[scale=0.177]{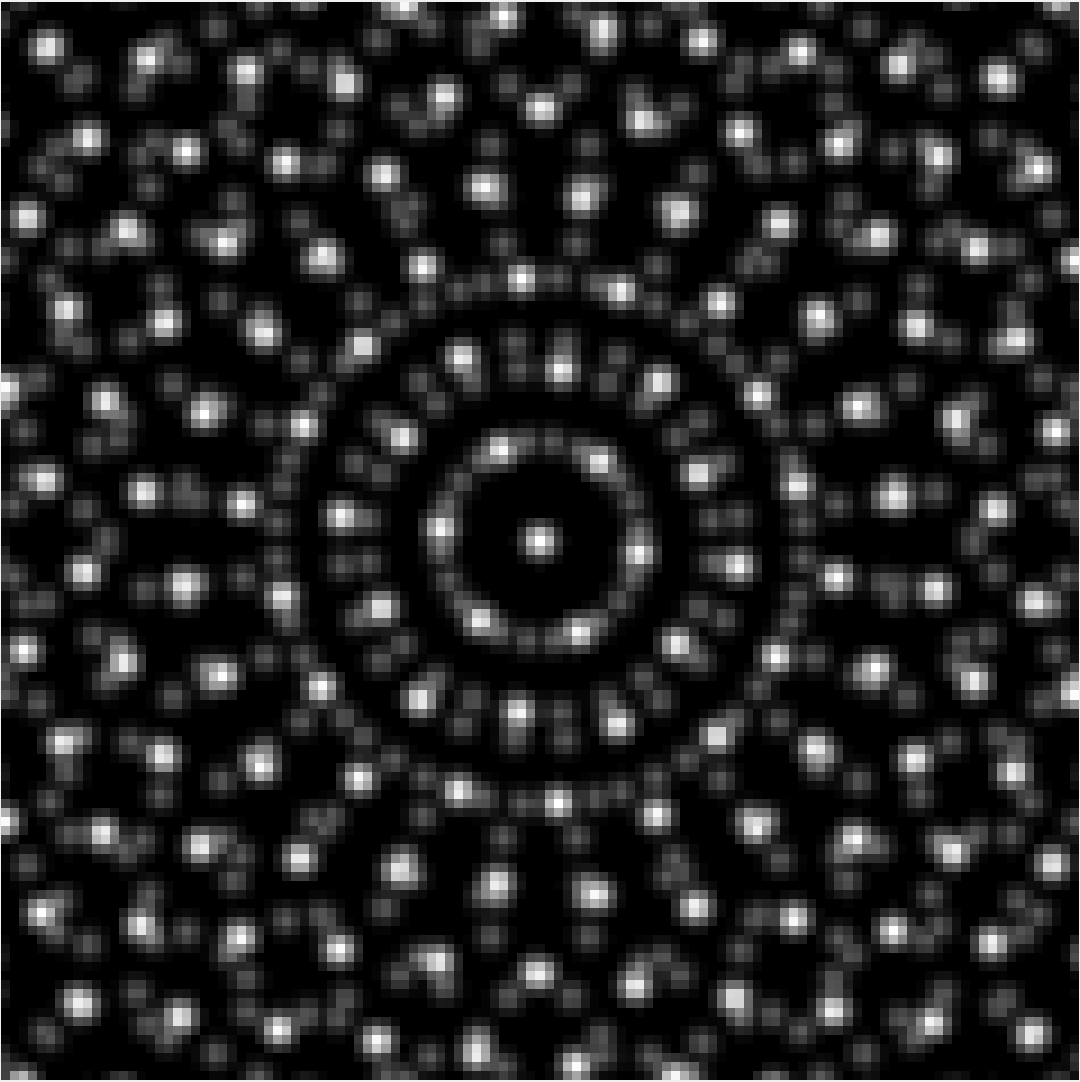} &
    \includegraphics[scale=0.177]{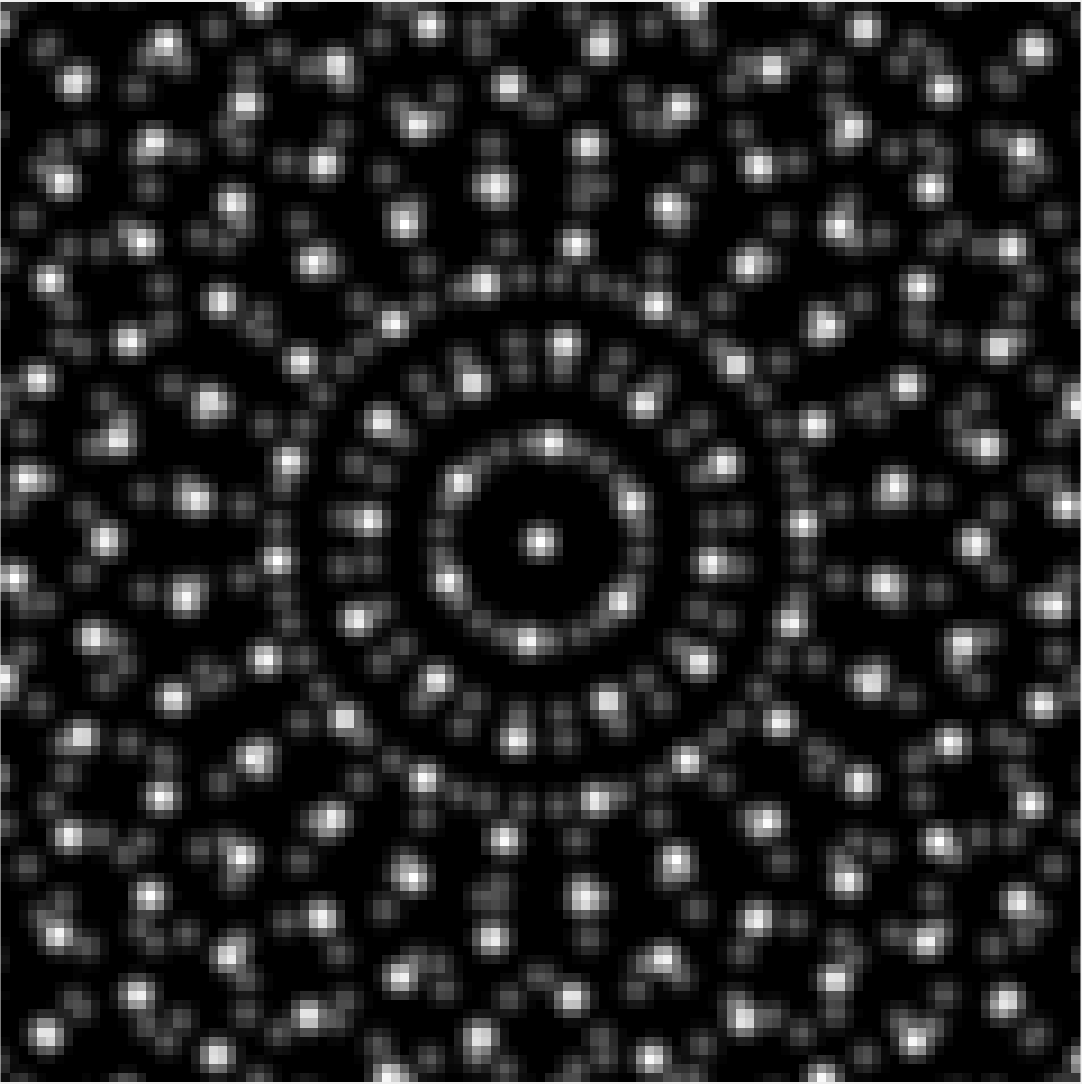} &
    \includegraphics[scale=0.177]{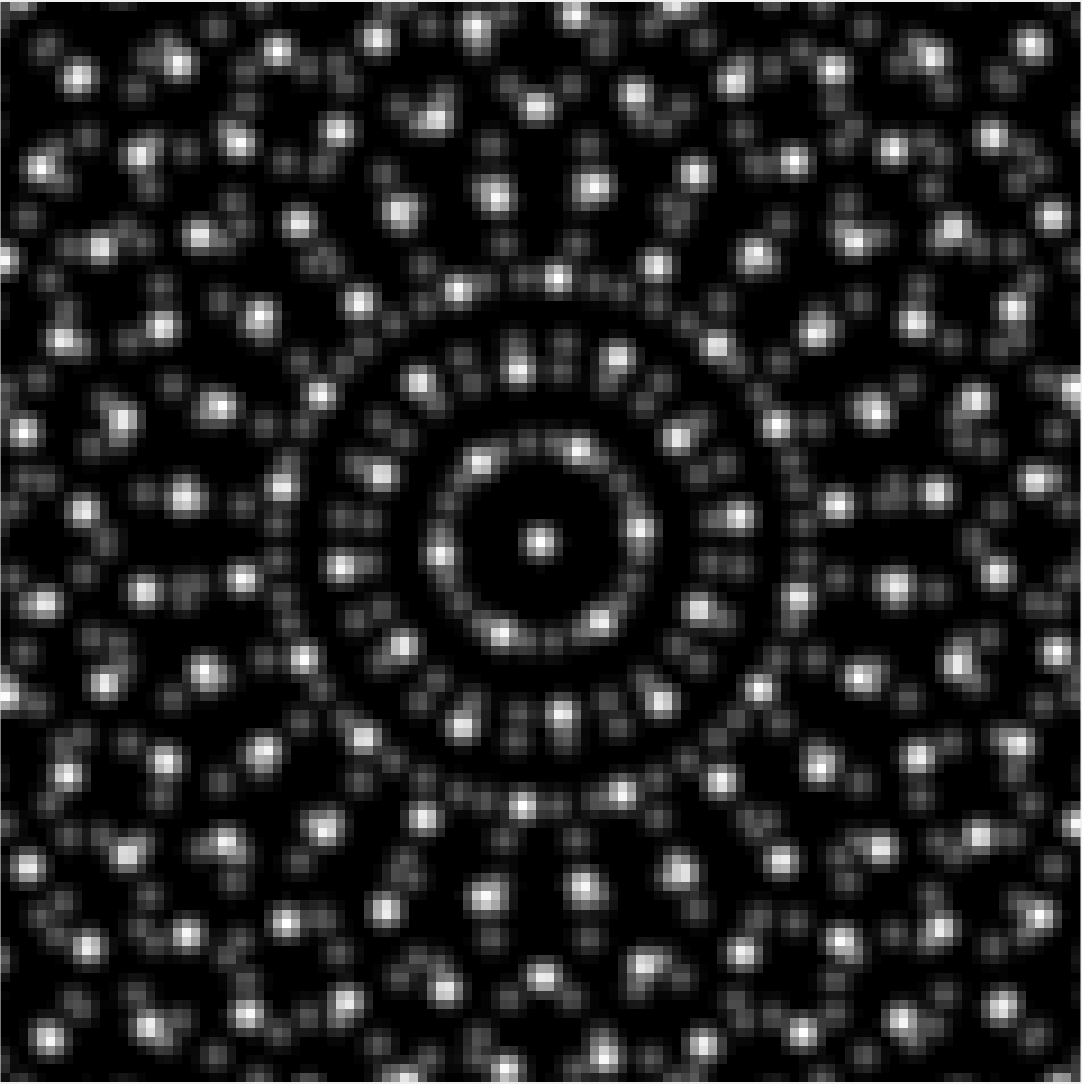} &
     \includegraphics[scale=0.177]{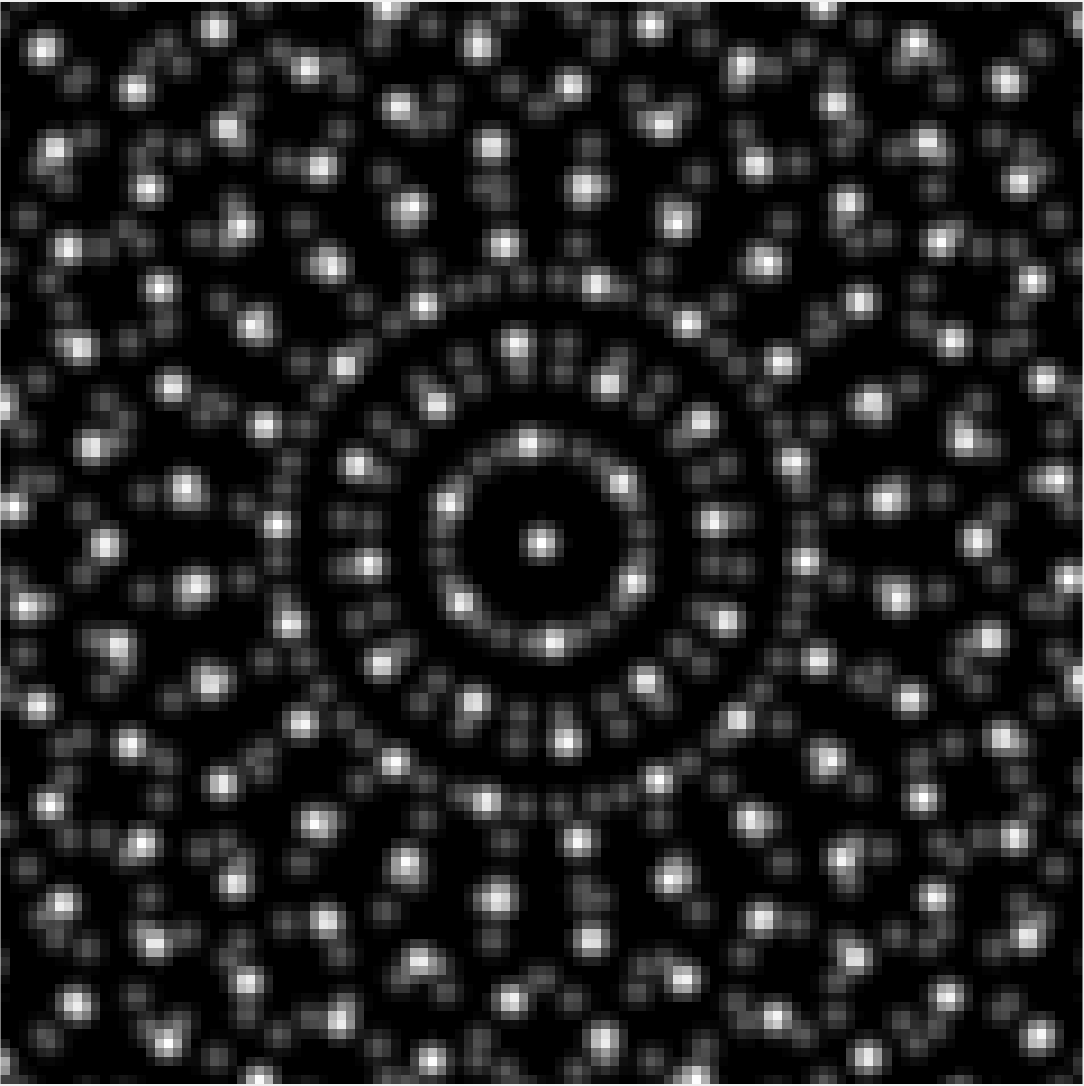} 
 \end{tabular} 
 \end{center}   
  \caption{[Flower pattern generated by lattices] Identification of four lattices in the rotational pattern in (a). The hexagonal lattices with scale descriptors having common norm $|\beta| = 11$ have inclination angles equal to (b) $53^{\circ}$ (c) $143^{\circ}$ (d) $-53^{\circ}$ and (e) $-143^{\circ}$. }
  \label{syn33}
\end{figure}
 
Figure~\ref{realimage} (a) displays a portion of an image in~\cite{zheng2014high}, which is acquired by performing SAED on $\text{Na}$-exfoliated single-layer $\text{MoS}_{\text{2}}$. As mentioned in the introduction, TMD monolayer has three layers of lattices.  In the top-view, $\text{S}$-atoms on the top overlap with those in the bottom. LISA successfully identifies the visible lattices. Figure~\ref{realimage} (b) shows a part of a HREM image of single layer of $\text{MoSe}_{\text{2}}$ from~\cite{rao2013graphene}. Underlying lattices with bright lattice particles are identified and separated by LISA, yet the dimmer lattice particles fail to be recognized. 
This can be addressed by lowering the threshold obtained by Otsu's method, image enhancing techniques, or sophisticated feature point detectors.

\begin{figure}
\begin{center}
\begin{tabular}{cccc}
(a) & (b) & (c) & (d)\\
\includegraphics[height = 1.3in]{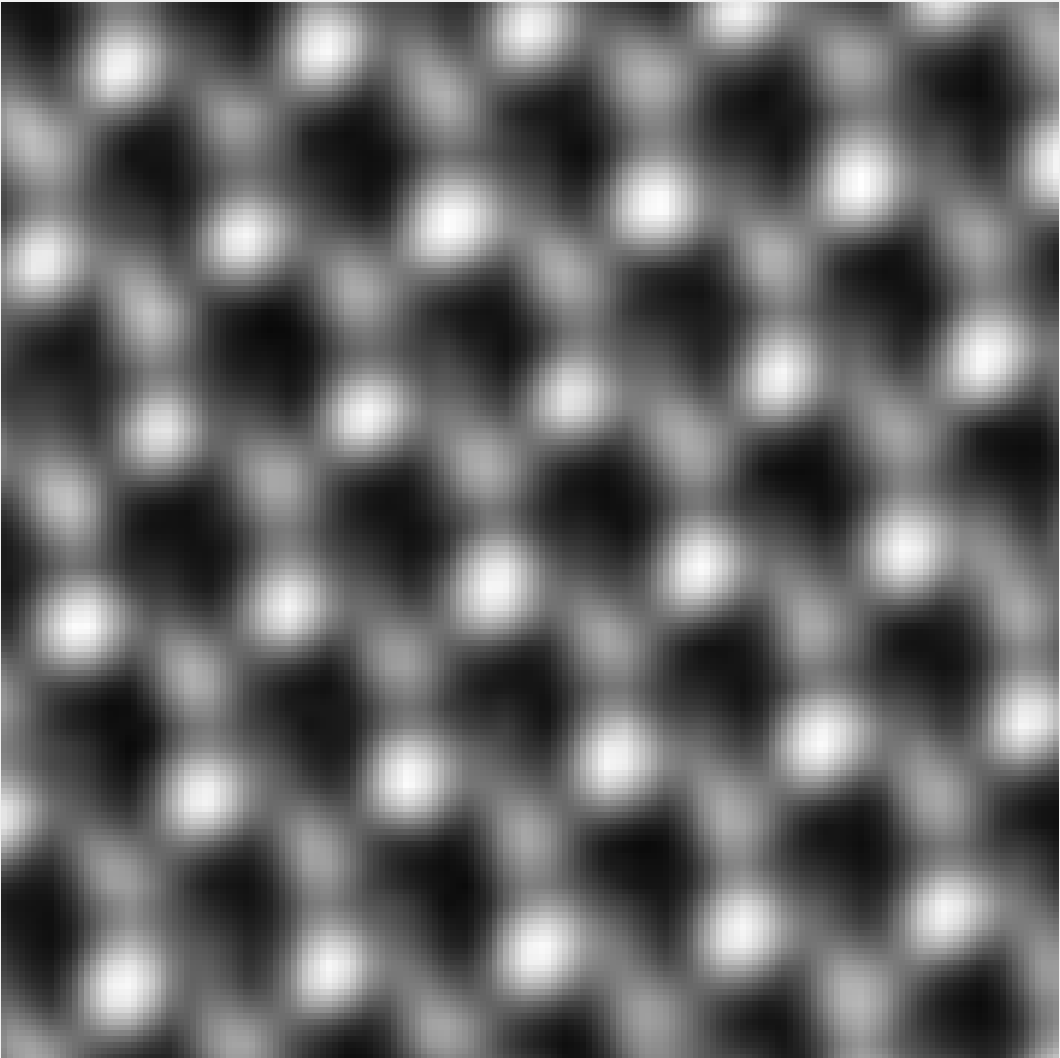} &
\includegraphics[height = 1.3in]{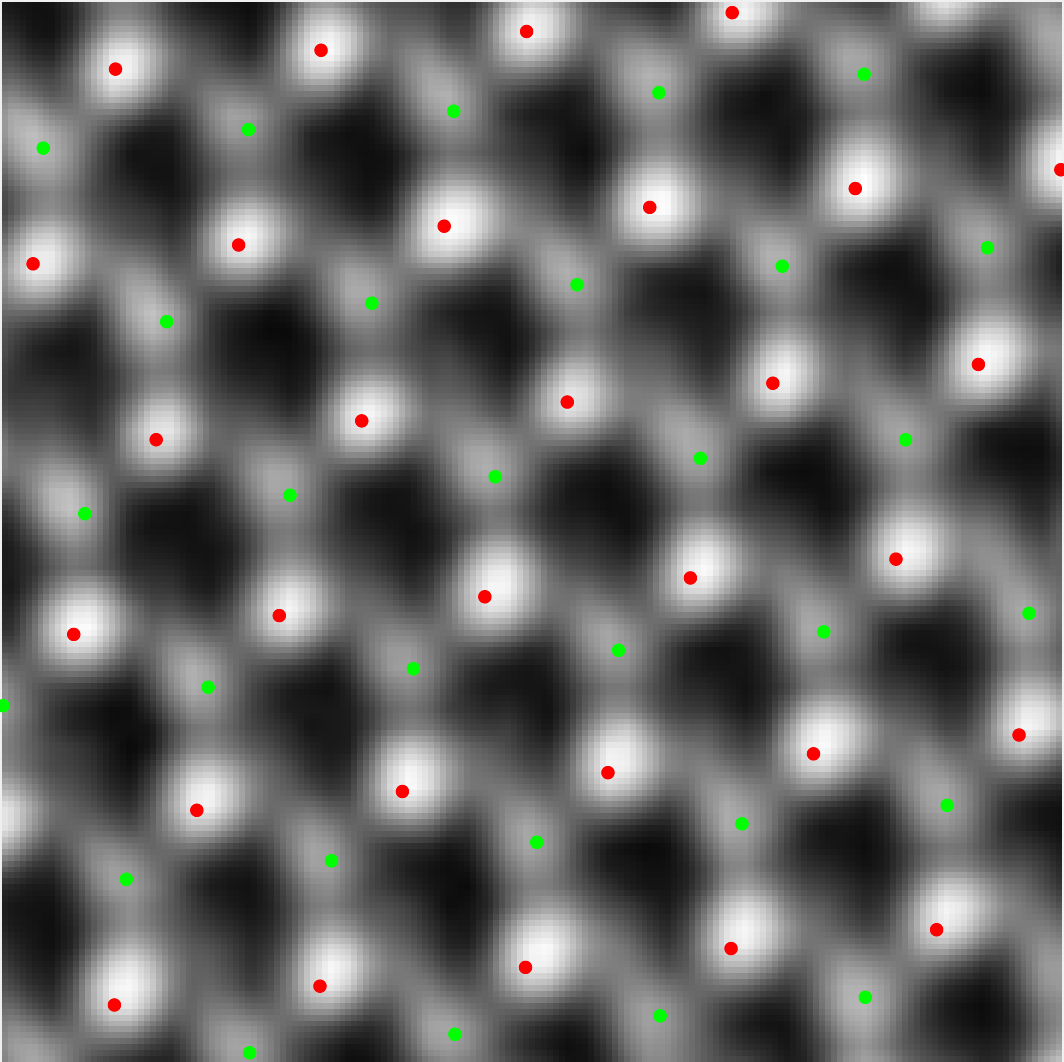}&
\includegraphics[height = 1.3in]{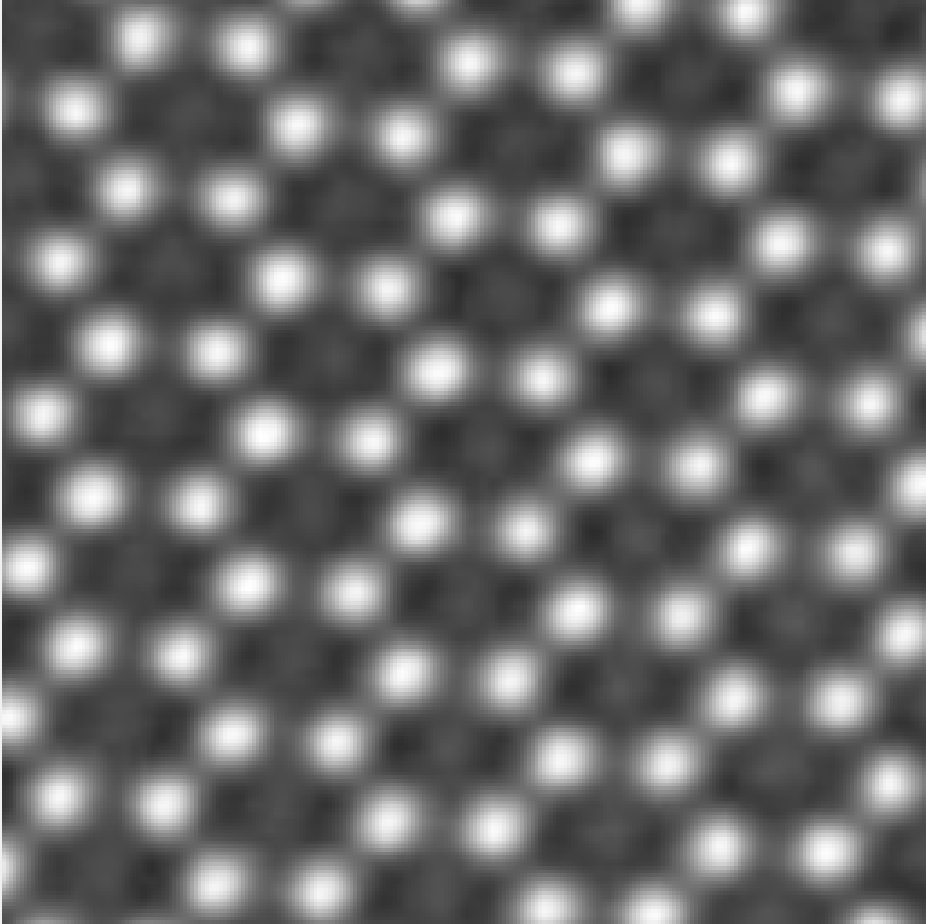} &
\includegraphics[height = 1.3in]{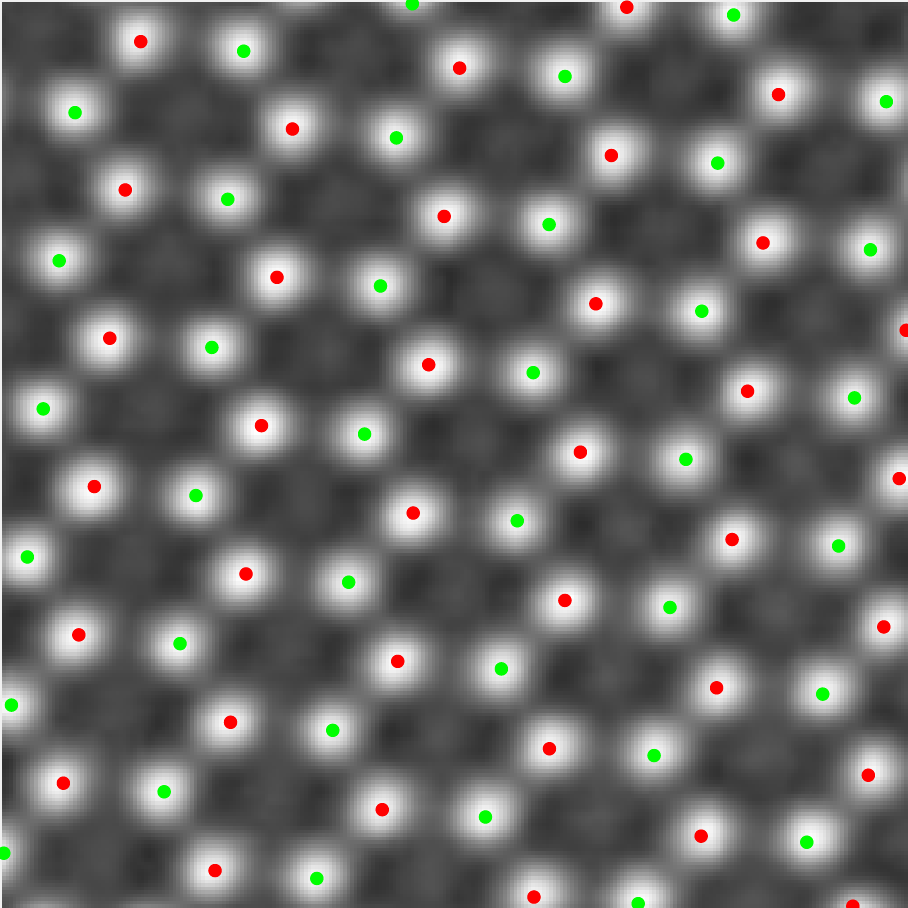}
\end{tabular}
\end{center}
  \caption{[LISA on real images] The underlying lattices in TMD monolayers (a) and (c) detected and separated by LISA are presented in (b) and (d) respectively, using different colors. (a) is cropped from~\cite{zheng2014high} Figure 3 (c), and (c) is cropped from~\cite{rao2013graphene} Figure 1 (c).}\label{realimage}
\end{figure}

We also test LISA on an important class of images in material sciences focusing on the grain boundaries. Patches of regular lattices are identified using LISA, then by directly comparing with the preprocessed image,  regions of homogeneous patterns are separated. Consequently, the grain boundary is revealed. Figure~\ref{gbmerge}(a) shows a part of an image from \cite{biro2013grain}, where a grain boundary is formed in the graphene grown by chemical vapor deposition (CVD). LISA detects two lattices as in (b) and (c). In (d), particles in (b) shared with the preprocessed image is colored green, those in (c) are red, and the white pixel indicates where (b) and (c) intersect. This example shows the potential applications of LISA besides superlattice separation. 
\begin{figure} 
 \begin{center}
\begin{tabular}{cccc}
(a)  & (b) & (c) & (d) \\
 \includegraphics[height=1.3in]{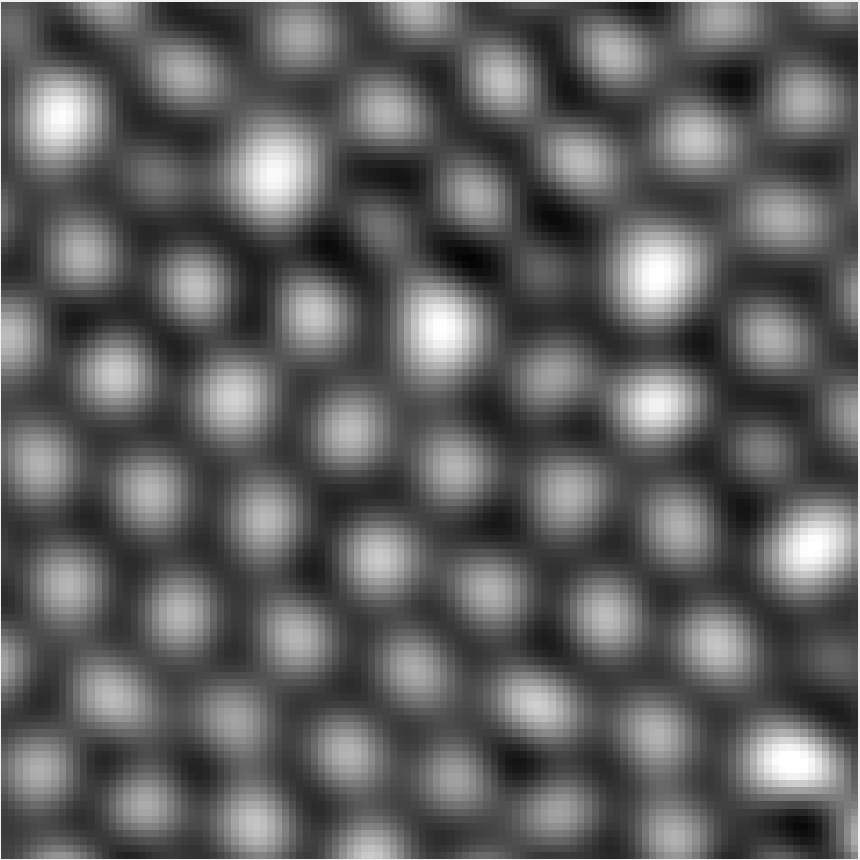} &
 \includegraphics[height=1.3in]{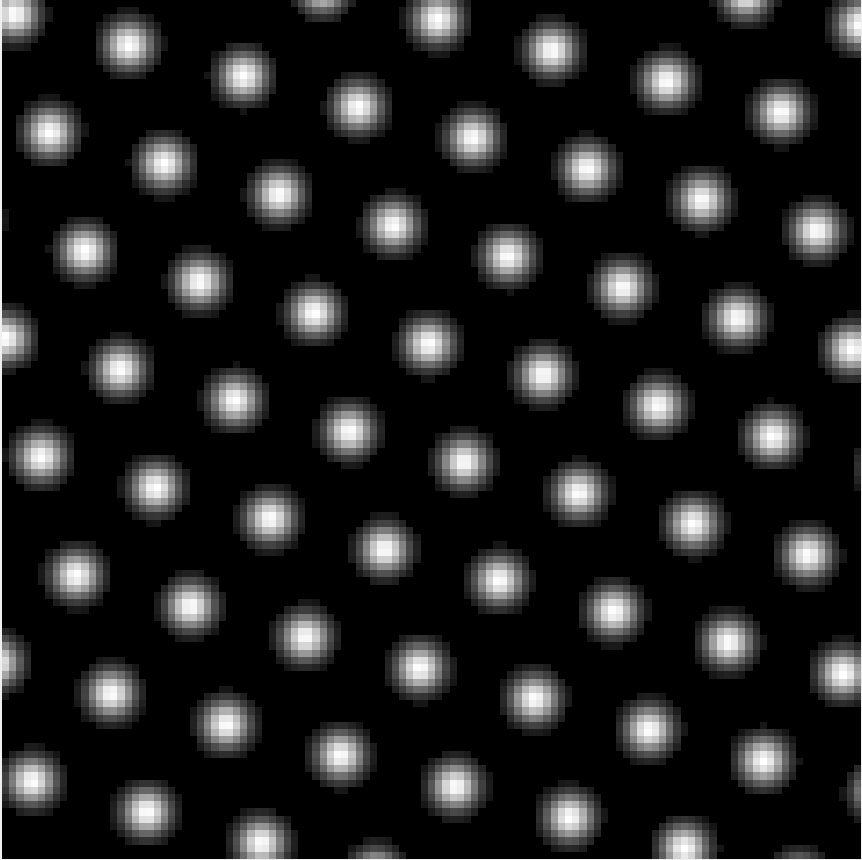} &
 \includegraphics[height=1.3in]{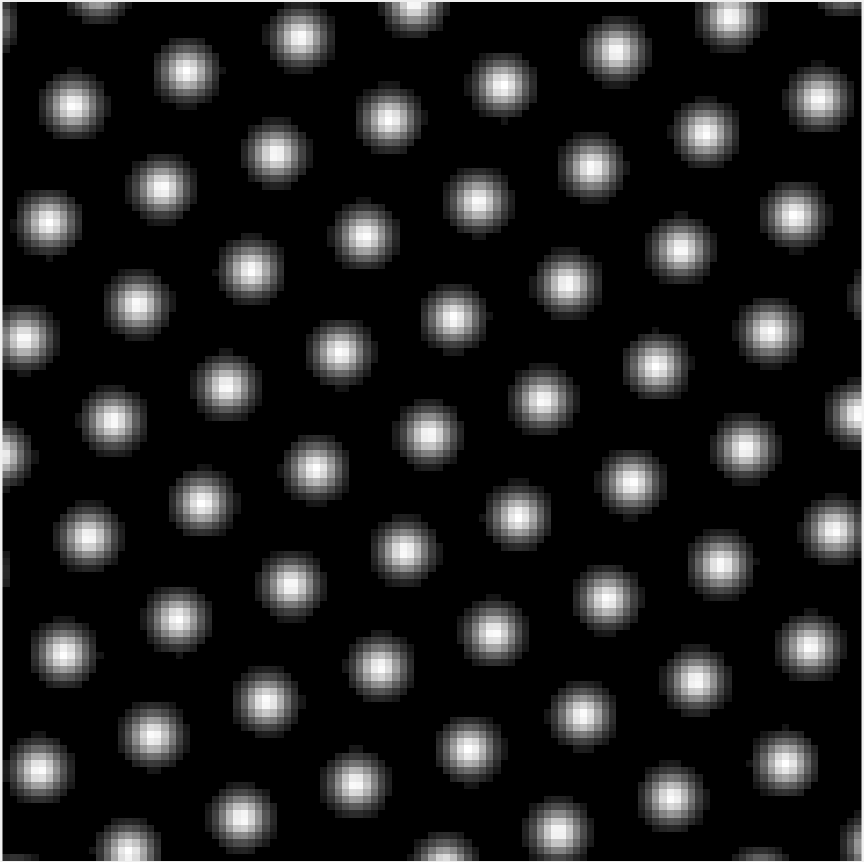} &
 \includegraphics[height=1.3in]{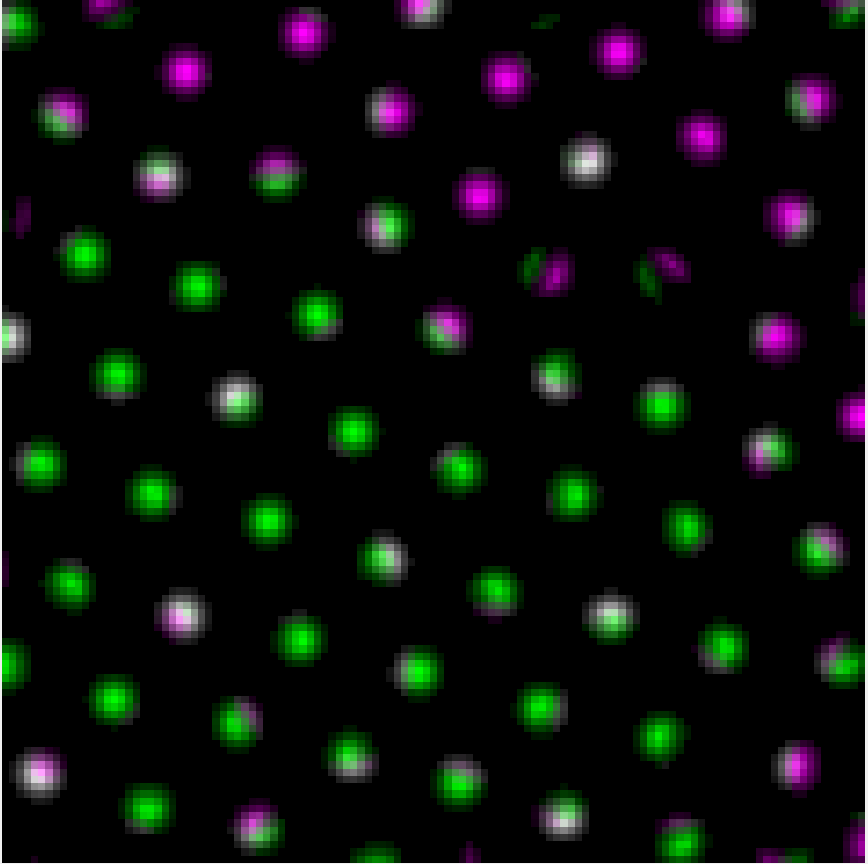}
 \end{tabular} 
 \end{center}   
  \caption{[LISA on grain boundary] A grain boundary in graphene formed by CVD is shown in (a). LISA detects the lattice in (b) $\mathcal{T}_{-1.3794+9.7510i}\Lambda\langle-10.9881 -12.1163i,-0.4579 + 0.8950i\rangle$ and in (c) $\mathcal{T}_{9.6287+9.5640i}\Lambda\langle-15.7326 - 4.7420i,0.4813 + 0.8800i\rangle$. (d) shows the homogeneous regions in different color, and the grain boundary is revealed. (a) is adjusted from \cite{biro2013grain} Figure 15 (a)}
  \label{gbmerge}
\end{figure}

Finally, we investigate the efficiency of LISA.  We focus on three major factors contributing to the runtime of LISA: image size, number of connected components on the spectrum surface ($J$ in Table~\ref{LISAtabel}), and the number of stabilizing iteration ($K$ in Table~\ref{LISAtabel}). Fixing a superlattice consisting of two lattices $\mathcal{T}_{0}\Lambda\langle 12,i\rangle$ and $\mathcal{T}_{0}\Lambda\langle 11.2763 + 4.1042i, e^{i4\pi/9}\rangle$, the CPU times (in second) of LISA are plotted against each one of these factors when the other two are controlled. The results show that the complexity of LISA depends linearly on $J$ and $K$, and quadratically on the image width. This is consistent with the analysis in \cite{splineRadon}, where the complexity of B-spline convolution-based Radon transform is proportional to the image size, i.e., image width times image length.
\begin{figure}
\begin{center}
\begin{tabular}{ccc}
(a) & (b) & (c) \\
\includegraphics[height = 1.5in]{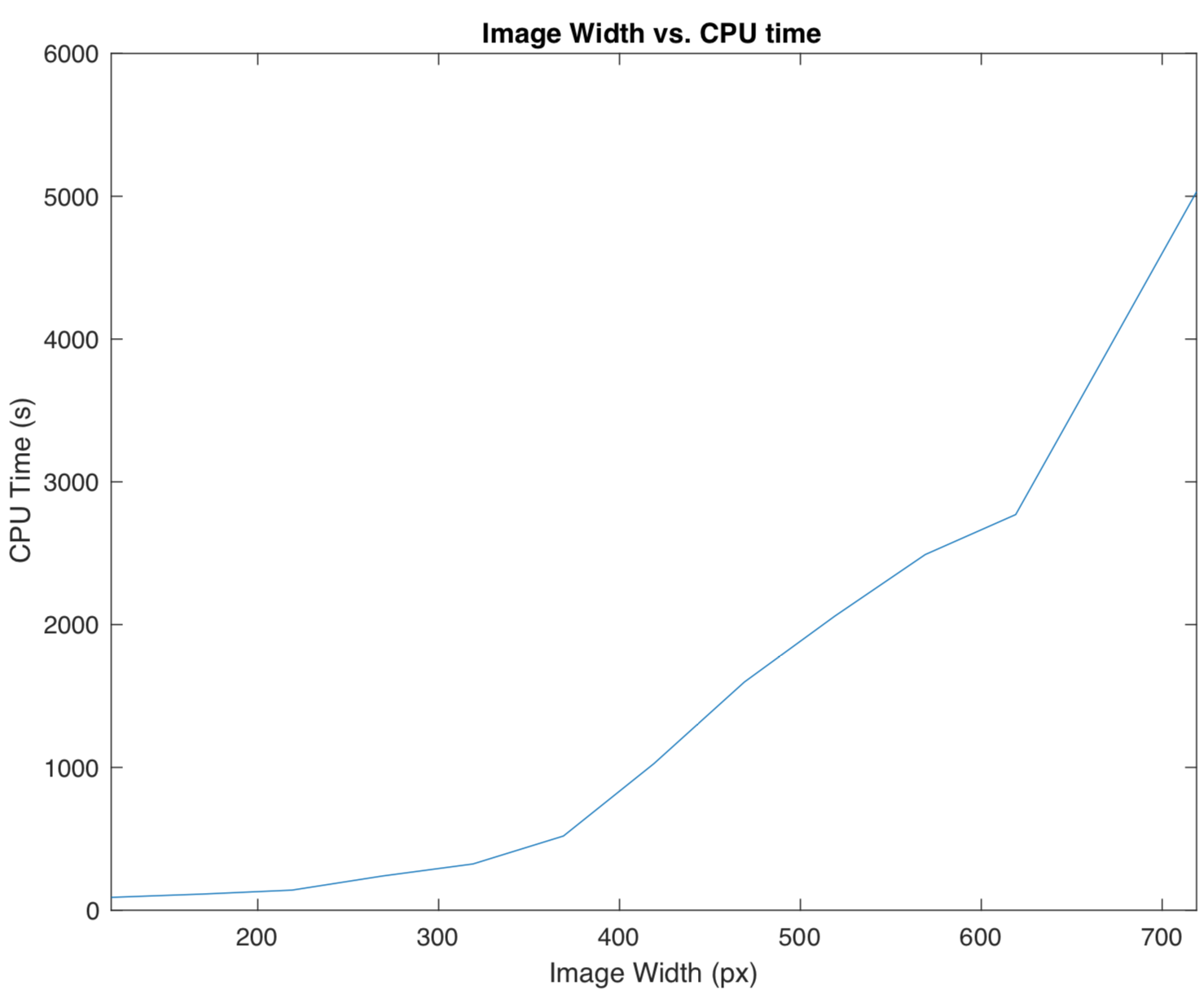} &
\includegraphics[height = 1.5in]{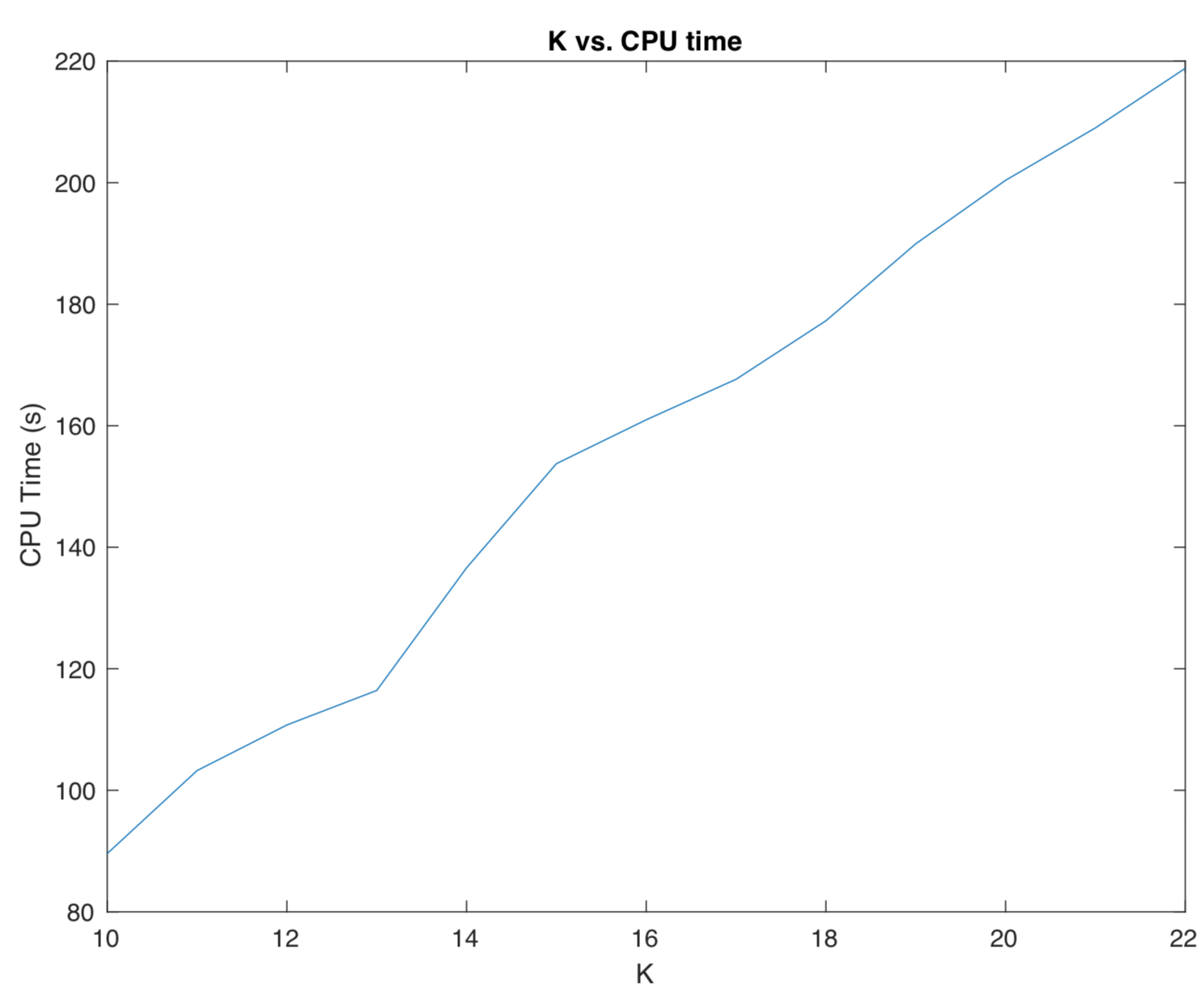}&
\includegraphics[height = 1.5in]{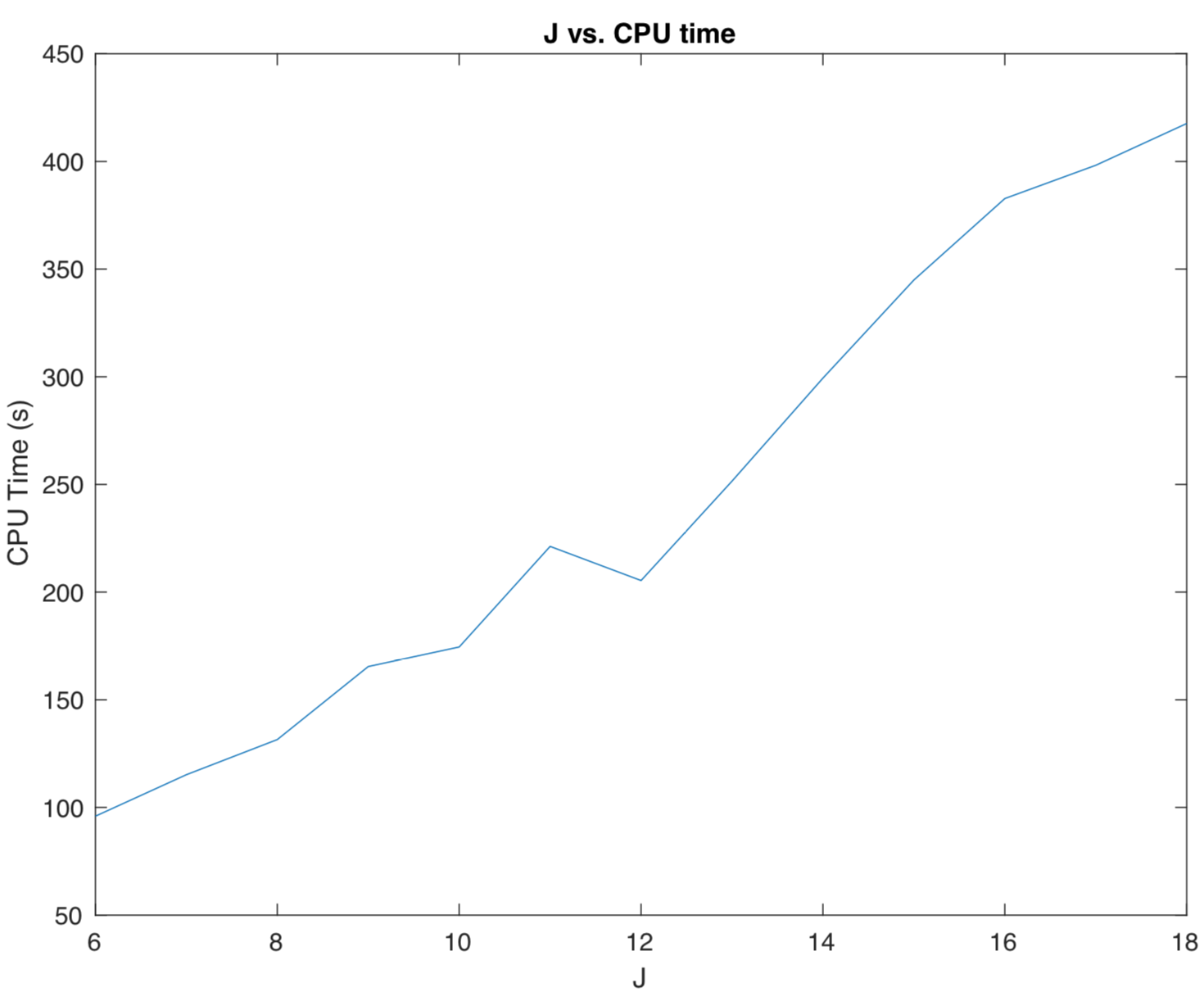}
\end{tabular}
\end{center}
\caption{[CPU time of LISA] Fixing two lattices, the base image width is $119$, $K=10$, $J=6$. Changing image size (a), the number of stabilizing iteration $K$ (b) or the number of connected components $J$ (c),  while keeping the other two as in the base case, the CPU times (in sec) of LISA are plotted respectively.}\label{runtimetest}
\end{figure}

\section{Summary} \label{sec:conc}

This paper addresses two questions of lattice identification and separation in superlattices. The first one is: \textit{What is a proper space where any two lattices can be compared quantitatively?} Starting from the positive minimal bases, we exploit the modular group theory and Poincar\'{e} metric to define a lattice space $\mathscr{L}$ with a natural metric structure. This new definition provides rich geometrical intuition for the collection of lattices. Computation of the metric $d_{\mathscr{L}}$ yields compact and visually consistent information about differences between lattice patterns. Compatible with wallpaper group theory, $\mathscr{L}$ provides finer classifications, which is more suitable for the purpose of measurement.\par
 The second question is: \textit{How to practically identify and separate lattices from a superlattice?} We introduce the algorithm LISA. Without prior knowledge of the lattices and number of layers of superposed lattices, LISA sequentially identifies and extracts lattice patterns until the remainder has insufficient intensity. We show the importance of density restriction when evaluating lattice candidates, which are indicated by  pairs of high responses on the spectrum surface. This evaluation method renders LISA's robustness against moir\'{e} patterns. An analytical framework is presented to consider the effects of relative translations and Gaussian perturbation. The metric space $(\mathscr{L},d_{\mathscr{L}})$ allows more discussion about special families of lattices, and its geometrical properties are interesting to explore. LISA produces a series of regular lattice patterns, which can be extended to the separation of near-regular lattices, and identification of grain boundaries.

\begin{appendices}

\section{Sub-lattices and Parent-lattices in Lattice Space} \label{A:subL}

In section~\ref{sec:lattice}, we regard the collection of M\"{o}bius transforms as a group, and exploit its subgroup, the modular group $\Gamma$, to address the problems of basis representation. More generally, the group of M\"{o}bius transforms has a monoid structure, i.e., inverse elements are not required compared with the group definition. 
In this section, we explore further the value of M\"{o}bius transforms by investigating one of its submonoids, $M_{2}(\mathbb{Z})$. We present the close relation between sub-lattices of a lattice and the monoid $M_{2}(\mathbb{Z})$. A one-to-one correspondence between sub-lattices and parent-lattices of a lattice is then proved to extend this relation to that between parent-lattices and $M_{2}(\mathbb{Z})$.

This section also has a practical significance. In section~\ref{sec:LISA}, when evaluating lattice candidates, the confusion caused by inhomogeneous texton regions and moir\'{e} effect is eliminated by attaching a density restriction in (\ref{eval}).  The necessity of the density term is theoretically driven from the framework of this section.

\subsection{Sub-lattices and Parent-lattices using Descriptors}  The classical notions of sub- and parent-lattices can be paraphrased using descriptors $\beta$ and $\rho$.
\begin{Def}[Sub-lattice]\label{sublatticedef}
	Let $\Lambda=\Lambda\langle\beta,\rho\rangle$ and $\Lambda'=\Lambda\langle\beta',\rho'\rangle$ be two lattices. We say that $\Lambda'$ is a sub-lattice of $\Lambda$, if there exists $(k_{1},k_{2},k_{3},k_{4})\in\mathbb{Z}^{4}$ with $k_{1}k_{4}-k_{2}k_{3}>0$ such that:
	\begin{align*}
	\begin{cases}
		\beta'=\beta(k_{1}+k_{2}\rho)\\
		\rho' = (k_{3}+k_{4}\rho)/(k_{1}+k_{2}\rho)
	\end{cases}.
	\end{align*}
And $\Lambda\langle\beta',\rho'\rangle$ is said to be a sub-lattice of $\Lambda\langle\beta,\rho\rangle$ induced by $(k_{1},k_{2},k_{3},k_{4})$.
\end{Def}
This definition is derived from the equivalent expression:
\begin{align*}
	\begin{cases}
		\beta'=k_{1}\beta+k_{2}\beta\rho\\
		\beta'\rho'=k_{3}\beta+k_{4}\beta\rho
	\end{cases}
\end{align*}
with $k_{1}k_{4}-k_{2}k_{3}>0$, which says that the basis for a sub-lattice comes from a non-degenerate linear combination of the basis of the original lattice using integer coefficients. The transformations\begin{align*}
	z\mapsto\frac{k_{3}+k_{4}z}{k_{1}+k_{2}z},~\{k_{i}\}_{i=1}^{4}\subset\mathbb{Z}, \text{such that~}k_{1}k_{4}-k_{2}k_{3}>0,~\text{for any} z\in\mathcal{H}.
\end{align*}
form a monoid with function composition, which is denoted by $M_{2}(\mathbb{Z})$. In the category of monoids, $\text{PSL}_{2}(\mathbb{Z})\leq M_{2}(\mathbb{Z})\leq\text{PGL}_{2}(\mathbb{Z})$, hence the discussion here is a generalization of section~\ref{sec:lattice}. Symmetrically, we define parent-lattices as follows:
\begin{Def}[Parent-lattice]\label{parentlatticedef}
Let $\Lambda=\Lambda\langle\beta,\rho\rangle$ and $\Lambda'=\Lambda\langle\beta',\rho'\rangle$ be two lattices. We say that $\Lambda'$ is a parent-lattice of $\Lambda$, if there exists $(k_{1},k_{2},k_{3},k_{4})\in\mathbb{Z}^{4}$ with $v:=1/(k_{1}k_{4}-k_{2}k_{3})>0$ such that:
	\begin{align*}
	\begin{cases}
		\beta'=v\beta(k_{1}+k_{2}\rho)\\
		\rho' = (k_{3}+k_{4}\rho)/(k_{1}+k_{2}\rho)
	\end{cases}.
	\end{align*}
$\Lambda'$ is said to be a parent-lattice of $\Lambda$ induced by $(k_{1},k_{2},k_{3},k_{4})$.	
\end{Def}
This definition is equivalent to the normal notion of parent-lattice. Suppose $\Lambda'$ is a parent-lattice of $\Lambda$ in the normal sense, then there exist real numbers $a,b,c,d$ with $u:=ad-cb>0$ such that:
\begin{align}
\begin{cases}
\beta'=a\beta+b\beta\rho\\
\beta'\rho'=c\beta+d\beta\rho	
\end{cases}\iff \begin{cases}
\beta=d\beta'/u-b\beta'\rho'/u\\
\beta\rho=-c\beta'/u+a\beta'\rho'/u	
\end{cases}\label{eq2}.
\end{align}
Since symmetrically $\Lambda$ is a sub-lattice of $\Lambda'$, we have that $k_{4}=d/u,k_{2}=b/u,k_{3}=c/u$ and $k_{1}=a/u$ are actually integers.	Therefore,  (\ref{eq2}) becomes:
\begin{align*}
\begin{cases}
\beta'=\beta(k_{1}+k_{2}\rho)u\\
\rho'=(k_{3}+k_{4}\rho)/(k_{1}+k_{2}\rho)	
\end{cases}.
\end{align*}
But also notice that $k_{1}k_{4}-k_{2}k_{3}=(ad-cb)/u^{2}=1/u$, so we prove our claim. After checking the equivalence relations among descriptors, we have the following proposition,
\begin{Prop}
Given a lattice $\Lambda=\Lambda\langle\beta,\rho\rangle$, 	we have the following one-to-one correspondence:
\begin{align}
\big\{\text{Sub-lattices of }\Lambda\big\}\overset{\varphi}{\longleftrightarrow}	 \big\{\text{Parent-lattices of }\Lambda\big\}\label{correspondence}
\end{align}
via $\varphi$ well-defined as follows: if $(k_{1},k_{2},k_{3},k_{4})$ determines a sub-lattice via~(\ref{sublatticedef}), then they determines a parent-lattice via~(\ref{parentlatticedef}).
\end{Prop}
An example of this correspondence is in Figure~\ref{sublatDemo}.  Once we find all the sub-lattices of a lattice, all its parent-lattices come for free.  For a given lattice $\Lambda\langle\beta,\rho\rangle$, we only need to focus on finding $(k_{1},k_{2},k_{3},k_{4})\in\mathbb{Z}^{4}$ such that $(k_{3}+k_{4}\rho)/(k_{1}+k_{2}\rho)$ remains in $\mathcal{P}$.

\begin{figure}
\begin{center}
\begin{tabular}{ccc}
(a) & (b)  & (c) \\
    \includegraphics[height=1.3in]{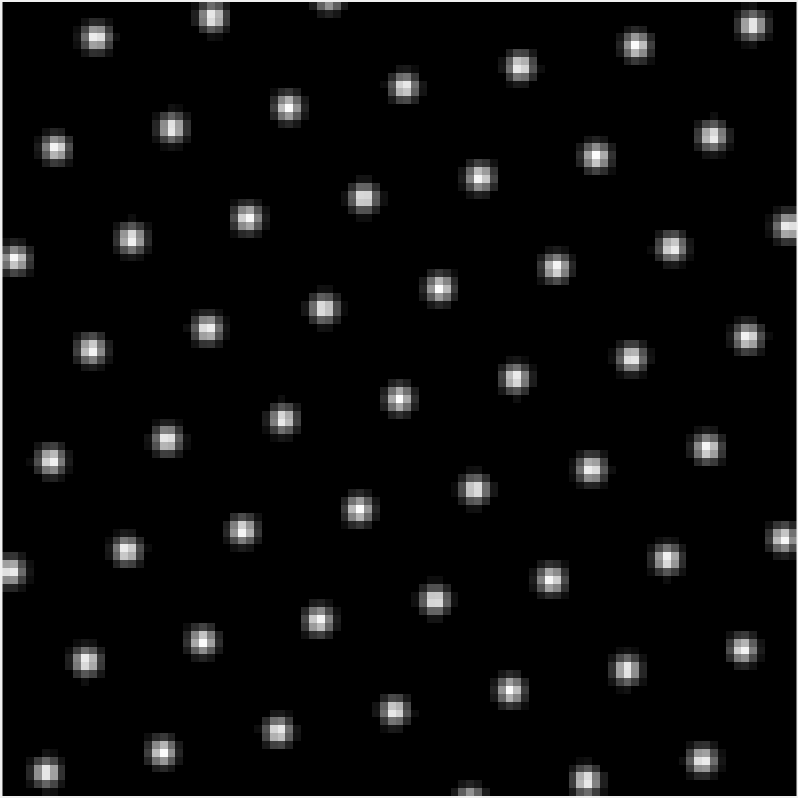} &
    \includegraphics[height=1.3in]{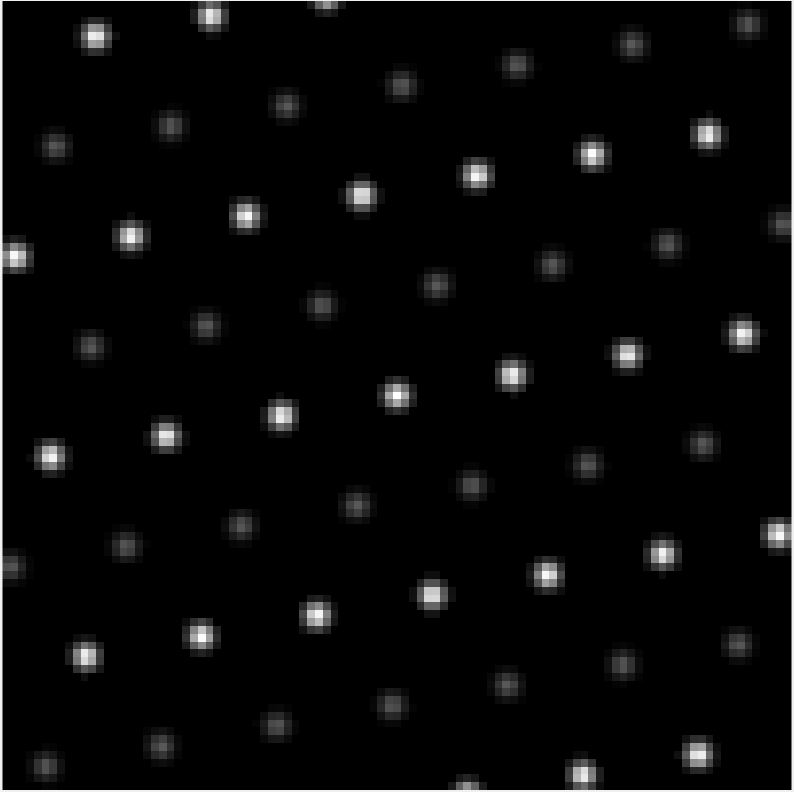} &
    \includegraphics[height=1.3in]{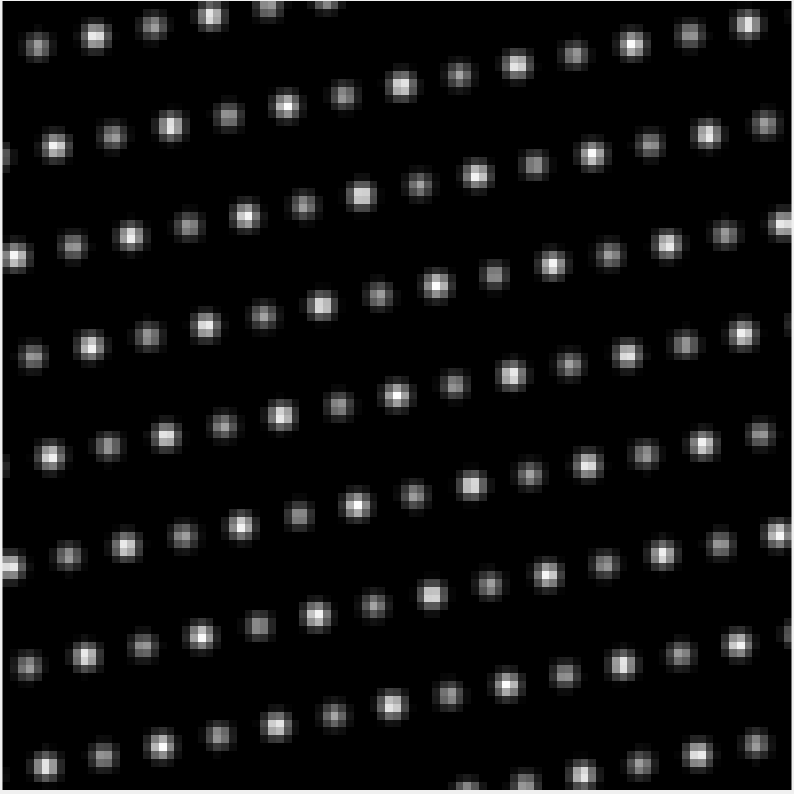} 
 \end{tabular}
 \end{center}
  \caption{[One-to-one correspondence between sub- and parent-lattice] (a) is a lattice $\Lambda\langle\beta,\rho\rangle=\Lambda\langle14.7721 + 2.6047i,e^{i\pi/3}\rangle$. $\Lambda\langle\beta,2\rho+1\rangle$ gives a sub-lattice in (b) and it corresponds to a parent-lattice $\Lambda\langle\beta/2,2\rho+1\rangle$. }\label{sublatDemo}
\end{figure}

\subsection{$M_{2}(\mathbb{Z})$-actions and Sub-lattices} 
As suggested above, we consider an element in $M_{2}(\mathbb{Z})$ 
whose image of $\mathcal{P}$ has non-empty intersection with $\mathcal{P}$. It suffices to see its action on three distinct points. Generally, given a $\rho\in\mathcal{P}$, finding all such possible integer coefficients is hard and not fruitful. There are two most easily found families of sub-lattices:
\begin{enumerate}
\item Corresponding to $(k_{1},k_{2},k_{3},k_{4})=(n,0,0,1)$, if $[\beta,\rho]\in\mathcal{L}$ with $|\rho|\geq n$, $n\in\mathbb{N}_{\geq 1}$, then for any $m\leq n$, $[m\beta,\rho/m]$ is a sub-lattice of $[\beta,\rho]$.
\item Corresponding to $(k_{1},k_{2},k_{3},k_{4})=(1,0,0,n)$, if $[\beta,\rho]\in\mathcal{L}$ with $|\text{Re}(\rho)|\leq 1/(2n)$, $n\in\mathbb{N}_{\geq 1}$, then for any $m\leq n$, $[\beta,m\rho]$ is a sub-lattice of $[\beta,\rho]$.
\end{enumerate}
In some special cases, it is also easy to find conditions for $(k_{1},k_{2},k_{3},k_{4})\in\mathbb{Z}^{4}$ such that there exists some $\rho\in\mathcal{P}$ who can be mapped to $\mathcal{P}$. For example, when $k_{2}=0$  (which forces $k_{1}\neq 0$), we find $\infty\to\infty$, $0\mapsto k_{3}/k_{1}$ and $\pm1/2\mapsto (\pm k_{4}/2+k_{3})/k_{1}$ by the $M_{2}(\mathbb{Z})$-action determined by this $(k_1,k_2,k_3,k_4)$. In order to have non-empty intersection with $\mathcal{P}$, we must require:
\begin{align*}
&\min\{(\pm k_{4}/2+k_{3})/k_{1}\}\leq 1/2 \quad \text{or}\\
& k_{3}^{2}-k_{4}^{2}/4<0~\text{and}~\min\{(\pm k_{4}/2+k_{3})/k_{1}\}\geq 1/2.
\end{align*} 
By the correspondence in (\ref{correspondence}), these results also extend symmetrically to parent-lattices. 

\section{Psudo-code for computing $d_{\mathscr{L}}$}  \label{A:code}
The definition of $d_{\mathscr{L}}$~(\ref{dLdef}) requires multiple comparisons. For paths passing through $E = \{(\beta,\rho) \mid \beta\in\mathcal{K}, |\rho|=1, \rho\in\mathcal{P}\}$, minimizations are involved.

\noindent\textbf{Inputs}: two lattice bases $(b_{1},b_{2})$ and $(b_{1}',b_{2}')\in\mathbb{C}^{2}$.\\[5pt]
\textbf{Step 1.} Transfer to descriptors: $\beta\gets b_{1}$, $\beta'\gets b_{1}'$, $\rho\gets b_{2}/b_{1}$, and $\rho'\gets b'_{2}/b'_{1}$.\\[5pt]
\textbf{Step 2.} Define two sub-routines:
\begin{align*}
D_{\mathcal{K}}&: (x,y)\in\mathbb{C}^{2}\mapsto\sqrt{w(|x|-|y|)^{2}+(1-w)(\cos^{-1}\frac{\text{Re}(x\overline{y})}{|x||y|})^{2}},~(w=0.05)\\
D_{\mathcal{P}}&: (x,y)\in\mathbb{C}^{2}\mapsto 2\ln\frac{|x-y|+|x-\overline{y}|}{2\sqrt{\text{Im}(x)\text{Im}(y)}}.	
\end{align*}
Extend them to two new sub-routines by:
\begin{align*}
d_{\mathcal{K}}&:(x,y)\in\mathbb{C}^{2}\mapsto\min\{D_{\mathcal{K}}(x,y), D_{\mathcal{K}}(-x,y)\}.\\
d_{\mathcal{P}}&:(x,y)\in\mathbb{C}^{2}\mapsto\min\{D_{\mathcal{P}}(x,y),D_{\mathcal{P}}(x-1,y),D_{\mathcal{P}}(x+1,y)\}.	
\end{align*}
Then define:
\begin{align*}
D:(x,y,z,w)\in\mathbb{C}^{2}\mapsto \sqrt{d_{\mathcal{K}}(x,z)^{2}+d_{\mathcal{P}}(y,w)^{2}}	
\end{align*}
\textbf{Step 3.} Fix an integer $N$.\\[5pt]
For $j=0,1,\cdots,N$:\\
\indent For $k=0,1,\cdots,N$:\\
\indent\indent $D_{j,k}\gets d(\beta,\rho',\beta,\rho')$;\\
\indent\indent $D_{j,k}\gets \min\{D_{j,k},d(\beta,\rho,\beta',e^{i(\pi/3+k\pi/3)})+d(\beta',e^{i(\pi/3+k\pi/3)},\beta',\rho')\}$;\\
\indent\indent $D_{j,k}\gets \min\{D_{j,k},d(\beta,\rho,e^{i(\pi/3+k\pi/3)}\beta',-1/e^{i(\pi/3+k\pi/3)})+d(\beta',e^{i(\pi/3+k\pi/3)},\beta',\rho')\}$;\\
\indent\indent $D_{j,k}\gets \min\{D_{j,k},d(\beta,\rho,\beta,e^{i(\pi/3+j\pi/3)})+d(\beta,e^{i(\pi/3+j\pi/3)},\beta',\rho')\}$;\\
\indent\indent $D_{j,k}\gets \min\{D_{j,k},d(\beta,\rho,\beta,e^{i(\pi/3+j\pi/3)})+d(\beta,e^{i(\pi/3+j\pi/3)},\beta',e^{i(\pi/3+k\pi/3)})+d(\beta',e^{i(\pi/3+k\pi/3)},\beta',\rho')\}$;\\
\indent\indent $D_{j,k}\gets \min\{D_{j,k},d(\beta,\rho,\beta,e^{i(\pi/3+j\pi/3)})+d(\beta,e^{i(\pi/3+j\pi/3)},e^{i(\pi/3+k\pi/3)}\beta',-1/e^{i(\pi/3+k\pi/3)})+d(e^{i(\pi/3+k\pi/3)}\beta',-1/e^{i(\pi/3+k\pi/3)},\beta',\rho')\}$;\\
\indent\indent $D_{j,k}\gets \min\{D_{j,k},d(\beta,\rho,\beta,e^{i(\pi/3+j\pi/3)})+d(e^{i(\pi/3+j\pi/3)}\beta,-1/e^{i(\pi/3+j\pi/3)},\beta',\rho')\}$;\\
\indent\indent $D_{j,k}\gets \min\{D_{j,k},d(\beta,\rho,\beta,e^{i(\pi/3+j\pi/3)})+d(e^{i(\pi/3+j\pi/3)}\beta,-1/e^{i(\pi/3+j\pi/3)},\beta',e^{i(\pi/3+j\pi/3)})+d(\beta',e^{i(\pi/3+j\pi/3)},\beta',\rho')\}$;\\
\indent\indent $D_{j,k}\gets \min\{D_{j,k},d(\beta,\rho,\beta,e^{i(\pi/3+j\pi/3)})+d(e^{i(\pi/3+j\pi/3)}\beta,-1/e^{i(\pi/3+j\pi/3)},...$\\
\indent\indent\indent$e^{i(\pi/3+j\pi/3)}\beta',-1/e^{i(\pi/3+j\pi/3)})+d(\beta',e^{i(\pi/3+j\pi/3)},\beta',\rho')\}$;\\
\indent End For\\
End For\\[5pt]
$d_{\mathscr{L}}((\beta,\rho),(\beta',\rho'))\gets\min_{j,k}D_{j,k}$.
\end{appendices}


\end{document}